\documentclass[11pt,a4paper]{amsart}
\usepackage[foot]{amsaddr}
\usepackage{ifxetex}
\ifxetex
  \usepackage[no-math]{fontspec}
\else
\fi
\usepackage{amsmath}
\usepackage{amsfonts}
\usepackage{amssymb}
\usepackage{amsthm}
\usepackage{fullpage}
\usepackage{microtype}
\usepackage[libertine]{newtxmath}
\usepackage[tt=false]{libertine} %% tt is ugly 
\usepackage{caption}
\usepackage{bbm}
\usepackage{hyperref, color}
\hypersetup{colorlinks=true,citecolor=blue, linkcolor=blue, urlcolor=blue}
\usepackage[linesnumbered,boxed,ruled,vlined]{algorithm2e}
\usepackage{bm}
\usepackage{bbm}
\usepackage[numbers]{natbib}
\usepackage{xcolor}
\usepackage{enumerate} 
\usepackage{enumitem}
\usepackage{tabularx}
\usepackage{array}
\usepackage{cleveref}
\newcolumntype{L}[1]{>{\raggedright\arraybackslash}p{#1}}
\newcolumntype{C}[1]{>{\centering\arraybackslash}m{#1}}
\newcolumntype{R}[1]{>{\raggedleft\arraybackslash}p{#1}}

\usepackage{makecell}
\usepackage{footnote}
\usepackage{float}
\makesavenoteenv{tabular}

\renewcommand{\epsilon}{\varepsilon}

\newtheorem{theorem}{Theorem}[section]
\newtheorem{observation}[theorem]{Observation}

\newtheorem*{claim*}{Claim}

\newtheorem{lemma}[theorem]{Lemma}
\newtheorem{proposition}[theorem]{Proposition}

\theoremstyle{definition}

\newtheorem{remark}[theorem]{Remark}
\newtheorem*{remark*}{Remark}

\def\Pr{\mathop{\mathbf{Pr}}\nolimits}
\def\Ex{\mathop{\mathbf{E}}\nolimits}

\renewcommand{\emptyset}{\varnothing}

\newcommand{\norm}[1]{\left\Vert#1\right\Vert}
\newcommand{\abs}[1]{\left\vert#1\right\vert}
\newcommand{\set}[1]{\left\{#1\right\}}
\newcommand{\tuple}[1]{\left(#1\right)}

\newcommand{\tp}{\tuple}

\renewcommand{\P}{\ensuremath{\mathbf{P}}}

\newcommand{\defeq}{:=}

\newcommand{\wt}{\text{wt}}

\newcommand{\sgn}{\mathrm{sgn}}

\newcommand{\holant}{{\mathrm{Holant}}}

\newcommand{\Tmatrix}{{\left(\begin{smallmatrix}1&1\\1&-1\end{smallmatrix}\right)}}
\newcommand{\odd}{\text{odd}}
\newcommand{\DTV}[2]{d_{\mathrm{TV}}\left({#1},{#2}\right)}
\newcommand{\Gap}{\ensuremath{\mathfrak{Gap}}}

\def\*#1{\bm{#1}} % Use \*A for \mathbf{A}
\def\+#1{\mathcal{#1}} % Use \+A for \mathcal{A}
\def\-#1{\mathrm{#1}} % Use \-A for \mathrm{A}
\def\=#1{\mathbb{#1}} % Use \=A for \mathbb{A}

\usepackage[textsize=tiny]{todonotes}

\usepackage{xifthen}

\makeatletter
\def\prob#1#2#3{\goodbreak\begin{list}{}{\labelwidth\z@ \itemindent-\leftmargin
                        \itemsep\z@  \topsep6\p@\@plus6\p@
                        \let\makelabel\descriptionlabel}
                \item[\it Name]#1
               \item[\it Instance]                #2
                \item[\it Output]#3
                \end{list}}
\makeatother

\newbool{doubleblind}
\setbool{doubleblind}{false}

\title{Swendsen-Wang dynamics for the ferromagnetic Ising model with external fields}

\ifdoubleblind
  \author{Author(s)}
\else
\author{Weiming Feng, Heng Guo, Jiaheng Wang}

\address[Weiming Feng, Heng Guo, Jiaheng Wang]{School of Informatics, University of Edinburgh, Informatics Forum, Edinburgh, EH8 9AB, United Kingdom. \textnormal{E-mail: \url{wfeng@ed.ac.uk}, \url{hguo@inf.ed.ac.uk}, \url{jiaheng.wang@ed.ac.uk}}}
\fi

\newcommand{\Pup}{P^{\uparrow}}
\newcommand{\Pdown}{P^{\downarrow}}
\newcommand{\Pdownup}{P^{\vee}}
\newcommand{\Pupdown}{P^{\wedge}}

\newcommand{\inner}[2]{\langle #1,#2\rangle}

\newcommand{\Ent}[2]{\ensuremath{\textnormal{Ent}_{#1}\left(#2\right)}}
\newcommand{\KL}[2]{\ensuremath{D_{\textnormal{KL}}\left(#1\parallel#2\right)}}
\newcommand{\Var}[2]{\ensuremath{\textnormal{Var}_{#1}\left(#2\right)}}
\newcommand{\DF}[2]{\ensuremath{D_{f}\left(#1\parallel#2\right)}}
\newcommand{\chisq}[2]{\ensuremath{D_{\chi^2}\left(#1\parallel#2\right)}}

\newcommand{\mixingtime}[1]{\ensuremath{t_{\textnormal{mix}}(#1)}}

\begin{document}

\begin{abstract}

We study the sampling problem for the ferromagnetic Ising model with consistent external fields,
and in particular, Swendsen-Wang dynamics on this model.
We introduce a new grand model unifying two closely related models: the subgraph world and the random cluster model.
Through this new viewpoint, we show:
\begin{enumerate}
  \item polynomial mixing time bounds for Swendsen-Wang dynamics and (edge-flipping) Glauber dynamics of the random cluster model,
    generalising the bounds and simplifying the proofs for the no-field case by Guo and Jerrum (2018);
    %This yields a simpler and more general proof of Guo and Jerrum (2018), which only handled the case without fields.
  \item near linear mixing time for the two dynamics above if the maximum degree is bounded and all fields are (consistent and) bounded away from $1$.
\end{enumerate}
%In addition, Glauber dynamics can be made perfect via coupling from the past.
\end{abstract}

\maketitle

\section{Introduction}

The Ising model is a classical statistical physics model for ferromagnetism that had far-reaching impact in many areas.
In computer science / combinatorics terms, the model defines a weighted distribution over cuts of a graph.
To be more precise, let $G=(V,E)$ be a simple undirected graph.
For each edge $e \in E$, we have the local interaction strength $\beta_e \in \mathbb{R}_{> 0}$,
and for each vertex $v \in V$, we have the external magnetic field (namely vertex weight)  $\lambda_v \in \mathbb{R}_{>0}$.
An Ising model is specified by the tuple $(G;\*\beta,\*\lambda)$, where $\*\beta = (\beta_e)_{e \in E}$ and $\*\lambda = (\lambda_v)_{v \in V}$.
We assign spins $\{0,1\}$ to the vertices $V$.
For each spin configuration $\sigma \in \{0,1\}^V$, the \emph{weight} of $\sigma$ is defined by
\begin{align*}
\wt_{\text{Ising}}(\sigma) \defeq \prod_{e=(u,v)\in E}\beta_e^{\mathbb{I}[\sigma(u)=\sigma(v)]}\prod_{u\in V}\lambda^{\sigma (u)}_u,
\end{align*}
where $\mathbb{I}[\sigma(u)=\sigma(v)]$ is the indicator variable of the event $\sigma(u)= \sigma(v)$.
The \emph{Gibbs distribution} $\pi_{\text{Ising}}$ is defined by
\begin{align}\label{eqn:Ising}
\forall \sigma \in \{0,1\}^V,\quad \pi_{\text{Ising}}(\sigma) = \frac{\wt_{\text{Ising}}(\sigma)}{Z_{\text{Ising}}},
\end{align}
where
\begin{align*}
Z_{\text{Ising}} = Z_{\text{Ising}}(G;\*\beta,\*\lambda) \defeq \sum_{\tau \in \{0,1\}^V } \wt_{\text{Ising}}(\tau)	
\end{align*}
is the \emph{partition function}.
In this paper we focus on the \emph{ferromagnetic} case, where $\beta_e > 1$ for all $e \in E$,
with \emph{consistent} fields, where $\lambda_v\in(0,1]$ for all $v\in V$.
Note that by flipping the spins, the last assumption is equivalent to assuming $\lambda_v\in[1,\infty)$ for all $v\in V$.

There is extensive computational interest in simulating the Ising model and in evaluating various quantities related to it.
A major contribution in the rigorous algorithmic study of the model is the Jerrum-Sinclair algorithm \cite{JS93}, 
which is the first \emph{fully polynomial-time randomised approximation scheme} (FPRAS) for the partition function $Z_{\text{Ising}}$ of the ferromagnetic Ising model with consistent fields on any graph.
The main ingredient of their algorithm is to show that a natural Markov chain mixes in polynomial-time to sample from the so-called ``subgraph-world'' model,
which has the same partition function up to some easy to compute factors.

Usually, using self-reducibility, approximately evaluating the partition function is computationally inter-reducible to approximate sampling \cite{JVV86}.
However, in the case of the Ising model,
the original algorithm by Jerrum and Sinclair does not directly yield a sampling algorithm for spin configurations.
This is because inconsistent fields may be created during the self-reduction, making the algorithm no longer applicable.
To circumvent this issue, Randall and Wilson \cite{RW99} showed that when there is no external field, 
an efficient approximate sampler for spin configurations exists by doing self-reductions in the so-called random cluster model.
This is a model introduced by Fortuin and Kasteleyn \cite{FK72} and also has the same partition function as the previous two models up to some easy to compute factors.\footnote{The random cluster model has a parameter $q>0$. The Ising model corresponds to the case of $q=2$.}

On the other hand, a different Markov chain introduced by Swendsen and Wang \cite{SW87} has shown great performance on sampling Ising configurations in practice.
This dynamics is best understood via the Edwards-Sokal distribution \cite{ES88}, which is a joint distribution on both edges and vertices.
The marginal distribution on vertices is the Ising model, and the marginal distribution on edges is the random cluster model.
Sokal and later Peres\footnote{Peres further conjectured that the sharp mixing time bound is $O(|V|^{1/4})$.} conjectured that the Swendsen-Wang (SW) dynamics mixes in polynomial-time for ferromagnetic Ising models,
and this was resolved in affirmative by Guo and Jerrum \cite{GJ18}.
They showed that the edge-flipping dynamics for the random cluster model mixes in polynomial-time, 
and this dynamics is known to be no faster than the SW dynamics \cite{Ull14}.
Another consequence of \cite{GJ18} is that there is a perfect sampler for the ferromagnetic Ising model and the corresponding random cluster model, improving upon the approximate sampler of \cite{RW99}.
This is done via monotone coupling from the past (CFTP) \cite{propp1996exact} as the random cluster model is monotone.

One restriction of \cite{GJ18} is that their result only applies to the ferromagnetic Ising model without external fields.
The original random cluster formulation of \cite{FK72} does not incorporate external fields,
although it is not hard to do so by generalising to a weighted random cluster formulation.
Indeed, Park, Jang, Galanis, Shin, {\v{S}}tefankovi{\v{c}}, and Vigoda~\cite{PJGSSV17} generalised the SW dynamics $P_{\mathrm{SW}}^{\mathrm{Ising}}$ (see \Cref{sec:SW} for detailed description) in the presence of external fields.
They also showed efficiency of this algorithm in certain parameter regimes and on random graphs.
This left open the question if the generalised SW dynamics is efficient in general.

To start stating our results, let us first define the \emph{mixing time} of Markov chains,
which measures the convergence rate and efficiency of Markov chain based algorithms.
Let $P$ be a Markov chain whose stationary distribution is $\pi$ over the state space $\Omega$. 
The \emph{mixing time} of $P$ is defined by
%\begin{align*}
%  \forall 0 < \epsilon < 1,\quad T_{\-{mix}}(P,\epsilon) = \max_{X_0 \in \Omega} \min\set{t \mid \DTV{P^t(X_0, \cdot)}{\pi} \leq \epsilon},
%\end{align*}
\begin{align*}
  \forall 0 < \epsilon < 1,\quad T_{\-{mix}}(P,\epsilon) = \max_{X_0 \in \Omega} \min\set{t \mid \DTV{P^t(X_0, \cdot)}{\pi} \leq \epsilon},
\end{align*}
where $\DTV{\mu}{\pi} = \frac{1}{2}\sum_{\sigma \in \Omega}|\mu(\sigma)-\pi(\sigma)|$ is the \emph{total variation distance} between two distributions.

First, we show that the edge-flipping dynamics for the weighted random cluster model mixes in polynomial-time.
By adapting \cite{Ull14} to the case with fields, this implies that the generalised SW dynamics has a polynomial running time for any ferromagnetic Ising model with consistent fields on any graph,
answering the question above.

\begin{theorem}\label{theorem-sw-main}
Let $1 < \beta_{\min} \leq  \beta_{\max}$ be constants.
For any ferromagnetic Ising model on graph $G=(V,E)$ with parameters $(\beta_e)_{e \in E}$ and $(\lambda_v)_{v \in V}$, where $\beta_{\min} \leq  \beta_e \leq \beta_{\max}$ and $0 < \lambda_v \leq 1$,
the mixing time of Swendsen-Wang dynamics is $O(N^4 m^2 \tp{m + \log \frac{1}{\epsilon})}$, where $N= \min\left\{n, \frac{1}{1 - \lambda_{\max}} \right\}$, $\lambda_{\max} = \max_{v \in V}\lambda_v$, $n = |V|$ and $m = |E|$.
\end{theorem}
Note that if $\beta_e=1$ for some $e\in E$, it is equivalent to remove such an edge.
Also if $\lambda_v=0$ for some $v\in V$, it is equivalent to pin $v$ to $0$ and then absorb $v$ into its neighbours external fields.
Thus, any ferromagnetic Ising model with consistent external fields can be transformed into one satisfying the condition of \Cref{theorem-sw-main}.
The big-$O$ notation hides a constant factor depending only on $ \beta_{\min}$ and $ \beta_{\max }$. See~\eqref{eq-constant-1} for the details of the hidden constant.

The main technical innovation of ours is to introduce a grand model, which incorporates both the so-called subgraph world \cite{JS93} and the random cluster model.
The subgraph world assigns weights to subsets of edges,
where each vertex of an odd degree in the induced graph suffers a penalty corresponding to its external field (or the lack thereof).
Detailed definitions of the basic models are given in \Cref{sec:models}.

The main inspiration of our grand model is the coupling given by Grimmett and Janson \cite{GJ07} between the two models above without external fields.
Our model assigns 3 states to each edge:~$0,1,2$.
A sample of our model can be generated as follows:
first, we sample a subset of edges from the subgraph world model, and assign $1$ to them;
then, we assign $0$ or $2$ to each remaining independently with a carefully chosen probability.
Detailed definitions are in \Cref{sec:grand-def}.
The marginal distribution of edges assigned $1$ clearly follow the subgraph world distribution,
and we show that the non-zero edges follow the weighted random cluster model (\Cref{lem:weighted-coupling}).
This last step is done using Valiant's holographic transformations \cite{Val08}.
It is also a generalisation of \cite{GJ07} in the presence of external fields.

We give a polynomial upper bound of the mixing time of the Glauber dynamics for the grand model in \Cref{section-var} via the method of canonical paths \cite{JS89}.
Our construction of the canonical path is a variation of the original paths by Jerrum and Sinclair \cite{JS93}.
The projection of this dynamics to the non-zero edges is exactly the Glauber dynamics for the weighted random cluster model.
We show that this project does not slow down the dynamics (\Cref{sec:grand-RC}), and therefore mixing time bounds for the weighted random cluster model is a direct consequence.
This implies \Cref{theorem-sw-main}.
When there is no field, our argument recovers the result of Guo and Jerrum \cite{GJ18}.
However, our argument is both simpler and more general.

Another important feature of the grand model is that it gives a Gibbs distribution, 
in the sense that variables are independent if we condition on a subset of edges which disconnect the graph.
This is a feature absent in the random cluster models.
Recently, there is a lot of progress in analysing the mixing time of dynamics for Gibbs distributions, especially using the notion of spectral independence~\cite{ALO20}.
Since the domain in our case is not Boolean, we use a generalisation of~\cite{FGYZ21} (see also~\cite{CGSV21} for a different generalisation).
An important development along this line is that in bounded degree graphs, spectral independence implies near-linear mixing time of dynamics for the Gibbs distribution \cite{CLV21}.
To be more precise, they showed a constant decay rate for the relative entropy in this setting.

Back to the Ising model, when the maximum degree is bounded and all external fields are bounded away from $1$,
Chen, Liu, and Vigoda \cite{CLV21a} established spectral independence for the subgraph world model.
Using our grand model, this implies spectral independence for the random cluster model as well.
However, since the random cluster model does not have conditional independence, the method of \cite{CLV21} does not apply.
Instead, we show spectral independence for the grand model in this setting.
Thus, via the method of \cite{CLV21} and exploiting the fact that the grand model is indeed a Gibbs distribution,
we obtain a constant decay rate for the relative entropy for the (edge-flipping) Glauber dynamics for the weighted random cluster model. (We apply the result of projecting chains in \Cref{sec:grand-RC} here again.)

However, this is still not quite enough to obtain desired mixing time bounds for the Swendsen-Wang dynamics.
The reason is that the aforementioned comparison techniques of \cite{Ull14} is an analysis of the eigenvalues of transition matrices, and thus it works only for spectral gaps but not for relative entropies.
For this last step, we introduce a new comparison argument for the decay rate of relative entropies between the (edge-flipping) Glauber dynamics and the Swendsen-Wang dynamics in \Cref{sec:RC-SW}.

To be more precise, we perform a careful analysis between the Glauber dynamics and the so-called ``single-bond'' dynamics introduced in \cite{Ull14}.
Our analysis utilises ideas from high-dimensional random walks~\cite{ALOV19,CGM19}.
For both the Glauber dynamics and the single-bond dynamics, we decompose them into two sub-steps: the down-walk and the up-walk.
Using our grand model, we bound the decay rate of relative entropy for the down-walk of Glauber dynamics.
By a new comparison argument, we show that the relative entropy also decays for the down-walk of ``single-bond'' dynamics with a slightly weaker rate.
Finally, we compare the down-walk of ``single-bond'' dynamics to the Swendsen-Wang dynamics via a simple application of the data processing inequality. 
Our analysis not only works for the decay of relative entropy, but also gives a very simple proof  (see \Cref{remark-simple-ull}) to the main result in \cite{Ull14}.

\begin{theorem}\label{theorem-nlogn}
Let $1 < \beta_{\min} \leq  \beta_{\max}, \Delta \geq 3$ and $0 < \delta < 1$ be constants.	
For any ferromagnetic Ising model on graph $G=(V,E)$ with parameters $(\beta_e)_{e \in E}$ and $(\lambda_v)_{v \in V}$, where $\beta_{\min} \leq  \beta_e \leq \beta_{\max}$, $0 < \lambda_v \leq 1 - \delta$ and the maximum degree of $G$ is at most $\Delta$,
the mixing time of Swendsen-Wang dynamics is $O(n \log \frac{n}{\epsilon})$, where $n = |V|$.
\end{theorem}

By the same reasoning below \Cref{theorem-sw-main}, we do not lose generality by assuming $\beta_{\min}>1$ and $\lambda_v>0$ in \Cref{theorem-nlogn}.
The big-$O$ notation hides a constant factor depending only on $ \beta_{\min}, \beta_{\max }, \delta$ and $\Delta$. See~\eqref{eq-constant-2} for the details of the hidden constant.

Comparing to \Cref{theorem-sw-main},
\Cref{theorem-nlogn} has a faster mixing time bound but comes with two further assumptions: constant degree bound and no trivial field.
It would be very interesting to relax either restriction.
Essentially, the bottleneck in \Cref{theorem-sw-main} comes from the overhead in the canonical path \cite{JS93} or multicommodity flow method \cite{Sin92} arguments.
Unfortunately, there does not seem to be any progress in improving the mixing time bound of these methods in the last three decades.
Instead, \Cref{theorem-nlogn} relies on recent progress of analysing spin systems via high-dimensional random walks \cite{CLV21,CLV21a}.
This method has very recently been generalised to bypass the bounded degree restriction \cite{AJKPV21,CFYZ22,CE22} in various models.
It is an interesting question whether this is also possible in the setting of \Cref{theorem-nlogn}.
To bypass the no trivial field restriction, we would need a new spectral independence bound,
for which there seems to be less progress.
In particular, it seems hard to explain the $\Theta(n^{1/4})$ mixing time on the complete graph without fields \cite{LNNP14} with spectral independence.

Previously, most studies on Swendsen-Wang dynamics focus on the case without fields (with the exception of \cite{PJGSSV17} discussed earlier).
%Since each step of SW dynamics updates all edges, it is often expected to be faster than Glauber dynamics by a factor of $\abs{V}$ or $\abs{E}$,
%and the ideal mixing time is $O(\log \abs{V})$.
%There is little progress 
Very sharp mixing time bounds have been obtained recently, either for special cases of graphs such as $\mathbb{Z}^d$ \cite{BCPSV21}, or in the tree uniqueness region for general graphs \cite{BCCPSV22}.
Our \Cref{theorem-nlogn} does not have these restrictions, but it only works with the presence of non-trivial external fields.
In the settings of \Cref{theorem-nlogn}, we conjecture that the sharp mixing time bound is $O(\log n)$, and the bottleneck is the comparison step.
Namely, instead of a ``no-slower'' comparison argument, the truth might be that SW dynamics is faster than Glauber dynamics by a factor $\abs{E}$.

Lastly, by applying the monotone CFTP \cite{propp1996exact}, we obtain perfect sampling versions of the (edge-flipping) Glauber dynamics in \Cref{sec:CFTP} for the weighted random cluster models.
Using that, we achieve perfectly sampling for the ferromagnetic Ising model with consistent external fields.

\begin{theorem}\label{theorem-perfect-ising}
Let $1 < \beta_{\min} \leq  \beta_{\max}$ be two constants.
There is a perfect sampling algorithm such that given any ferromagnetic Ising model on graph $G=(V,E)$ with parameters $(\beta_e)_{e \in E}$ and $(\lambda_v)_{v \in V}$, where $\beta_{\min} \leq  \beta_e \leq \beta_{\max}$ and $0 < \lambda_v < 1$, the algorithm returns a perfect sample in expected time $O(N^4  m^4 \log n)$, where $N= \min\left\{n, \frac{1}{1 - \lambda_{\max}} \right\}$ and $\lambda_{\max}=\max_{v \in V}\lambda_v$.

Furthermore, if $G$ has bounded maximum degree $\Delta = O(1)$ and there exists a constant $\delta > 0$ such that $\lambda_v \leq 1- \delta$ for all $v \in V$, the algorithm runs in expected time $O(n^2 \log^2 n )$.
\end{theorem}

We remark that the overhead in monotone CFTP is $O(\log \abs{V})$ and there is an extra factor $m=\abs{E}$ to implement each step of CFTP.
The hidden constants can be found in \eqref{eqn:perfect-sampler-constant}.

A natural question is if we can relax the assumptions on the parameters in \Cref{theorem-sw-main},~\ref{theorem-nlogn}, and~\ref{theorem-perfect-ising}.
For anti-ferromagnetic Ising models, the sampling problem (either approximate or perfect) has no polynomial-time algorithm unless $\textbf{NP}=\textbf{RP}$ \cite{JS93}.
Even restricted to the ferromagnetic case, 
Goldberg and Jerrum \cite{GJ07a} showed that the problem becomes \#BIS-equivalent when inconsistent fields are allowed,
where \#BIS stands for counting bipartite independent sets.
Its approximation complexity is a major open problem and is usually conjectured to have no polynomial-time algorithm.
Thus, it is unlikely to extend the range of parameters in in \Cref{theorem-sw-main},~\ref{theorem-nlogn}, and~\ref{theorem-perfect-ising}.

\section{Preliminaries}
\subsection{The models and their equivalences}
\label{sec:models}

Here we define the weighted random cluster model, and the subgraph-world model.
Then we give some equivalence results between them and the ferromagnetic Ising model.

%\subsubsection{Ferromagnetic Ising model.}
%Let $G=(V,E)$ be a simple undirected graph.
%%
%For each edge $e \in E$, we have the local interaction strength $\beta_e \in \mathbb{R}_{> 0}$,
%and for each vertex $v \in V$, we have the external magnetic field (namely vertex weight)  $\lambda_v \in \mathbb{R}_{>0}$.
%%
%An Ising model is specified by the tuple $(G;\*\beta,\*\lambda)$, where $\*\beta = (\beta_e)_{e \in E}$ and $\*\lambda = (\lambda_v)_{v \in V}$.
%%
%We assign spins $\{0,1\}$ to the vertices $V$.
%For each spin configuration $\sigma \in \{0,1\}^V$, the \emph{weight} of $\sigma$ is defined by
%\begin{align*}
%\wt_{\text{Ising}}(\sigma) \defeq \prod_{e=(u,v)\in E}\beta_e^{\mathbb{I}[\sigma(u)=\sigma(v)]}\prod_{u\in V}\lambda^{\sigma (u)}_u,
%\end{align*}
%where $\mathbb{I}[\sigma(u)=\sigma(v)]$ is the indicator variable of the event $\sigma(u)= \sigma(v)$.
%The \emph{Gibbs distribution} $\pi_{\text{Ising}}$ is defined by
%\begin{align}\label{eqn:Ising}
%\forall \sigma \in \{0,1\}^V,\quad \pi_{\text{Ising}}(\sigma) = \frac{\wt_{\text{Ising}}(\sigma)}{Z_{\text{Ising}}},
%\end{align}
%where
%\begin{align*}
%Z_{\text{Ising}} = Z_{\text{Ising}}(G;\*\beta,\*\lambda) \defeq \sum_{\tau \in \{0,1\}^V } \wt_{\text{Ising}}(\tau)	
%\end{align*}
%is the \emph{partition function}.
%In this paper we focus on the \emph{ferromagnetic} case, where $\beta_e > 1$ for all $e \in E$,
%with \emph{consistent} fields, where $\lambda_v\in(0,1]$ for all $v\in V$.
%Note that by flipping the spins, the last assumption is equivalent to assuming $\lambda_v\in[1,\infty)$ for all $v\in V$.

\subsubsection{Weighted random cluster model.} 
The standard random cluster model (at $q=2$) is equivalent to the ferromagnetic Ising model without external fields.
To handle Ising models with fields, we need to introduce weights to the random cluster model.
%The weighted random cluster model can be regarded as a weighted version of the random cluster model. 
Given a graph $G=(V,E)$, the parameters of this model are $\*p = (p_e)_{e \in E}$ and $\*\lambda=(\lambda_v)_{v \in V}$, where $0 < p_e < 1$ and $\lambda_v > 0$. The weight of any subset of edges $S\subseteq E$ is given by
\begin{equation} \label{equ:rc_weight}
  \wt_{\text{wrc}}(S)\defeq \prod_{e \in S} p_e \prod_{f \in E \setminus S} (1-p_f)\prod_{C\in\kappa(V,S)}\left(1+ \prod_{u \in C}\lambda_u \right),
\end{equation}
where $\kappa(V,S)$ is the set of all connected components of the graph $(V,S)$, where each $C \in \kappa(V,S)$ is a subset of vertices that forms a connected subgraph.
The probability that $S$ is drawn is
\begin{align}\label{eqn:wrc}
  \pi_{\text{wrc}}(S)=\frac{\wt_{\text{wrc}}(S)}{Z_{\text{wrc}}}
\end{align}
where
\[
  Z_{\text{wrc}} = Z_{\text{wrc}}(G;\*p,\*\lambda):=\sum_{S\subseteq E}\wt_{\text{wrc}}(S)
\]
is the partition function of the weighted random cluster model.
The (general) standard random cluster model allows a uniform weight $q$ for each connected component,
and in the special case of $\lambda_v = 1$ for all $v \in V$, 
the weighted random cluster model degenerates to the standard random cluster model at $q=2$. 
On the other hand, in our model the weight of each cluster depends on the vertices inside it, 
which makes it different from the standard random cluster models.
%We use $Z_{\-{wrc}}(G;p,1)$\htodo{If we use subscript rc to indicate that there is no weight, why keeping 1 in the argument? Is this notation even needed later?} to denote the partition function of this standard random cluster model.

\subsubsection{Subgraph-world model} Fix a graph $G=(V,E)$. For any subset of edges $S\subseteq E$, denote by $\odd(S)$ the set of vertices with odd degree in $S$. The subgraph-world model \cite{JS93} with parameters $\*p = (p_e)_{e \in E}$ and $\*\eta = (\eta_v)_{v \in V}$ is defined by following: each subset of edges $S$ has weight
\begin{equation} \label{equ:sg_weight}
  \wt_{\text{sg}}(S):=\prod_{e \in S}p_e \prod_{ f \in E \setminus S }(1-p_f)\prod_{v \in \odd(S)}\eta_v. 
\end{equation}
The probability that $S$ is drawn is
\begin{align}\label{eqn:sg}
  \pi_{\text{sg}}(S)=\frac{\wt_{\text{sg}}(S)}{Z_{\text{sg}}}
\end{align}
where
\[
  Z_{\text{sg}}=Z_{\text{sg}}(G;\*p,\*\eta):=\sum_{S\subseteq E}\wt_{\text{sg}}(S)
\]
is the partition function of the subgraph-world model. In the special case where $p_e = p \in (0,1)$ for all $e \in E$ and $\eta_v=0$ for all $v \in V$, the weight of any subgraph $S$ does not vanish if and only if $S$ is an even subgraph, i.e., $\odd(S)=\emptyset$. This yields the even subgraph model, or the so-called ``high-temperature expansion'' in the context of statistical mechanics.
%Denote by $Z_{\text{even}}(G;p)$ its partition function. %\htodo{Same issue as the above}

\subsubsection{Equivalences of the three models}
We have the following equivalence result among the ferromagnetic Ising model with external fields, the weighted random cluster model, and the subgraph-world model. 
The proof of the equivalence result is given in \Cref{app-eq} for completeness.
\begin{proposition}\label{theorem-eqv}
Given any graph $G=(V,E)$, any $\*\beta = (\beta_e)_{e \in E}$ and $\*\lambda = (\lambda_v)_{v \in V}$ satisfying $\beta_e > 1$ for all $e \in E$ and $0 < \lambda_v \leq 1$ for all $v \in V$, it holds that
\begin{equation} \label{equ:three_equivalence}
 \tp{\prod_{e \in E} \beta_e} \cdot Z_{\text{wrc}}(G;2\*p,\*\lambda)=Z_{\text{Ising}}(G;\*\beta,\*\lambda)=\tp{\prod_{v \in V}(1+\lambda_v)}\tp{\prod_{e \in E}\beta_e}  Z_{\text{sg}}(G;\*p,\*\eta),
\end{equation}
where $\*p = (p_e)_{e \in E}$ satisfying $p_e = \frac{1}{2}\tp{1 - \frac{1}{\beta_e}}$ and $\*\eta = (\eta_v)_{v \in V}$ satisfying $\eta_v=\frac{1-\lambda_v}{1+\lambda_v}$.
\end{proposition}

In addition, there are also probabilistic equivalence relations among the models, which will be the topic in \Cref{sec:grand}.

% $\lambda=1$, the product in (\ref{equ:rc_weight}) becomes $2^{|\kappa(V,S)|}$, corresponding to the usual random cluster model with parameter $(p,2)$. Denote by $Z_{\-{wrc}}(G;p,2)=Z_{\text{wrc}}(G;p,\lambda)$ its partition function. 
\begin{remark}
For the ferromagnetic Ising model $(G; \*\beta, \*\lambda) = (G; \beta, 1)$, where $\beta_e = \beta > 1$ for all $e \in E$ and $\lambda_v = 1$ for all $v \in V$, 
its relationship with the even subgraph model and the random cluster model is well known (see e.g.~\cite{vdW41,FK72,Gri06}). Formally,
\[
  \beta^{|E|}Z_{\text{wrc}}(G;2p,1)=Z_{\text{Ising}}(G;\beta,1)=2^{|V|}\beta^{|E|}Z_{\text{sg}}(G;p,0)%Z_{\text{even}}(G;p), 
  \text{ where } p=\frac{1}{2}\left(1-\frac{1}{\beta}\right),
\]
which is a special case of \Cref{theorem-eqv}. 
\end{remark}

\subsection{\texorpdfstring{$f$}{f}-divergences}
A widely-used quantity for measuring the difference between two distributions is the $f$-divergence.
Let $f: \mathbb{R}_{\geq 0} \to \mathbb{R}$ be a \emph{convex} function satisfying $f(1) = 0$. 
Let $\mu$ be a distribution with (finite) support $\Omega = \Omega(\mu)$.
Let $\nu$ be a distributions with support $\Omega(\nu) \subseteq \Omega$.
The $f$-divergence between $\nu$ and $\mu$ is defined by 
\begin{align*}
\DF{\nu}{\mu} \defeq \Ex_{X \sim \mu} \left[ f\tp{\frac{\nu(X)}{\mu(X)}} \right]. 
\end{align*}
In this paper, we consider three important $f$-divergences: the total variation distance, the $\chi^2$-divergence, and the Kullback-Leibler divergence (KL divergence).

Let $f(x) = \frac{1}{2}|x - 1|$. The  \emph{total variation distance} between $\nu$ and $\mu$ is defined by
\begin{align*}
	\DTV{\nu}{\mu} \defeq \frac{1}{2} \sum_{x \in \Omega} \abs{\nu(x) - \mu(x)}. 
\end{align*}
We say the random variable $(X,Y) \in \Omega \times \Omega$ is a \emph{coupling} between $\nu$ and $\mu$ if the marginal distributions satisfy $X \sim \nu$ and $Y \sim \mu$. 
The \emph{coupling inequality} states that for any coupling $(X,Y)$,
\begin{align}\label{eqn:coupling-ineq}
\Pr[X \neq Y] \geq \DTV{\nu}{\mu},
\end{align}
and there exists an \emph{optimal coupling} between $\nu$ and $\mu$ such that equality holds.

Let $f(x) = x^2-1$. The  \emph{$\chi^2$ divergence} between $\nu$ and $\mu$ is defined by
\begin{align*}
	\chisq{\nu}{\mu} \defeq \sum_{x \in \Omega}\frac{\nu^2(x)}{\mu(x)} - 1.
\end{align*}
A similar notion is the relative variance of a function $g:\Omega\rightarrow\=R_{\ge 0}$ over $\mu$:
\begin{align*}
	\Var{\mu}{g} = \Ex_{\mu}[g^2] - \Ex_{\mu}^2[g] = \sum_{x \in \Omega}\mu(x)g^2(x) - \tp{ \sum_{x \in \Omega}\mu(x)g(x) }^2.
\end{align*}
Clearly, if $g(x) = \frac{\nu(x)}{\mu(x)}$, then $\Var{\mu}{g} = \chisq{\nu}{\mu}$.
The following relation is well-known
\begin{align}\label{eqn:chisq}
	\DTV{\nu}{\mu} \leq \sqrt{ \chisq{\nu}{\mu}  }.
\end{align}

Let $f(x) = x\log x$. %The  \emph{$\chi^2$ divergence} between $\nu$ and $\mu$ is defined by
The  \emph{Kullback-Leibler divergence} (KL divergence) is defined by
\begin{align*}
  \KL{\nu}{\mu}\defeq\sum_{x\in \Omega}\nu(x)\log\left( \frac{\nu(x)}{\mu(x)} \right).
\end{align*}
A similar notion is the relative entropy of a function $g:\Omega\rightarrow\=R_{\ge 0}$ over $\mu$:
\begin{align*}
  \Ent{\mu}{g} \defeq \Ex_{\mu}[g \log g] - \Ex_{\mu}[g] \log \Ex_{\mu}[g]
   =  \sum_{x\in\Omega}\mu(x)g(x)\log g(x) - \tp{\sum_{x\in\Omega}\mu(x)g(x)} \log \tp{\sum_{x\in\Omega}\mu(x)g(x)},
\end{align*}
where the convention is that $0\log 0=0$.
Clearly, if $g(x)=\frac{\nu(x)}{\mu(x)}$, $\Ent{\mu}{g}=\KL{\nu}{\mu}$.
The following Pinsker's inequality is well known % and  inverse Pinsker's inequality are well-known
\begin{align}\label{eqn:Pinsker}
  \DTV{\nu}{\mu} \leq \sqrt{\frac{\KL{\nu}{\mu}}{2}}. %\leq \frac{1}{\min_{x \in \Omega}\mu(x)}	\DTV{\nu}{\mu}^2 .
\end{align}

%The following data-processing  inequality is well-known.
For any stochastic  matrix $P$ that transforms any $x \in \Omega$ to a random $y \in \Omega'$ ($\Omega'$ is not necessarily the same as $\Omega$), the following data-processing  inequality is well-known: for any $f$-divergence,
\begin{align*}
	\DF{\nu P}{\mu P} \leq \DF{\nu}{\mu}.
\end{align*}

\subsection{Markov chains and down-up walks}
Let $\Omega$ be a finite state space. Let $(X_t)_{t \geq 0}$ be a Markov chain over $\Omega$ and $P$ denote the transition matrix. 
We say $P$ is 
\begin{itemize}
	\item \emph{irreducible} if for any $x ,y \in \Omega$, there exists $t  > 0$ such that $P^t(x, y ) > 0$;
	\item \emph{aperiodic} if $\gcd \{t \mid P^t(x,x) > 0\} = 1$ for all $x \in \Omega$;
	\item \emph{reversible} with respect to $\mu$ if the following detailed balance equation holds
	\begin{align*}
		\forall x,y \in \Omega, \quad \mu(x)P(x,y) = \mu(y)P(y,x).
	\end{align*}
\end{itemize}
We say the distribution $\mu$ is a stationary distribution of $P$ if $\mu P = \mu$.
If $P$ is reversible with respect to $\mu$, then $\mu$ is a stationary distribution of $P$.
If $P$ is both irreducible and aperiodic, then $P$ has a unique stationary distribution.
The \emph{mixing time} of $P$ is defined by
\begin{align*}
	\forall \epsilon > 0, \quad \mixingtime{P,\epsilon} \defeq \max_{x \in \Omega} \min \{t \mid \DTV{P^t(x, \cdot)}{\mu} \leq \epsilon\}.
\end{align*}

In this paper, we consider two Markov chains: Glauber dynamics and Swendsen-Wang dynamics.
It will be convenient for us to view Glauber dynamics as a so-called ``down-up'' walk, which we will define next.

Let $\Omega_0$ and $\Omega_1$ denote two finite state spaces.
Let $\mu_0$ and $\mu_1$ denote two distributions over $\Omega_0$ and $\Omega_1$ respectively. 
For $f,g: \Omega_i\to \mathbb{R}$, define $\inner{f}{g}_{\mu_i} = \sum_{x \in \Omega_i}\mu_i(x)f(x)g(x)$. 
Let $\Pup: \Omega_0 \times \Omega_1 \to \mathbb{R}_{\geq 0}$ and $\Pdown: \Omega_1 \times \Omega_0 \to \mathbb{R}_{\geq 0}$ denote two transition matrices.
We say $\Pup$ and $\Pdown$ are a pair of adjoint operator if 
\begin{align*}
	\forall f: \Omega_0 \to \mathbb{R},\, g : \Omega_1 \to \mathbb{R},\quad \inner{f}{\Pup g}_{\mu_0} = \inner{\Pdown f}{g}_{\mu_1}.
\end{align*}
The following equation holds for adjoint $\Pup$ and $\Pdown$:
\begin{align*}
	\forall x_0 \in \Omega_0, x_1 \in \Omega_1, \quad \mu_0(x_0)\Pup(x_0,x_1) = \mu_1(x_1)\Pdown(x_1,x_0).
\end{align*}
Moreover, for any distribution $\nu$ over $\Omega_1$ and $f=\frac{\nu}{\mu_1}$, it holds that
\begin{align*}
  \KL{\nu \Pdown}{\mu_0} = \Ent{\mu_1}{\Pup f} \text{ and } \chisq{\nu \Pdown}{\mu_0} = \Var{\mu_1}{\Pup f}.
\end{align*}
It is straightforward to verify $\Pdownup = \Pdown\Pup$ and $\Pupdown = \Pup\Pdown$ are self-adjoint, i.e. $\inner{f}{\Pdownup g}_{\mu_1} = \inner{\Pdownup f}{g}_{\mu_1}$ and $\inner{f}{\Pupdown g}_{\mu_0} = \inner{\Pupdown f}{g}_{\mu_0}$.
Hence, $\Pdownup$ and $\Pupdown$ are reversible with respect to $\mu_1$ and $\mu_0$ respectively.
%$f,g: \Omega_0\to \mathbb{R}$, 

\subsubsection{Glauber dynamics.}\label{sec:GD}
Given a distribution $\mu$ with support $Q^V$, let $\Omega_1 = Q^V$ and $\Omega_0 = \{\sigma \in Q^{V \setminus \{v\}} \mid v \in V\}$. 
and the current state $X \in \Omega$, the transition $X \to X'$ of Glauber dynamics can be interpreted as the following two steps
\begin{itemize}
	\item down walk $\Pdown_{\mathrm{Glauber}}$: pick $v \in V$ uniformly at random and transform $X \in \Omega_1$ to $X_{V \setminus v} \in \Omega_0$;
	\item up walk $\Pup_{\mathrm{Glauber}}$: sample $c \sim \mu_v^{X_{V \setminus \{v\}}}$ and transform $X_{V \setminus v} \in \Omega_0$ to $X' \in \Omega_1$ such that $X_v' = c$ and $X_{V \setminus \{v\}}' = X_{V \setminus \{v\}}$.
\end{itemize}
Let $\mu_0 \defeq \mu \Pdown_{\mathrm{Glauber}}$ be a distribution over $\Omega_0$.
Then $\Pdown_{\mathrm{Glauber}}$ and $\Pup_{\mathrm{Glauber}}$ is a pair of adjoint operators with respect to distributions $\mu_1 = \mu$ and $\mu_0$. 
Thus, Glauber dynamics is a down-up walk and is reversible with respect to $\mu$.

\subsubsection{Swendsen-Wang dynamics.}\label{sec:SW}
Let $G = (V,E)$ be a graph. 
Consider the ferromagnetic Ising model on $G$ with parameters $\*\beta = (\beta_e)_{e \in E}$ and $\*\lambda=(\lambda_v)_{v \in V}$, where $\beta_e > 1$ for all $e \in E$,
%Let $\pi_{\mathrm{Ising}}$ over $\{0,1\}^V$ denote the Gibbs distribution of the Ising model. 
and the weighted random cluster model on $G$ with parameters $\*p = (p_e)_{e \in E}$ and $\*\lambda = (\lambda_v)_{v \in V}$, where $p_e = 1 - \frac{1}{\beta_e}$ for all $e\in E$.
%Let $\pi_{\mathrm{WRC}}$ over $\{0,1\}^E$ denote the distribution induced from weighted random cluster models.
Recall $\pi_{\mathrm{Ising}}$ from \eqref{eqn:Ising} and $\pi_{\mathrm{wrc}}$ from \eqref{eqn:wrc}.

Define the following two transformations between Ising and weighted random cluster models.
\begin{itemize}
	\item $P_{\+I \to \+R}: \{0,1\}^V \to 2^E$: Given any Ising configuration $\sigma \in \{0,1\}^V$, $P_{\+I \to \+R}$ transforms $\sigma$ into a weighted random cluster model configuration $S \subseteq E$. For each edge $e = \{u,v\} \in E$ with $\sigma(u) = \sigma(v)$, add $e$ independently into $S$ with probability $p_e = 1 - \frac{1}{\beta_e}$. Formally,
	\begin{align}\label{eq-def-P-I-R}
		\forall \sigma \in \{0,1\}^V, S \subseteq E,\quad P_{\+I \to \+R}(\sigma,S) = \mathbb{I}[S\subseteq M(\sigma)]\cdot \prod_{e \in S}\tp{1-\frac{1}{\beta_e}}\cdot \prod_{f \in M(\sigma) \setminus S}\frac{1}{\beta_f},
	\end{align}
	where $M(\sigma) = \{e=\{u,v\} \in E \mid \sigma_u = \sigma_v\}$ is the set of monochromatic edges with respect to $\sigma$.
	\item $P_{\+R \to \+I}: \{0,1\}^E \to \{0,1\}^V$: Given any weighted random cluster model configuration $S \subseteq E$, $P_{\+R \to \+I}$ transforms $S$ to an Ising configuration $\sigma \in \{0,1\}^V$. For each connected component $C \subseteq V$ in graph $G' = (V,S)$, sample $x_C \in \{0,1\}$ independently according to the following distribution
	\begin{align*}
		x_C = \begin{cases}
			1 &\text{with probability } \frac{\prod_{v \in C}\lambda_v}{1+\prod_{v \in C}\lambda_v};\\
			0 &\text{with probability } \frac{1}{1+\prod_{v \in C}\lambda_v},\\
		\end{cases}
	\end{align*}
	and then let $\sigma(v) = x_C$ for all vertices $v \in C$. Formally,
	\begin{align}\label{eq-def-P-R-I}
	\forall \sigma \in \{0,1\}^V, S \subseteq E,\quad P_{\+R \to \+I}(S,\sigma)  = \mathbb{I}[S \subseteq M(\sigma)] \cdot \prod_{C \in\kappa(V,S)} 	\frac{\prod_{v \in C}\lambda_v^{\sigma(v)}}{1+\prod_{v \in C}\lambda_v},
	\end{align}
	where $\kappa(V,S)$ is the set of connected components in graph $G'=(V,S)$.
\end{itemize}

The Swendsen-Wang dynamics for Ising models is defined by 
\begin{align}\label{eq-SW-Ising}
  P_{\mathrm{SW}}^{\mathrm{Ising}} \defeq P_{\+I \to \+R}P_{\+R \to \+I},
\end{align}
and the Swendsen-Wang dynamics for weighted random cluster models is defined by 
\begin{align}\label{eq-SW-wrc}
  P_{\mathrm{SW}}^{\mathrm{wrc}} \defeq P_{\+R \to \+I}P_{\+I \to \+R}.
\end{align}

The following adjoint result about Swendsen-Wang dynamics is well-known.
However, here we consider more general Ising models with external fields and weighted random cluster models. 
For completeness, we provide a proof of the following proposition in \Cref{app-adjoint}.
\begin{proposition}\label{proposition-adjoint}
For any functions $f: \{0,1\}^V \to \mathbb{R}$ and $g: 2^E \to \mathbb{R}$, it holds that 
\begin{align}
\label{eq-self-adjoint}
	\inner{f}{P_{\+I \to \+R}g}_{\pi_{\-{Ising}}} = \inner{P_{\+R \to \+I}f}{g}_{\pi_{\-{wrc}}}.
\end{align}	
\end{proposition}
By \Cref{proposition-adjoint}, it holds that $\pi_{\-{Ising}}P_{\+I \to \+R} =\pi_{\-{wrc}} $ and $\pi_{\-{wrc}}P_{\+R \to \+I} = \pi_{\-{Ising}}$.
Both  $P_{\mathrm{SW}}^{\mathrm{Ising}}$ and $P_{\mathrm{SW}}^{\mathrm{wrc}}$ are down-up walks, and their stationary distributions are $\pi_{\-{Ising}}$ and $\pi_{\-{wrc}}$ respectively. 

Finally, the mixing times of  $P_{\mathrm{SW}}^{\mathrm{Ising}}$ and $P_{\mathrm{SW}}^{\mathrm{wrc}}$ have the following relationships:
\begin{align}\label{eq-mixing-sw}
T_{\-{mix}}\tp{P_{\mathrm{SW}}^{\mathrm{Ising}}, \epsilon} \leq T_{\-{mix}}\tp{P_{\mathrm{SW}}^{\mathrm{wrc}}, \epsilon} + 1 \quad\text{and}\quad T_{\-{mix}}\tp{P_{\mathrm{SW}}^{\mathrm{wrc}}, \epsilon} \leq T_{\-{mix}}\tp{P_{\mathrm{SW}}^{\mathrm{Ising}}, \epsilon} + 1.
\end{align}
We prove the first one, the second one holds similarly. 
Let $T=T_{\-{mix}}\tp{P_{\mathrm{SW}}^{\mathrm{wrc}}, \epsilon}$.
For any distribution $\nu$ over $\{0,1\}^V$, we have 
\begin{align*}
	\DTV{\nu (P_{\mathrm{SW}}^{\mathrm{Ising}})^{T+1} }{\pi_{\-{Ising}}} &=  	\DTV{(\nu P_{\+I \to \+R}) (P_{\mathrm{SW}}^{\mathrm{wrc}})^T P_{\+R \to \+I } }{\pi_{\-{wrc}}P_{\+R \to \+I} }\\
\text{(by data processing inequality)}\quad	 &\leq \DTV{(\nu P_{\+I \to \+R}) (P_{\mathrm{SW}}^{\mathrm{wrc}})^T }{\pi_{\-{wrc}} } \leq \epsilon.
\end{align*}
%By \Cref{proposition-adjoint}, we have the following result.
%\begin{proposition}\label{proposition-SW}
%The following results hold for $P_{\+I \to \+R}$, $P_{\+R \to \+I}$, $P_{\mathrm{SW}}^{\mathrm{Ising}}$ and $P_{\mathrm{SW}}^{\mathrm{wrc}}$:
%\begin{itemize}
%	\item for any $\sigma \in \{0,1\}^V$, $S \subseteq 2^E$, $\pi_{\-{Ising}}(\sigma)P_{\+I \to \+R}(\sigma, S) = \pi_{\-{wrc}}(S)P_{\+R \to \+I}(\sigma)$. 
%	\item $\pi_{\-{Ising}}P_{\+I \to \+R} = \pi_{\-{wrc}}$ and $\pi_{\-{wrc}}P_{\+R \to \+I} = \pi_{\-{Ising}}$;
%	\item $P_{\mathrm{SW}}^{\mathrm{Ising}}$ is reversible with respect to $\pi_{\-{Ising}}$; $P_{\mathrm{SW}}^{\mathrm{wrc}}$ is reversible with respect to $\pi_{\-{wrc}}$.
%	\item $P_{\mathrm{SW}}^{\mathrm{Ising}}$ and $P_{\mathrm{SW}}^{\mathrm{wrc}}$ are both positive semidefinite.
%\end{itemize}
%\end{proposition}

\subsection{Canonical paths and variance decay}
Let $P$ denote a random walk over $\Omega$ that is reversible with respect to $\mu$.
%Let $\Pdownup = \Pdown\Pup$ denote the down-up walk over $\Omega_1$, where $\Pdown: \Omega_1 \times \Omega_0 \to \mathbb{R}_{\geq 0}$ and $\Pup: \Omega_0 \times \Omega_1 \to \mathbb{R}_{\geq 0}$ are a pair of adjoint operators with respect to distribution $\mu_1$ over $\Omega_1$ and $\mu_2$ over $\Omega_2$. 
%
%It is straightforward to verify $\Pdownup$ has the stationary distribution $\mu_1$ and non-negative real eigenvalues $1 \geq \lambda_1 \geq \lambda_2 \geq \ldots \geq \lambda_{|\Omega_1|}$. 
%For simplicity, we denote $\Omega_1$ by $\Omega$, and we denote $\mu_1$ by $\mu$.
It is well-known that $P$ has real eigenvalues $1 \geq \lambda_1 \geq \lambda_2 \geq \ldots \geq \lambda_{|\Omega|}$. 
The \emph{spectral gap} is defined by
	\begin{align*}
  \Gap(P) = 1 - \lambda_2.
\end{align*}

Define the Dirichlet form of $P$ by for any functions $f,g: \Omega \to \mathbb{R}$,
\begin{align*}
	\+E_{P}(f,g) = \inner{f}{(I-\Pdownup)g}_{\mu} = \frac{1}{2}\sum_{x,y \in \Omega}\mu(x)P(x,y)(f(x)-f(y))(g(x)-g(y)).
\end{align*} 
We can also characterise the spectral gap $\Gap(P)$ in a variational form:
\begin{align}\label{eq-def-gap}
 \Gap(P) = \inf \left\{ \frac{\+E_{P}(f,f)}{\Var{\mu}{f}} \mid f:\Omega \to \mathbb{R} \land \Var{\mu}{f} \neq 0  \right\}.
\end{align}

A useful tool to analyse the spectral gap of a reversible Markov chains is the canonical path introduced by Jerrum and Sinclair~\cite{JS89}.
Let $P$ be a reversible Markov chain over the state space $\Omega$ with stationary distribution $\pi$.
%For any pair $x,y \in \Omega$, let $\gamma_{xy}=(z_0 = x, z_1,z_2,\ldots,z_\ell = y)$ denote the canonical path from $x$ to $y$ that moves in the state space $\Omega$ using the transitions of the Markov chain $P$, 
%i.e.~for any $1 \leq i < \ell$, $P(z_i,z_{i+1}) > 0$. 
%We call $\ell$ the length of the path $P$.
%Denote by $x\overset{\gamma}{\leadsto} y$ that $\gamma$ moves from $x$ to $y$.
%
%In this paper, we use the multicommodity flow formulation of this technique. 
Let $\gamma_{XY}=(Z_0 = X, Z_1,Z_2,\ldots,Z_\ell = Y)$ be a path of length $\ell$ moving in the state space using transitions of $P$, i.e.~for any $i\in[\ell]$, $P(Z_{i-1},Z_{i}) > 0$.
For each pair of $X,Y\in\Omega$, its path $\gamma_{XY}$ is assigned a weight $w(\gamma_{XY})=\mu(X)\mu(Y)$.
Let $\Gamma$ be the collection of all canonical paths. 
The congestion of $\Gamma$ is defined by
\begin{equation}
\varrho(\Gamma):=\max_{(Z,Z')\in\Omega^2,P(Z,Z')>0}\frac{L}{\mu(Z)P(Z,Z')}\sum_{\gamma\in\Gamma:(Z,Z')\in\gamma}w(\gamma)
\end{equation}
where $L$ is the maximum length of path in $\Gamma$. 
Sinclair \cite{Sin92} showed that the congestion of any collection of paths $\Gamma$ for a Markov chain $P$ is an upper bound of the inverse of its spectral gap, namely,
\begin{align*}
\frac{1}{\Gap(P)} \leq \varrho(\Gamma).
\end{align*}

Consider the down-up walk $\Pdownup = \Pdown\Pup$ over $\Omega_1$, where $\Pdown: \Omega_1 \times \Omega_0 \to \mathbb{R}_{\geq 0}$ and $\Pup: \Omega_0 \times \Omega_1 \to \mathbb{R}_{\geq 0}$ are a pair of adjoint operators with respect to distribution $\mu_0$ over $\Omega_0$ and $\mu_1$ over $\Omega_1$. 
For simplicity, we denote $\Omega_1$ by $\Omega$, and we denote $\mu_1$ by $\mu$.
The following result holds for $\Pdownup$.
\begin{proposition}\label{theorem-gap-decay}
Let $\Pdownup = \Pdown\Pup$ be a down-up walk over $\Omega$ that is reversible with respect to $\mu$. 
For any $0 < \delta <1$, the spectral gap $\Gap(\Pdownup) \geq \delta$ if and only if
 for any distribution $\nu$ over $\Omega$,
\begin{align}\label{eq-chi-decay}
	\chisq{\nu \Pdown}{\mu \Pdown} \leq \tp{1 - \delta} \chisq{\nu}{\mu}. 
\end{align} 
\end{proposition}
\begin{proof}
%Assume $\Gap(\Pdownup) \geq \delta$.
Let $f = \frac{\nu}{\mu}$.
It holds that 
\begin{align*}
	\+E_{\Pdownup}(f,f) = \inner{f}{f}_{\mu} - \inner{f}{\Pdownup f}_{\mu} =   \inner{f}{f}_{\mu} - \inner{\Pup f}{\Pup f}_{\mu_0} = \Var{\mu}{f} - \Var{\mu_0}{\Pup f}.
\end{align*}
Then the lemma follows from $\chisq{\nu \Pdown}{\mu \Pdown} = \Var{\mu_0}{\Pup f}$, $\chisq{\nu}{\mu} = \Var{\mu}{f}$, and~\eqref{eq-def-gap}.
%Assume that~\eqref{eq-chi-decay} holds.	By the above argument, for any function $f$ with $\Ex_{\mu}[f] = 1$,
%\begin{align*}
%	\+E_{\Pdownup}(f,f) \geq \delta \Var{\mu}{f}.
%\end{align*}
%The above result can be extended to all functions $f$ by normalisation, which implies $\Gap(\Pdownup) \geq \delta$.
\end{proof}

%the following bound
%\begin{align*}
%	\frac{1}{\lambda_{\-{gap}}(P)} \leq \varrho(\Gamma).
%\end{align*}
%Furthermore, if $P$ is positive semidefinite, we have
%\begin{align*}
%  T_{\-{mix}}(P,\epsilon) \leq \varrho(\Gamma)	\tp{ \log \frac{1}{\mu_{\min}}+ \log \frac{1}{\epsilon}}.
%\end{align*}

\subsection{Spectral independence and entropy decay}\label{section-si-pre}
Let $Q$ be a finite set.
Let $\mu$ be a distribution with support $Q^V$. 
Fix a partial pinning $\tau \in Q^\Lambda$ for some $\Lambda \subseteq V$. Define the \emph{absolute influence matrix} $\Psi^\tau_\mu$ by  %$\Psi^\tau_\mu: (V \setminus \Lambda) \times (V \setminus \Lambda)\to \mathbb{R}_{\geq 0}$ by  %for all $u,v\in V \setminus \Lambda$, $\Psi^\tau_\mu(u,v) = $
\begin{align*}
	\forall u,v \in V \setminus \Lambda \text{ with } u \neq v,&\quad \Psi^\tau_\mu(u,v) \defeq \max_{i,j \in Q}\DTV{\mu^{\tau \wedge (u \gets i)}_v}{\mu^{\tau \wedge (u \gets j)}_v}\\
	\forall v \in V \setminus \Lambda, &\quad \Psi^\tau_\mu(v,v) \defeq  0.
\end{align*}
where $\DTV{\cdot}{\cdot}$ denotes the total variation distance and $\mu^{\tau \wedge (u \gets i)}_v$ denotes the marginal distribution on $v$ conditional on that variables in $\Lambda$ take the value $\tau$ and $u$ takes the value $i$.
We say that the distribution $\mu$ is $\ell_\infty$-\emph{spectrally independent} with parameter $\zeta$ if
\begin{align*}
	\forall \Lambda \subset V, \sigma \in Q^\Lambda, \quad \norm{ \Psi^\sigma_\mu }_\infty = \max_{u \notin \Lambda} \sum_{v \notin \Lambda}\Psi^\sigma_\mu(u,v) \leq \zeta.
\end{align*}
Call $\mu$ \emph{$b$-marginally bounded} if
\begin{align*}
	\min_{\Lambda \subseteq V, v \notin \Lambda} \min_{\sigma \in Q^\Lambda, c \in Q} \mu^\sigma_v(c) \geq b. %\min_{v \in V \setminus \Lambda} 
\end{align*}

In this paper, we are particularly interested in \emph{Gibbs distributions}.
%Let $V$ be a set of variables, where each $v \in V$ takes values from a finite domain $Q$.
We will consider a slightly more general than usual version defined over hypergraphs.
%Let $Q$ be a finite domain.
Let $H = (V,\+E)$ be a hypergraph.
Given weight functions $(\phi_v)_{v \in V}$ and $(\phi_e)_{e \in \+E}$, where $\phi_v: Q \to \mathbb{R}_{>0}$ and $\phi_e: Q^e \to \mathbb{R}_{>0}$, define the Gibbs distribution $\mu$ over $Q^V$ by 
\begin{align*}
	\forall \sigma \in Q^V, \quad \mu(\sigma) \propto \prod_{v \in V}\phi_v(\sigma_v) \prod_{e \in \+E} \phi_e(\sigma_e).
\end{align*}
Let $G_\mu = (V,E)$ be a graph such that $\{u,v\} \in E$ if $u \in e'$ and $v \in e'$ for some $e' \in \+E$. 
For any disjoint $A,B,C \subseteq V$, if the removal of $C$ disconnects $A$ and $B$ in $G_\mu$, it holds that variables in $A$ and $B$ are independent in $\mu$ conditional on any assignment on $C$. 
%For any $e \in \+E$, let $\Gamma_\mu(e) = \{e' \in  \+E \mid e' \cap e \neq \emptyset \} $. Note that $e \in \Gamma_\mu(e)$.
%Let $d_v$ denote the number of $e \in \+E$ such that $v \in e$. 
Define \emph{maximum degree} $D_\mu$ of the Gibbs distribution $\mu$ as the maximum degree of the graph $G_{\mu}$.

The spectral independence is related to the mixing time of Glauber dynamics.
The following result is proved in \cite{CLV21,BCCPSV22} (see also \cite[Theorem 13]{CLV21a})
\begin{theorem}[\cite{CLV21,BCCPSV22}]\label{theorem-CLV}
Let $\zeta,b,D >0$. For any Gibbs distribution $\mu$ over $Q^V$, where $|V| = n$, if $\mu$ is $\ell_\infty$-\emph{spectrally independent} with parameter $\zeta$, $b$-marginally bounded and has the maximum degree at most $D$, then the down walk of the Glauber dynamics satisfies that 
\begin{align*}
	\forall \text{distribution } \nu \text{ over } Q^V, \quad \KL{\nu \Pdown_{\mathrm{Glauber}}}{\mu \Pdown_{\mathrm{Glauber}}} \leq \tp{1 - \frac{1}{Cn}}\KL{\nu }{\mu }, 
\end{align*}
where $C = \left(\frac{D}{b}\right)^{1+2\left\lceil\frac{\zeta}{b}\right\rceil}>1$ is a constant depending only on $\zeta,b$ and $D$.
\end{theorem}
In \cite{CLV21,BCCPSV22}, they mainly establish the so-called ``approximate tensorization of entropy'' property for $\mu$.
However this is equivalent to the contraction of relative entropy by~$\Pdown_{\mathrm{Glauber}}$~\cite{CLV21}.

\subsection{Holographic transformation} 
%We show the second equation in  (\ref{equ:three_equivalence}).
%Here we establish the equivalence between the Ising model and subgraph-world model via a holographic reduction \cite{Val08}. 

We will need holographic transformations \cite{Val08} to show couplings between the subgraph-world model and the weighted random cluster model.
Let $f:\{0,1\}^d\rightarrow \=C$ be a function.
We may represent it by a vector (either row or column vector) $\left(f_0,\cdots,f_x,\cdots,f_{2^d-1}\right)$ where $f_x$ is the value of $f$ on $x\in\{0,1\}^d$ by regarding $x$ as a binary representation.
%A signature $f$ of arity $d$ over the Boolean domain can be represented by a vector (either row or column vector) $\left(f_0,\cdots,f_x,\cdots,f_{2^d-1}\right)$ where $f_x$ is the value of $f$ on $x\in\{0,1\}^d$ regarded as a binary representation. 
In the symmetric case where $f$ is invariant under permutations of indices,
we use a succinct ``signature'' $[f_0,\cdots,f_w,\cdots,f_d]$ to express $f$,
where $f_w$ is the value of $f$ on inputs of Hamming weight $w$, i.e. all $x \in \{0,1\}^d$ satisfying $|x| = w$.
%$f_x=f_x'$ if $x$ and $x'$ has equal Hamming weight in binary representation, the signature can also be expressed as $[f_0,\cdots,f_w,\cdots,f_d]$ where $f_w$ is the value of $f$ on input of Hamming weight $w$, i.e. all $x \in \{0,1\}^d$ satisfying $|x| = w$.

Given a bipartite graph $H=(V,E)$ with partition $V = V_1 \uplus V_2$. 
Let $\+F = (f_v)_{v \in V_1}$ and $\+G=(g_v)_{v \in V_2}$ be two sets of functions such that the arity of the function is the degree of the corresponding vertex.
The \emph{Holant} (an edge weighted partition function) is defined by
\begin{align*}
 \holant(H;\+F \mid \+G):=\sum_{\sigma:E\to\{0,1\}}\prod_{v\in V_1}f_v\left(\sigma\mid_{E(v)}\right)\prod_{u\in V_2}g_u\left(\sigma\mid_{E(u)}\right),	
\end{align*}
where $\sigma\mid_{E(v)}$ stands for the restriction of the assignment $\sigma$ to the incident edges of $v$.

Let ${\bm M}$ be a $2\times 2$ matrix and $f$ be a function of arity $d$. 
If $f$ is represented by a column (resp.~row) vector, we write ${\bm M}f={\bm M}^{\otimes d}f$ (resp.~$f{\bm M}=f{\bm M}^{\otimes d}$) as the transformed signature. 
Given $\holant(H;\+F\mid\+G)$ and an invertible matrix ${\bm T}\in\=C^{2\times 2}$, we view signatures in $\+F$ as row vectors and define $\+F{\bm T}=\{f'_v \mid v \in V_1 \land f'_v = f_v\*T\}$; and view signatures in $\+G$ as column vectors and define ${\bm T}^{-1}\+G=\{g'_v \mid v \in V_2 \land g'_v =  \*T^{-1}g_v\}$.
%Define ${\bm T}\+F=\{{\bm T}f\mid f\in\+F\}$ when we view signatures in $\+F$ as column vectors and $\+F{\bm T}=\{f{\bm T}\mid f\in\+F\}$ when we view signatures in $\+F$ as row vectors. 
Valiant's celebrated Holant Theorem \cite{Val08} states

\begin{theorem} \label{thm:holant}
  %If ${\bm T}\in\=C^{2\times 2}$ is an invertible matrix, then we have 
  $\holant(H;\+F\mid\+G)=\holant(H;\+F{\bm T}\mid {\bm T}^{-1}\+G)$.
\end{theorem}

\section{The grand model and a generalised Grimmett--Janson coupling}
\label{sec:grand}

We introduce a grand model, inspired by \cite{GJ07}, that unifies the subgraph and random cluster models introduced in \Cref{sec:models}.
We also generalise the coupling of Grimmett and Janson \cite{GJ07} for ferromagnetic Ising models with external fields.
It is possible to also include vertex configurations in this grand model \`{a} la Edwards and Sokal \cite{ES88},
so that the Ising model is also unified under this framework.
However it does not appear to have much benefit and we choose not to do so.

\subsection{The grand model}\label{sec:grand-def}
Let $G=(V,E)$ be a simple undirected graph. 
The \emph{grand model}, specified by parameters $\*p = (p_e)_{e \in E}$ and $\*\eta = (\eta_v)_{v \in V}$ where $0\leq p_e\leq 1/2$ and $0\leq \eta_v\leq 1$, defines a distribution $\pi_{\-{gm}}$ over all configurations on the edges of three states $X:E\to\{0,1,2\}$. 
Given an assignment $X$ in the grand model, denote by $X^{-1}(q)$ the set of edges that are assigned $q$ under $X$ where $q=0,1,2$. 
The weight of each configuration is given by 
\begin{equation}\label{equ:gm-weight}
%  \wt_{\-{gm}}(X)=\prod_{e \in E:X(e)\in\{1,2\}}p_e\prod_{f \in E:X(e)=0}(1-2p_f)\prod_{v\in\+{O}(X)}\eta_v
  \wt_{\-{gm}}(X)=\prod_{e \in X^{-1}(\{1,2\})}p_e\prod_{f\in X^{-1}(0)}(1-2p_f)\prod_{v\in\+{O}(X)}\eta_v,
\end{equation}
where $\+{O}(X)$ is the set of vertices of odd degree in the subgraph $(V,X^{-1}(1))$. 
The probability of each configuration $X$ is
\begin{equation}\label{equ:gm-prob}
  \pi_{\-{gm}}(X)=\frac{\wt_{\-{gm}}(X)}{Z_{\-{gm}}}
\end{equation}
where
\[
  Z_{\-{gm}}=Z_{\-{gm}}(G;\*p,\*\eta):=\sum_{X\in\Omega_{\-{gm}}(G)}\wt_{\-{gm}}(X)
\]
is the partition function of the grand model. 

Equivalently,  a random sample from the grand model can be generated by the following procedure.
\begin{itemize}
	\item \textbf{Step-I}: Sample $S \sim \pi_{\-{sg}}$, where $\pi_{\-{sg}}$ is the distribution specified by the subgraph-world model with parameters $(\*p,\*\eta)$; for each $e \in E$, let $X(e) = 1$ if $e \in S$ and let $X(e) = *$ if $e \notin S$.
	\item \textbf{Step-II}: Independently for each $e \in E$ with $X_e = *$, set $X(e) = 2$ with probability $\frac{p_e}{1-p_e}$, and $X(e) = 0$ otherwise. 
\end{itemize}
It is straightforward to verify that the outcome distribution is exactly the grand model distribution. 
%

%The grand model has some important properties.
%The following observation follows directly from the definition. 
Recall the definition of a Gibbs distribution and its maximum degree in \Cref{section-si-pre}.
The grand model is indeed a Gibbs distribution, 
where the variables are all edges of $G$, and each vertex $v\in V$ introduces a hyperedge consisting of all edges adjacent to $v$.
In other words, $G_{\pi_{\-{gm}}}$ (as defined in \Cref{section-si-pre}) is the line graph of $G$.
Thus we have the following observation.

\begin{observation}\label{lemma-I}
 The distribution $\pi_{\-{gm}}$ is a Gibbs distribution with maximum degree $D\le 2\Delta-1$, where $\Delta$ is the maximum degree of the graph $G=(V,E)$.
\end{observation}

The next lemma gives the relation among the grand model, the subgraph-world model and the random cluster model.
%how we get a random sample for the subgraph-world model and the random cluster model by projection. 
\begin{lemma}\label{lemma-marginal}
Let $X \sim \pi_{\-{gm}}$ be a random sample from the grand model with parameter  $\*p=(p_e)_{e \in E}$ and $\*\eta = (\eta_v)_{v \in V}$, where $0 \leq p_e \leq 1/2$ and $0\leq \eta_v\leq 1$. It holds that
\begin{itemize}
  \item $\+S = \{e \in E \mid X(e) = 1\}$ follows the distribution specified by the subgraph-world model with parameters $(\*p,\*\eta)$;
  \item $\+R = \{e \in E \mid X(e) = 1 \lor X(e) = 2\}$ follows the distribution specified by the random cluster model with parameters $(2\*p, \*\lambda)$, where $\lambda_v = \frac{1-\eta_v}{1+\eta_v}$ for all $v \in V$.
\end{itemize}	
\end{lemma}
Namely, $X(e)=1$ means $e$ is present in the subgraph-world model (\textbf{Step-I}),
and $X(e)=2$ means $e$ is absent in the subgraph-world model, but gets added into the random cluster model in \textbf{Step-II}. 
$X(e)=0$ means $e$ is absent in both models. 

The first part of \Cref{lemma-marginal} holds trivially. 
The second part is proved by a generalised Grimmett--Janson coupling~\cite{GJ07}.
The proof of the second part is given in \Cref{section-GJ-coupling}.

\subsection{Coupling via holographic transformation}\label{section-GJ-coupling}

Under the unweighted setting, Grimmett and Janson \cite[Theorem 3.5]{GJ07} discovered a coupling between random even subgraphs and random cluster configurations. 
%More specifically, take a random even subgraph $S$ from the distribution $\pi_{\mathrm{even}}$ with parameter $p\leq\frac{1}{2}$ and then add each remaining edge $e\notin S$ independently with probability $p/(1-p)$ to get $R$. Then $R$ is a random cluster configuration with parameters $(2p,2)$. 
%In this subsection, we will show that this coupling still holds from the subgraph-world model to the random cluster model by a holographic reduction. %similar to \Cref{lem:ising-subgraph-holant}. 
The following lemma is a generalisation to the weighted case via holographic transformations.

\begin{lemma} \label{lem:weighted-coupling}
Let $G=(V,E)$ be a graph, $\*p = (p_e)_{e \in E}$ and $\*\eta = (\eta_v)_{v \in V}$, where $0\leq  p_e \leq 1/2$ for all $e \in E$ and $\eta_v \geq 0$ for all $v \in V$.
Let $\+S \subseteq E$ be a random sample from the  subgraph-world model $(G;\*p,\*\eta)$.
%Let $S$ be a uniform sample from the subgraph-world model with parameter $(p,q)$. 
Let $\+R$ be $\+S$ with each remaining edge $e \in E \setminus \+S$ added into $\+R$ independently with probability $p_e/(1-p_e)$. Then the random subgraph $\+R$ satisfies the distribution of the random cluster model with parameter $(2\*p,\*\lambda)$ where $\eta_v=\frac{1-\lambda_v}{1+\lambda_v}$ for all $v \in V$.
\end{lemma}
We remark that the second part of \Cref{lemma-marginal} is a straightforward consequence of \Cref{lem:weighted-coupling}.
%The holographic reduction used in the proof of \Cref{lem:weighted-coupling} is stated as follows. 
We need the following lemma to prove \Cref{lem:weighted-coupling}.
\begin{lemma}\label{lem-aux-hol}
Let $G=(V,E)$ be a graph. Let $\*\lambda = (\lambda_v)_{v \in V}$ where $0 \leq \lambda_v < 1$ for all $v \in V$.
For each $v \in V$, let $\eta_v = \frac{1-\lambda_v}{1+\lambda_v}$. It holds that
\begin{equation} \label{equ:subgraph-counting}
  \prod_{C\in\kappa(V,E)}\left(1+\prod_{u \in C} \lambda_u\right)=\tp{\prod_{v \in V}(1+\lambda_v)}\left(\frac{1}{2}\right)^{|E|}\sum_{E'\subset E}\prod_{u \in \odd(E')}\eta_u,
\end{equation}
where $\kappa(V,E)$ is the set of connected components in graph $G=(V,E)$.
\end{lemma}
\begin{proof}
Define a bipartite graph $H$ with left part $V_1 = V$ corresponding to vertices in $G$ and right part $V_2=E$ corresponding to edges in $G$.
Two vertices $v \in V_1$ and $e \in V_2$ are adjacent in $H$ if $v$ is incident to $e$ in $G$.
Let $d_v$ denote the degree of $v$ in $G$. Consider the following set of signatures
\begin{align*}
	\+F^{(1)} &= \set{f_v^{(1)} = \left[1,0\right]^{\otimes d_v}+\lambda_v \left[0,1\right]^{\otimes d_v}\mid v \in V},\\
	\+F^{(2)} &= \set{f_v^{(2)} = \frac{1}{1+\lambda_v}\tp{[1,1]^{\otimes d_v} + \lambda_v[1,-1]^{\otimes d_v}} \mid v \in V},\\
	\+G &= \set{g_e = [1,0,1] \mid e \in E}.
\end{align*}
We remark that $f_v^{(2)} = [1,\eta_v,1,\eta_v,\ldots]$.
Let $\*T=\Tmatrix$. Observe that $f^{(1)}_v \*T = (1+\lambda_v)f_v^{(2)}$ and $\*T ^{-1}g_e =\frac{1}{2}g_e$. By \Cref{thm:holant}, it holds that
\begin{align}\label{equ:subgraph-counting-holant}
\holant \tp{H;\+F^{(1)} \mid \+G} = \tp{\prod_{v \in V}(1+\lambda_v)} \tp{\frac{1}{2}}^{|E|} 	\holant \tp{H;\+F^{(2)} \mid \+G}.
\end{align}
This equation is indeed \eqref{equ:subgraph-counting} in disguise.
The equivalence between the left-hand sides of (\ref{equ:subgraph-counting-holant}) and (\ref{equ:subgraph-counting}) is a simple observation that the signature $[1,0,1]$ on the edge forces the spins of vertices in each connected component $C$ to be the same. Each component contributes a weight $1+\prod_{u \in C}\lambda_u$. The equivalence between the right-hand sides of (\ref{equ:subgraph-counting-holant}) and (\ref{equ:subgraph-counting}) follows from how $\+F^{(2)}$ and $\+G$ are defined.
This proves the lemma. 
\end{proof}

\begin{proof}[Proof of \Cref{lem:weighted-coupling}]

  For each subgraph $R\subseteq E$ of $G=(V,E)$, 
  \begin{align*}
  \Pr[\+R=R]&=\frac{1}{Z_{\text{sg}}(G;\*p,\*\eta)}\sum_{S \subseteq R}\prod_{u \in \odd(S)}\eta_u \prod_{e \in S}p_e \prod_{f \in E \setminus S}(1-p_f)\prod_{g \in R \setminus S} \frac{p_g}{1-p_g} \prod_{h \in E \setminus R} \frac{1-2p_h}{1-p_h} \\
  &=\frac{1}{Z_{\text{sg}}(G;\*p,\*\eta)}\sum_{S\subseteq R}\prod_{u \in \odd(S)}\eta_u  \prod_{e \in R}p_e \prod_{f \in E \setminus R} (1-2p_f)\\
  &=\frac{1}{Z_{\text{sg}}(G;\*p,\*\eta)}2^{-|R|}\prod_{e \in R}(2p_e) \prod_{f \in E \setminus R} (1-2p_f)\sum_{S\subseteq R}\prod_{u \in \odd(S)}\eta_u\\
  &=\frac{1}{Z_{\text{sg}}(G;\*p,\*\eta)}\prod_{e \in R}(2p_e) \prod_{f \in E \setminus R} (1-2p_f) \prod_{v \in V} \frac{1}{1+\lambda_v}\prod_{C\in\kappa(V,R)}\left(1+\prod_{u\in C}\lambda_u\right)\tag{By (\ref{equ:subgraph-counting}) on $(V,R)$}\\
  &=\frac{1}{Z_{\-{wrc}}(G;2\*p,\*\lambda)}\prod_{e \in R}(2p_e) \prod_{f \in E \setminus R} (1-2p_f) \prod_{C\in\kappa(V,R)}\left(1+\prod_{u\in C}\lambda_u\right).\tag{By (\ref{equ:three_equivalence})}\\ 
  &=\pi_{\-{wrc}}(R). \qedhere
  \end{align*} 
\end{proof}

\section{Variance decay of Glauber dynamics on the grand model}\label{section-var}
Let $G=(V,E)$ be a graph. 
Let $\*p = (p_e)_{e \in E}$ and $\*\eta = (\eta_v)_{v \in V}$, where $0 < p_e < 1/2$ and $0 < \eta_v < 1$.  
Let $\pi_{\-{gm}}$ denote the distribution specified by the grand model with parameters $\*p$ and $\*\eta$.
Let $\Omega(\pi_{\-{gm}})$ denote the support of $\pi_{\-{gm}}$.
%Let $\pi_{\-{wrc}}$ denote the distribution specified by the random cluster model with parameters $2\*p$ and $\*\lambda$, where $\lambda_v = \frac{1-\eta_v}{1+\eta_v}$. 
We use $P_{\-{GlauberGM}}$ to denote Glauber dynamics on $\pi_{\-{gm}}$ as defined in \Cref{sec:GD}. 
%We use $P_{\-{GlauberGM}}^{\downarrow}(\pi_{\-{gm}})$ and $P_{\-{GlauberGM}}^{\downarrow}(\pi_{\-{wrc}})$ to denote the corresponding down walk of the Glauber dynamics. 

\begin{lemma}\label{lemma-var-decay}
The Glauber dynamics $P_{\-{GlauberGM}}$  satisfies that for any distribution $\nu$ with support $\Omega(\nu) \subseteq \Omega(\pi_{\-{gm}})$, %it holds that
\begin{align*}
	\chisq{\nu P_{\-{GlauberGM}}^{\downarrow} }{\pi_{\-{gm}} P_{\-{GlauberGM}}^{\downarrow} } \leq \tp{1 - \frac{\eta_{\min}^4\min\left\{p_{\min},1-2p_{\max}\right\}}{m^2}}\chisq{\nu}{\pi_{\-{gm}}},
\end{align*}
where $\eta_{\min} = \min_{v \in V}\eta_v$ and $m = |E|$.
\end{lemma}

By \Cref{theorem-gap-decay}, we only need to bound the spectral gap of the Glauber dynamics. 
The rest of this section endeavours to show
\begin{equation} \label{equ:gm-spectral-gap}
  \Gap\left(P_{\-{GlauberGM}}\right)\geq \frac{\eta_{\min}^4}{m^2}\min\left\{p_{\min},1-2p_{\max}\right\}.
\end{equation}
This will be proved using the canonical path method adapted from \cite{JS93}.
%We introduce the following  natural Metropolis dynamics. 
%For any grand model configuration $X$, the notation $X^{e\to q}$ represents the configuration by reassigning the edge $e$ with $q$. 
%The Metropolis dynamics is given by
%\begin{align*}
%	\forall X, Y \in \Omega,\quad P_{\-{gm}}(X,Y) = \begin{cases}\displaystyle
%		\frac{1}{4m} \min \set{1, \frac{\pi_{\-{gm}}(Y)}{\pi_{\-{gm}}(X)}}, &\text{if } \left|\left\{e:X(e)\neq Y(e)\right\}\right| = 1;\\ \displaystyle
%	1 - \frac{1}{4m}\sum_{\substack{e \in E, q\in\{0,1,2\}:\\X(e)\neq q}}\min\set{1, \frac{\pi_{\-{gm}} \left(X^{e\to q}\right) }{\pi_{\-{gm}}(X) } } &\text{if } X = Y;\\ \displaystyle
%	0 &\text{otherwise.}
%	\end{cases}
%\end{align*}
%Note that this dynamics is lazy as the transition probability $P_{\-{gm}}(X,X)$ is at least $1/2$. 
%Let $P_{\-{Metropolis}}$ denote the transition matrix of the Metropolis chain. 
%For any configurations $X,Y \in \{0,1,2\}^E$ with $X \neq Y$, it holds that 
%\begin{align*}
%	P_{\-{GlauberGM}}(X,Y) \geq ? P_{\-{Metropolis}}(X,Y).
%\end{align*}
%By \eqref{eq-def-gap}, we only need to prove the following result 
%\begin{align*}
%	\Gap(P_{\-{Metropolis}}) \geq \frac{\eta_{\min}^4}{4m^2}.
%\end{align*}
%In the rest of this section, we analyse the spectral gap of the Metropolis chain. 

\subsection{Construction of the canonical path}

Below is the main lemma of this subsection. 
\begin{lemma} \label{lem:gm-cp}
For any grand model on a graph $G=(V,E)$ with parameters $\*p = (p_e)_{e \in E}$ and $\*\eta = (\eta_v)_{v \in V}$, if $0 < \eta_v < 1$ for all $v \in V$,
then there exists a set of canonical paths $\Gamma=\{\gamma_{XY}:X,Y\in\Omega\}$ for the Glauber dynamics $P_{\-{gm}}$ such that
\begin{enumerate}
  \item $w_{\-{gm}}(X,Y)=\pi_{\-{gm}}(X)\pi_{\-{gm}}(Y)$;
  \item $|\gamma_{XY}|\leq m$;
  \item for any transition $(Z,Z')$ with $\left|\left\{e:Z(e)\neq Z'(e)\right\}\right| = 1$, where the only edge $e$ of discrepancy is assigned $1$ in either $Z$ or $Z'$, it holds that
    \begin{equation} \label{equ:capacity-wind}
      \sum_{\gamma\in\Gamma:(Z,Z')\in\gamma}w_{\-{gm}}(\gamma)\leq \eta_{\min}^{-4}\min\set{\pi_{\-{gm}}(Z),\pi_{\-{gm}}(Z')}
    \end{equation}
    where $\eta_{\min}:=\min_{v} \eta_v$;
  \item for any transition $(Z,Z')$ with $\left|\left\{e:Z(e)\neq Z'(e)\right\}\right| = 1$, where the only edge $e$ of discrepancy is assigned $1$ in neither $Z$ nor $Z'$, it holds that
    \begin{equation} \label{equ:capacity-other}
      \sum_{\gamma\in\Gamma:(Z,Z')\in\gamma}w_{\-{gm}}(\gamma)\leq \min\set{\pi_{\-{gm}}(Z),\pi_{\-{gm}}(Z')}.
    \end{equation}
\end{enumerate}
\end{lemma}

\begin{proof}%[Proof of \Cref{lem:gm-cp}]
We begin the proof with the construction of the paths. 
Suppose all vertices and edges are indexed by distinct integers, and there is a fixed ordering $\prec$ for all paths and cycles of the graph $G$. 
For any pair of assignments $X,Y$ in the grand model, the canonical path $\gamma_{XY}$ contains two stages, moving from $X$ to $W$ and $W$ to $Y$ respectively. 

\textbf{Stage 1.} (\emph{$1$-edge mending.})
Midst this stage we mend the edges assigned $1$ in either $X$ or $Y$ but not the other. 
Denote the set of such edges $D:=X^{-1}(1)\oplus Y^{-1}(1)$. 
The resulting configuration $W$ has the property that
(1) for any edge $e\in D$, it holds that $W(e)=Y(e)$, and
(2) for any other edge $e\notin D$, it holds that $W(e)=X(e)$.  

Let $2k$ be the number of the odd-degree vertices in $D$. 
Then, $D$ can be decomposed into an edge-disjoint union of exactly $k$ paths $P_1,\cdots,P_k$ and cycles $C_1,\cdots,C_{k'}$. 
We pick the unique one such that $P_1,\cdots,P_k,C_1,\cdots,C_{k'}$ is the first one in the lexicographic order induced by $\prec$. 

To move from $X$ to $W$, we process each of the paths and cycles one by one. For each of them, we first choose the vertex and edge to start with. 
When winding (handling) a path, the starting vertex is one of the two open vertices of the path that has a smaller index; 
when winding a cycle, the starting vertex is the one with the smallest index, and the next vertex (which together with the starting one defines a starting edge) is one of the two neighbours of the starting vertex of the cycle that has a smaller index than the other one. 
After deciding the starting vertex and edge, we just move along the path/cycle. For each of the edge, we set the assignment to it as that in $Y$. Obviously this gives $W$ satisfying the properties aforementioned because every edge in $D$ is mended while the rest are left untouched. 

\textbf{Stage 2.} (\emph{$0,2$-edge mending.})
None of the conflicting edges between $W$ and $Y$ can be assigned $1$ in either of them.
In this stage, we simply change all remaining disagreeing edges from the value in $W$ to the value in $Y$ one by one according to the order of their indices.

We then show that the set of canonical paths $\Gamma$ constructed above fulfills \Cref{lem:gm-cp}. 
Assign weight $w_{\-{gm}}(\gamma)=\pi_{\-{gm}}(X)\pi_{\-{gm}}(Y)$ to the path $\gamma_{XY}$. 
The length (number of transitions) of each path $\gamma_{XY}$ is at most $m$, because each edge is mended at most once. 

We first prove (\ref{equ:capacity-wind}). 
Let $(Z,Z')$ be a transition with $\left|\left\{e:Z(e)\neq Z'(e)\right\}\right| = 1$, where the only edge $e$ of discrepancy is assigned $1$ in either $Z$ or $Z'$. 
Note that $(Z,Z')$ will only be used by any path in its first stage described above. 
Define a mapping $\varphi_{Z,Z'}:\Omega\times\Omega\to\Omega$ over any pair of configurations $X,Y$ whose corresponding path $\gamma_{XY}$ uses the transition $(Z,Z')$ by
\begin{equation}
  \varphi_{Z,Z'}(X,Y)=U\text{\quad where \quad} U(e)=X(e)+Y(e)-Z(e), \forall e\in E(G). 
\end{equation}
We claim that $\varphi_{Z,Z'}$ is an injection. 
Given $U$ and $Z$, we can recover $X(e)+Y(e)$ for any edge $e$. 
First we can find $D$, the set of conflicting $1$-edge in Stage 1, as it is simply $\{e:X(e)+Y(e)=1\text{ or }3\}$.
This gives rise to the unique edge-disjoint decomposition $P_1,\cdots,P_k,C_1,\cdots,C_{k'}$. 
By looking at $Z$ and $Z'$, we know the edge that is currently being wound, and, together with the edge-disjoint decomposition, the stage of the whole winding process. 
Therefore, we can continue the winding from $Z'$ with these information, and when finished, $W$ (defined in the process Stage 1) is obtained. 
To further recover $Y$, note that $e$ gets mended in Stage 2 if any only if $U(e)+Z(e)=2$ and $Z(e)\neq 1$. 
This follows from the fact that $Z(e)$ (in the first stage) is in line with $X(e)$ so long as $Z(e)\neq 1$. 
Therefore, we can decide all such edges and mend the assignment to obtain $Y$. 
To get $X$, we just reverse the operations backwards from $Z$. 

Given this injection, we compute $\sum_{\gamma\in\Gamma:(Z,Z')\in\gamma}w_{\-{gm}}(\gamma)$. The goal here is to bound the following ratio
\begin{equation} \label{equ:discrepancy-ratio}
\frac{\pi_{\-{gm}}(X)\pi_{\-{gm}}(Y)}{\pi_{\-{gm}}(U)\pi_{\-{gm}}(Z)},\text{\quad or equivalently,\quad} \frac{\wt_{\-{gm}}(X)\wt_{\-{gm}}(Y)}{\wt_{\-{gm}}(U)\wt_{\-{gm}}(Z)}.
\end{equation}
Recall that this ratio may contain two kinds of factors, emerging from both the vertices and edges.
For the factor from edges, the construction of $U$ ensures that 
(1) if $X(e)+Y(e)=U(e)+Z(e)\in\{0,1,3,4\}$, or $X(e)+Y(e)=2$ and $X(e)\neq 1$, then it must holds that either $X(e)=U(e)$ and $Y(e)=Z(e)$, or $X(e)=Z(e)$ and $Y(e)=U(e)$; 
(2) if $X(e)=Y(e)=1$, then $e$ never gets mended throughout the canonical path, and hence $Z(e)=U(e)=1$. 
In either case, all the terms rising from the edges in the numerator and denominator cancel. 
The terms rising from the vertices come from those in $\+O(X),\+O(Y),\+O(U),\+O(Z)$. 
It is not hard to see that the ones that do not get cancelled only arise from the current cycle or path that is being processed, and more specifically, the vertex incident to the two edges wound before and after $Z$, which contributes twice, and the starting vertex of the current cycle, which contributes twice as well. Therefore, 
\begin{equation} \label{equ:discrepancy_weight}
  \frac{\pi_{\-{gm}}(X)\pi_{\-{gm}}(Y)}{\pi_{\-{gm}}(U)\pi_{\-{gm}}(Z)}\leq \eta_{\min}^{-4},
\end{equation}
as $0<\eta_v<1$ for all $v$. 

Then, (\ref{equ:capacity-wind}) follows from (\ref{equ:discrepancy_weight}) that
\begin{align*}
  \sum_{\gamma\in\Gamma:(Z,Z')\in\gamma}w_{\-{gm}}(\gamma)&=\sum_{X,Y:(Z,Z')\in\gamma_{XY}}\pi_{\-{gm}}(X)\pi_{\-{gm}}(Y) \tag{By definition}\\
  &\leq \eta_{\min}^{-4}\sum_{X,Y:(Z,Z')\in\gamma_{XY}}\pi_{\-{gm}}(Z)\pi_{\-{gm}}(\varphi_{Z,Z'}(X,Y)) \tag{By (\ref{equ:discrepancy_weight})}\\
  &\leq \eta_{\min}^{-4}\pi_{\-{gm}}(Z). \tag{$\varphi_{Z,Z'}$ is injective}
\end{align*}

We construct the other mapping $\varphi'_{Z,Z'}(X,Y)$ by taking $\varphi_{Z,Z'}(X,Y)(e)=X(e)+Y(e)-Z'(e)$. 
The same argument shows that $\sum_{\gamma\in\Gamma:(Z,Z')\in\gamma}w_{\-{gm}}(\gamma)\leq\eta_{\min}^{-4}\pi_{\-{gm}}(Z')$. 

To prove (\ref{equ:capacity-other}), we look at the transition step $(Z,Z')$ with $\left|\left\{e:Z(e)\neq Z'(e)\right\}\right| = 1$ where the only edge $e$ of discrepancy is assigned $1$ in neither $Z$ nor $Z'$. 
We use the same mapping $\varphi_{Z,Z'}(X,Y)$ as above, and claim it is still injective in this case. 
Recall that $e$ gets mended in Stage 2 if and only if $U(e)+Z(e)=2$ and $Z(e)\neq 1$, and we can again determine the edges to be mended in Stage 2. 
Moreover, by looking at the difference of $Z$ and $Z'$, we know the index of the edge being mended, and therefore we can continue this process manually according to the instruction of Stage 2, knowing which edges to mend, to obtain $Y$. 
To get $X$, we first go backwards from $Z$ to the beginning of Stage 2 to obtain $W$, and revert the whole Stage 1 using the same argument aforementioned. 

To show (\ref{equ:capacity-other}), note that the edge factors in the ratio of (\ref{equ:discrepancy-ratio}) again cancel, and because no edge with assignment $1$ is involved, the vertex factors cancel as well. Hence the ratio is exactly $1$, and (\ref{equ:capacity-other}) follows according to the same calculation. 
\end{proof}

\subsection{Total congestion and rapid mixing}

We next bound the total congestion for $\Gamma_{\-{gm}}$. For each transition $(Z,Z')$ such that $\left|\left\{e:Z(e)\neq Z'(e)\right\}\right| = 1$, where the only edge of discrepancy is assigned $1$ in either $Z$ or $Z'$, we have
%\begin{align*}
%&\frac{L}{\pi_{\-{gm}}(Z)P_{\-{gm}}(Z,Z')}\sum_{\substack{\gamma\in\Gamma:\\(Z,Z')\in\gamma}}w_{\-{gm}}(\gamma)\leq\frac{m\eta_{\min}^{-4}\min\{\pi_{\-{gm}}(Z),\pi_{\-{gm}}(Z')\}}{\pi_{\-{gm}}(Z)P_{\-{gm}}(Z,Z')}\tag{\Cref{lem:gm-cp}}\\
%\leq&\frac{4m^2\eta_{\min}^{-4}\min\{\pi_{\-{gm}}(Z),\pi_{\-{gm}}(Z')\}}{\pi_{\-{gm}}(Z)\min\left\{1,\frac{\pi_{\-{gm}}(Z')}{\pi_{\-{gm}}(Z)}\right\}}=4m^2\eta_{\min}^{-4}. 
%\end{align*}
\begin{equation*}
  \frac{L}{\pi_{\-{gm}}(Z)P_{\-{gm}}(Z,Z')}\sum_{\substack{\gamma\in\Gamma:\\(Z,Z')\in\gamma}}w_{\-{gm}}(\gamma)\leq\frac{m\eta_{\min}^{-4}\min\{\pi_{\-{gm}}(Z),\pi_{\-{gm}}(Z')\}}{\pi_{\-{gm}}(Z)P_{\-{gm}}(Z,Z')}=:(\spadesuit)
\end{equation*}
by \Cref{lem:gm-cp}. 
To continue the calculation, there are several cases $(Z(e),Z'(e))=(0,1),(2,1),(1,0),(1,2)$. 
Below we only prove the case $(Z(e),Z'(e))=(0,1)$. 
The rest cases can be argued the same way and yield the same bound. 
Let $e=(u,v)$. There are some more subcases, depending on if $u$ or $v$ is in $\+O(Z)$. 
\begin{itemize}
  \item $u,v\notin\+O(Z)$. In this case, setting the edge to $1$ leads to extra factors from both vertices in $Z'$. Cancelling all the edges and vertices not involved, we obtain
  \begin{align*}
    (\spadesuit)=\frac{m^2\eta_{\min}^{-4}\min\{1-2p_e,p_e\eta_u\eta_v\}}{(1-2p_e)\frac{p_e\eta_u\eta_v}{(1-2p_e)+(p_e\eta_u\eta_v)+p_e}}\leq\frac{m^2\eta_{\min}^{-4}\min\{1-2p_e,p_e\eta_u\eta_v\}}{(1-2p_e)(p_e\eta_u\eta_v)}\leq\frac{m^2\eta_{\min}^{-4}}{1-2p_e}
  \end{align*}
  where we use the fact that $\eta_u,\eta_v\leq 1$. 
  \item $u,v\in\+O(Z)$. In this case, setting the edge to $1$ removes the factors from both vertices in $Z'$. Cancelling all the edges and vertices not involved, we obtain
  \begin{align*}
    (\spadesuit)=\frac{m^2\eta_{\min}^{-4}\min\{(1-2p_e)\eta_u\eta_v,p_e\}}{(1-2p_e)\eta_u\eta_v\frac{p_e}{(1-2p_e)\eta_u\eta_v+p_e+p_e\eta_u\eta_v}}\leq\frac{m^2\eta_{\min}^{-4}\min\{(1-2p_e)\eta_u\eta_v,p_e\}}{(1-2p_e)\eta_u\eta_v p_e}\leq\frac{m^2\eta_{\min}^{-4}}{p_e}
  \end{align*}
  where we use the fact that $\eta_u,\eta_v\leq 1$ again. 
  \item WLOG suppose $u\in\+O(Z), v\notin\+O(Z)$. In this case, setting the edge to $1$ causes the vertex factor to switch. Cancelling all the edges and vertices not involved, we obtain
  \begin{align*}
    (\spadesuit)=\frac{m^2\eta_{\min}^{-4}\min\{(1-2p_e)\eta_u,p_e\eta_v\}}{(1-2p_e)\eta_u\frac{p_e\eta_v}{(1-2p_e)\eta_u+(p_e\eta_v)+(p_e\eta_u)}}.
  \end{align*}
  If $\eta_u<\eta_v$, then above becomes
  \begin{align*}
	\frac{m^2\eta_{\min}^{-4}\min\{(1-2p_e)\frac{\eta_u}{\eta_v},p_e\}}{(1-2p_e)\frac{\eta_u}{\eta_v}\frac{p_e}{(1-2p_e)\frac{\eta_u}{\eta_v}+p_e+p_e\frac{\eta_u}{\eta_v}}}\leq\frac{m^2\eta_{\min}^{-4}}{p_e}.
  \end{align*}
  Otherwise, it can be written as
  \begin{align*}
	\frac{m^2\eta_{\min}^{-4}\min\{(1-2p_e),p_e\frac{\eta_v}{\eta_u}\}}{(1-2p_e)\frac{p_e\frac{\eta_v}{\eta_u}}{(1-2p_e)+p_e\frac{\eta_v}{\eta_u}+p_e}}\leq\frac{m^2\eta_{\min}^{-4}}{1-2p_e}.
  \end{align*}
\end{itemize}

For each transition $(Z,Z')$ such that $\left|\left\{e:Z(e)\neq Z'(e)\right\}\right| = 1$, where the only edge of discrepancy is assigned $1$ in none of $Z$ or $Z'$, the calculation is similar as above but simpler. WLOG assume $Z(e)=0$ and $Z'(e)=2$. 
\begin{align*}
&\frac{L}{\pi_{\-{gm}}(Z)P_{\-{gm}}(Z,Z')}\sum_{\substack{\gamma\in\Gamma:\\(Z,Z')\in\gamma}}w_{\-{gm}}(\gamma)\leq\frac{m\min\{\pi_{\-{gm}}(Z),\pi_{\-{gm}}(Z')\}}{\pi_{\-{gm}}(Z)P_{\-{gm}}(Z,Z')}\tag{\Cref{lem:gm-cp}}\\
\leq&\frac{m^2\min\{1-2p_e,p_e\}}{(1-2p_e)\frac{p_e}{1-2p_e+p_e+p_e\frac{1}{\eta_u\eta_v}}}\leq\min\left\{\frac{1}{p_e},\frac{1}{1-2p_e}\right\}m^2\eta_{\min}^{-2}. \tag{Worst case of $\eta$ terms}
\end{align*}
There is no canonical path using the self loop $(Z,Z)$, so the congestion is zero. In all, the congestion is bounded by $m^2\eta_{\min}^{-4}\max\left\{\frac{1}{p_{\min}},\frac{1}{1-2p_{\max}}\right\}$, from which (\ref{equ:gm-spectral-gap}) follows.

\section{Entropy decay of Glauber dynamics on the grand model} \label{section-ent}
In \Cref{section-var}, we analysed the variance decay of Glauber dynamics on the grand model. We now continue to analyse its relative entropy decay.
Let $G=(V,E)$ be a graph, and $\*p = (p_e)_{e \in E}$ and $\*\eta = (\eta_v)_{v \in V}$ be the parameters, where $0 < p_e < 1/2$ for any $e\in E$ and $\eta_v>0$ for any $v\in V$.
Let $\pi_{\-{gm}}$ denote the distribution specified by the grand model with parameters $\*p$ and $\*\eta$.
Let $\Omega(\pi_{\-{gm}})$ denote the support of $\pi_{\-{gm}}$.
%Let $\pi_{\-{wrc}}$ denote the distribution specified by the random cluster model with parameters $2\*p$ and $\*\lambda$, where $\lambda_v = \frac{1-\eta_v}{1+\eta_v}$. 
We use $P_{\-{GlauberGM}}$ to denote Glauber dynamics on $\pi_{\-{gm}}$. 

\begin{lemma}\label{lemma-ent-decay}
  If $0 < \eta_v < 1$ for all $v\in V$,
  then for any distribution $\nu$ with support $\Omega(\nu) \subseteq \Omega(\pi_{\-{gm}})$, Glauber dynamics $P_{\-{GlauberGM}}$ satisfies
  \begin{align*}
    \KL{\nu P_{\-{GlauberGM}}^{\downarrow} }{\pi_{\-{gm}} P_{\-{GlauberGM}}^{\downarrow}} \leq \tp{1 - \frac{1}{Cn}}\KL{\nu}{\pi_{\-{gm}}},
  \end{align*}
  where $C = C(\Delta, \eta_{\min},p_{\min},p_{\max})$, $\eta_{\min} = \min_{v \in V}\eta_v$, $p_{\min} = \min_{e \in E}p_e$, $p_{\max} = \max_{e \in E}p_e$, $\Delta$ is the maximum degree of $G$ and $n = |V|$.
\end{lemma}

\begin{remark}
For interested readers, the constant $C$ in the lemma above can be taken as
\[
  C=\Delta \left(\frac{2\Delta}{\eta_{\min}^2\min\left\{1-2p_{\max},p_{\min}\right\}}\right)^{2+\frac{16\Delta^2}{\eta_{\min}^4\min\{1-2p_{\max},p_{\min}\}}}.
\]
\end{remark}

\Cref{lemma-ent-decay} is proved by \Cref{theorem-CLV}. To apply \Cref{theorem-CLV}, we need to verify (1) $\pi_{\-{gm}}$ is a Gibbs distribution with maximum degree $D = 2\Delta-1$; (2) $\pi_{\-{gm}}$ is $\ell_{\infty}$-spectrally independent; (3) $\pi_{\-{gm}}$ is marginally bounded. The rest of this section is dedicated to the proof of \Cref{lemma-ent-decay}.

%\begin{lemma}\label{lemma-I}
% $\pi_{\-{gm}}$ is a Gibbs distribution with maximum degree $D = O(\Delta)$.	
%\end{lemma}
%\begin{proof}
%Given the graph $G=(V,E)$, define hypergraph $H=(E,\+F)$, where each vertex in $H$ corresponds to an edge in $G$ and each hyperedge $f_v \in \+F$ corresponds to a vertex $v \in V$ such that $f_v =\{e \in E \mid v \in e\}$. For each $f_v \in \+F$, define $\phi_{f_v}: \{0,1,2\}^{f_v} \to \mathbb{R}_{> 0}$ such that for all $\sigma \in \{0,1,2\}^{f_v}$,
%\begin{align*}
%	\phi_{f_v}(\sigma) = \begin{cases}
%		\eta_v & \text{if } |\{e \in f_v\mid \sigma_e = 1\}| \text{ is odd};\\
%		1 & \text{if } |\{e \in f_v\mid \sigma_e = 1\}| \text{ is even}.
%	\end{cases} 
%\end{align*}
%For each $e \in E$, define $\phi_e:\{0,1,2\} \to \mathbb{R}_{>0}$ such that for all $c \in \{0,1,2\}$,
%\begin{align*}
%	\phi_e(c) = \begin{cases}
%		1-2p_e & \text{if } c = 0;\\
%		p_e & \text{if } c = 1 \lor c = 2.
%	\end{cases} 
%\end{align*}
%It is straightforward to verify that for all $X \in \{0,1,2\}^E$,  $\pi_{\-{gm}}(X) \propto \prod_{e \in E}\phi_e(X_e)\prod_{f \in \+F}\phi_f(X_f)$. Hence,  $\pi_{\-{gm}}$ is a Gibbs distribution. The maximum degree of Gibbs distribution $\pi_{\-{gm}}$ is $D = O(\Delta)$.
%\end{proof}

\begin{lemma}\label{lemma-II}
 $\pi_{\-{gm}}$ is $\ell_\infty$-spectrally independent with parameter $\zeta = O(\Delta^2/\eta_{\min}^2)$.
\end{lemma}
We need the following result in \cite{CLV21a} to prove \Cref{lemma-II}.
We view the subgraph world as a distribution over $\{0,1\}^E$, where each $Y \in \{0,1\}^E$ corresponds to $S = \{e \in E \mid Y_e = 1\}$.
\begin{lemma}[\text{\cite{CLV21a}}]\label{lemma-CLV-inf}
Let $G=(V,E)$ be a graph with the maximum degree $\Delta \geq 3$. 
Let $\*p=(p_e)_{e \in E}$ and $\*\eta = (\eta_v)_{v \in V}$, where $0 \leq p_e < 1/2$ and $0 < \eta_v \leq 1$. 
The distribution $\pi_{\-{sg}}$ specified by the subgraph-world model with parameters $(\*p,\*\eta)$ is $\ell_\infty$-spectrally independent with parameter $\zeta = O(\Delta^2/\eta_{\min}^2)$.	
\end{lemma}

\begin{remark}
In \cite{CLV21a}, the authors only formalise the proof for the uniform case (i.e., all $\eta_v$'s take the same value) while stating that the argument works for non-uniform case without a proof. This in fact holds true by going through the proof and taking the worst region of stability. The final spectral independence parameter is
\[
  \zeta=8\left(\frac{\left(\frac{1+\eta_{\min}}{1-\eta_{\min}}\right)^{1/\Delta}+1}{\left(\frac{1+\eta_{\min}}{1-\eta_{\min}}\right)^{1/\Delta}-1}\right)^2\sim 8\Delta^2/\eta_{\min}^2.
\]
Note that the $\lambda$ in their paper is actually $p/(1-p)$ in our formulation of the subgraph-world model (under the uniform edge parameter setting). Also note that we are only considering the region $0<p<1/2$, so the $\lambda$ in their paper is bounded from above by $1$. 
\end{remark}

\begin{proof}[Proof of \Cref{lemma-II}]
Fix a pinning $\sigma \in \{0,1,2\}^\Lambda$ for some $\Lambda \subseteq E$. 
According to the definition of the grand model, to draw $X \sim \pi_{\-{gm}}$, we first sample $Y \sim \pi_{\-{sg}}$ (where $Y \in \{0,1\}^E$ as we view $\pi_{\-{sg}}$ as a distribution over $\{0,1\}^E$), then flip independent coins for each $e \in E$ with $Y_e = 0$. Define the pinning $\tau \in \{0,1\}^\Lambda$ by $\tau_e = 1$ if $\sigma_e = 1$ and $\tau_e = 0$ if $\sigma_e = 0$ or $\sigma_e = 2$. Consider the influence 
\begin{align*}
	\Psi^\sigma_{\pi_{\-{gm}}}(e,f) = \max\set{\DTV{\pi_{\-{gm},f}^{\sigma \land e \gets 0}}{\pi_{\-{gm},f}^{\sigma \land e \gets 1}}, \DTV{\pi_{\-{gm},f}^{\sigma \land e \gets 0}}{\pi_{\-{gm},f}^{\sigma \land e \gets 2}},\DTV{\pi_{\-{gm},f}^{\sigma \land e \gets 1}}{\pi_{\-{gm},f}^{\sigma \land e \gets 2}} },
\end{align*}
where $e,f \in E \setminus \Lambda$ and $e \neq f$.
Since each coin flipping is independent with the random sample from $\pi_{\-{gm}}$, we can couple two distributions $\pi_{\-{gm},f}^{\sigma \land e \gets 0}$ and $\pi_{\-{gm},f}^{\sigma \land e \gets 1}$ as follows:
\begin{itemize}
	\item sample $Y_f,Y'_f$ from the optimal coupling between $\pi_{\-{sg},f}^{\tau\land e \gets 0}$ and $\pi_{\-{sg},f}^{\tau\land e \gets 1}$;
	\item flip a coin $\+C$ independently with probability of HEADS being $\frac{p_f}{1 - p_f}$;
	\item if $Y_f = 1$, let $X_f = 1$; otherwise, if the outcome of $\+C$ is HEADS, let $X_f = 2$, if the outcome of $\+C$ is not HEADS, let $X_f = 0$;
	\item if $Y_f' = 1$, let $X_f' = 1$; otherwise, if the outcome of $\+C$ is HEADS, let $X_f' = 2$, if the outcome of $\+C$ is not HEADS, let $X_f' = 0$;
\end{itemize}
It is straightforward to verify that $(X_f,X'_f)$ is sampled from a coupling between $\pi_{\-{gm},f}^{\sigma \land e \gets 0}$ and $\pi_{\-{gm},f}^{\sigma \land e \gets 1}$. 
By the coupling inequality \eqref{eqn:coupling-ineq} and as $Y_f$ and $Y'_f$ are optimally coupled, we have
\begin{align*}
\DTV{\pi_{\-{gm},f}^{\sigma \land e \gets 0}}{\pi_{\-{gm},f}^{\sigma \land e \gets 1}} \leq \Pr\left[X_f \neq X'_f\right] = 	\Pr\left[Y_f \neq Y'_f\right] = \DTV{\pi_{\-{sg},f}^{\tau \land e \gets 0}}{\pi_{\-{sg},f}^{\tau\land e \gets 1}}.
\end{align*}
Similarly, we have
\begin{align*}
\DTV{\pi_{\-{gm},f}^{\sigma \land e \gets 0}}{\pi_{\-{gm},f}^{\sigma \land e \gets 2}}	= 0 \quad\text{and}\quad \DTV{\pi_{\-{gm},f}^{\sigma \land e \gets 1}}{\pi_{\-{gm},f}^{\sigma \land e \gets 2}} \leq \DTV{\pi_{\-{sg},f}^{\tau \land e \gets 0}}{\pi_{\-{sg},f}^{\tau\land e \gets 1}}.
\end{align*}
Hence, by \Cref{lemma-CLV-inf},
\begin{align*}
 \norm{ \Psi^\sigma_{\-{gm}} }_\infty \leq 	 \norm{ \Psi^\tau_{\-{sg}} }_\infty \leq  \zeta. &\qedhere
\end{align*}
\end{proof}

\begin{lemma}\label{lemma-III}
 $\pi_{\-{gm}}$ is $b$-marginally bounded, where $b = \eta_{\min}^2\min\left\{1-2p_{\max},p_{\min}\right\}$.
\end{lemma}
\begin{proof}
Consider the marginal distribution of an edge $e=(u,v)$. 
Let $e_1,\dots, e_k$ be the edges adjacent to either $u$ or $v$ (but not both).
Suppose we have an arbitrary pinning $X$ on $\Lambda\subset E$ and $e\not\in \Lambda$.
Let $Y$ be an arbitrary pinning on $\Lambda\cup\{e_1,\dots,e_k\}$ that is consistent with $X$.
The true marginal of $e$ under $X$ is a linear combination of marginals conditioned on all possibilities of $Y$ (namely, we first sample $Y$ and then sample $e$ conditioned on $Y$).
Thus, to establish a lower bound, it suffices to establish a lower bound under any $Y$.
Given $Y$, the marginal of $e$ depends only on $p_e$ and whether $u$ or $v$ is in $\+O(Y)$. 
These cases are verified as follows. 
\begin{itemize}
  \item $u,v\notin\+O(Y^{e\to 0})$, where $Y^{e\to 0}$ is the configuration of $Y$ with $e$ further pinned to $0$. In this case the marginal is at least
  \[
	\frac{\min\{1-2p_e,p_e\eta_u\eta_v,p_e\}}{1-2p_e+p_e\eta_u\eta_v+p_e}\geq\min\{1-2p_e,p_e\eta_u\eta_v\}.
  \]
  Note that the denominator is no greater than $1$ because $\eta_u,\eta_v\leq 1$. 
  \item $u,v\in\+O(Y^{e\to 0})$. Then the marginal is at least
  \[
	\frac{\min\{1-2p_e\eta_u\eta_v,p_e,p_e\eta_u\eta_v\}}{(1-2p_e)\eta_u\eta_v+p_e+p_e\eta_u\eta_v}\geq\min\{(1-2p_e)\eta_u\eta_v,p_e\}.
  \]
  \item In the remaining cases, assume w.l.o.g.~$u\in\+O(Y^{e\to 0})$ while $v\notin\+O(Y^{e\to 0})$. Then the marginal is at least
  \[
	\frac{\min\{(1-2p_e)\eta_u,p_e\eta_v,p_e\eta_u\}}{(1-2p_e)\eta_u+p_e\eta_v+p_e\eta_u}=
	\begin{cases}
		\frac{\min\{(1-2p_e)\frac{\eta_u}{\eta_v},p_e\frac{\eta_u}{\eta_v}\}}{(1-2p_e)\frac{\eta_u}{\eta_v}+p_e+p_e\frac{\eta_u}{\eta_v}}\geq\min\{(1-2p_e)\frac{\eta_u}{\eta_v},p_e\frac{\eta_u}{\eta_v}\}, &\text{ if }\eta_u<\eta_v;\\
		\frac{\min\{(1-2p_e),p_e\frac{\eta_v}{\eta_u}\}}{(1-2p_e)+p_e\frac{\eta_v}{\eta_u}+p_e}\geq\min\{(1-2p_e),p_e\frac{\eta_v}{\eta_u}\}, &\text{otherwise}.
	\end{cases}
  \]
\end{itemize}
In all cases, the value
\[
  b=\eta_{\min}^2\min\left\{1-2p_{\max},p_{\min}\right\}
\]
suffices as a marginal lower bound. 
\end{proof}

\begin{proof}[Proof of \Cref{lemma-ent-decay}]
Combine \Cref{theorem-CLV}, \Cref{lemma-I}, \Cref{lemma-II}, \Cref{lemma-III} and $m\leq n \Delta$.
\end{proof}

\section{Rapid mixing of Glauber dynamics on the random cluster model}\label{sec:grand-RC}
Let $\*p = (p_e)_{e \in E}$ and $\*\eta = (\eta_v)_{v \in V}$, where $0 < p_e < 1/2$ and $0 < \eta_v < 1$.  
Let $\pi_{\-{wrc}}$ denote the distribution specified by the random cluster model with parameters $2\*p$ and $\*\lambda$, where $\lambda_v = \frac{1-\eta_v}{1+\eta_v}$.
Let $\Omega(\pi_{\-{wrc}})$ denote the support of $\pi_{\-{wrc}}$.
%Let $\pi_{\-{wrc}}$ denote the distribution specified by the random cluster model with parameters $2\*p$ and $\*\lambda$, where $\lambda_v = \frac{1-\eta_v}{1+\eta_v}$. 
We use $P_{\-{GlauberRC}}$ to denote Glauber dynamics on $\pi_{\-{wrc}}$. 

\begin{lemma}\label{lemma-decay-RC}
Let $\pi_{\-{wrc}}$ be the distribution specified by weighted random cluster model with parameters $(2\*p, \*\lambda)$.
The Glauber dynamics $P_{\-{GlauberRC}}$  satisfies that for any distribution $\nu$ with support $\Omega(\nu) \subseteq \Omega(\pi_{\-{wrc}})$,
\begin{itemize}
	\item $\chisq{\nu P_{\-{GlauberRC}}^{\downarrow} }{\pi_{\-{wrc}} P_{\-{GlauberRC}}^{\downarrow} } \leq \tp{1 - \frac{\alpha}{m^2}}\chisq{\nu}{\pi_{\-{wrc}}}$,
	\item $\KL{\nu P_{\-{GlauberRC}}^{\downarrow} }{\pi_{\-{wrc}} P_{\-{GlauberRC}}^{\downarrow}} \leq \tp{1 - \frac{1}{Cn}}\KL{\nu}{\pi_{\-{wrc}}}$,
\end{itemize}
where 
\begin{align*}
	\alpha &= \tp{\frac{1-\lambda_{\max}}{1+\lambda_{\max}}}^4\min\left\{p_{\min},1-2p_{\max}\right\},\\
	 C&=\Delta \left(\frac{8\Delta}{(1-\lambda_{\max})^2\min\left\{1-2p_{\max},p_{\min}\right\}}\right)^{2+\frac{256\Delta^2}{(1-\lambda_{\max})^4\min\{1-2p_{\max},p_{\min}\}}},
\end{align*}
$\lambda_{\max}=\max_{v \in V}\lambda_v$ , $\lambda_{\min} = \min_{v \in V}\lambda_v$, $p_{\max} = \max_{e \in E}p_e$, $p_{\min} = \min_{e \in E}p_e$, $\Delta$ is the maximum degree of $G$, $n = |V|$ and $m = |E|$.
\end{lemma}

\Cref{lemma-decay-RC} projects the decay results (\Cref{lemma-var-decay} and \Cref{lemma-ent-decay}) from the grand model to the random cluster model.
\Cref{lemma-decay-RC} is proved by a comparison lemma in \Cref{section-compare-down-walk} that works for general projections and $f$-divergences.

Although the entropy decay rate is seemingly better than the $\chi^2$-divergence decay rate,
we still need the latter in the case where $\lambda_{\max}$ is $1$ or close to $1$.
Note that as $\lambda_{\max}$ goes to $1$, $\alpha\rightarrow 0$ and $C\rightarrow \infty$, making both statements meaningless.
However, we may perturb $\lambda$ by a factor of $1/n$ in such cases.
This incurs a cost of some extra factors of $n$ in case of $\alpha$ and an exponentially large factor in case of $C$.
Thus, the $\chi^2$-divergence decay rate in \Cref{lemma-decay-RC} is still useful in this case.
Specifically, in \Cref{section-perturbed} we showed the following.
%One issue with \Cref{lemma-decay-RC}  is the dependency of mixing time on $\alpha = \tp{\frac{1-\lambda_{\max}}{1+\lambda_{\max}}}^4\min\left\{p_{\min},1-2p_{\max}\right\}$. 
%Especially, when some $\lambda_v$ is very close to $1$, say $1-e^{-n}$, or even simply $1$, the decay rate is exponentially slow in $n$.
%Consider the random cluster model without external fields on vertices, i.e. $\lambda_v = 1$ for all $v \in V$. The decay rate is 0.
%We eliminate such extreme cases via perturbation.
%This incurs a cost of some extra factors of $n$. %The main lemma of this subsection is stated as follows. 
\begin{lemma}\label{lemma-decay-RC-perturbed}
Let $\pi_{\-{wrc}}$ be the distribution specified by the weighted random cluster model with parameters $(2\*p, \*\lambda)$.
The Glauber dynamics $P_{\-{GlauberRC}}$  satisfies that for any distribution $\nu$ with support $\Omega(\nu) \subseteq \Omega(\pi_{\-{wrc}})$,		
\begin{align*}
\chisq{\nu P_{\-{GlauberRC}}^{\downarrow} }{\pi_{\-{wrc}} P_{\-{GlauberRC}}^{\downarrow} } \leq \tp{1 - \frac{\min\left\{p_{\min},1-2p_{\max}\right\}}{10^4 n^4 m^2}}\chisq{\nu}{\pi_{\-{wrc}}}.
\end{align*}
\end{lemma}
%The core of the refined analysis is to uli perturbed chains. 
%\Cref{lemma-decay-RC-perturbed} is proved in \Cref{section-perturbed}.

We remark that both \Cref{lemma-decay-RC} and \Cref{lemma-decay-RC-perturbed} consider the random cluster model specified by parameters $(2\*p, \*\lambda)$. 
Combining  \Cref{lemma-decay-RC} and \Cref{lemma-decay-RC-perturbed}, we have the following mixing result for the Glauber dynamics on random cluster model.

\begin{theorem}\label{thm:RC}
Let $G=(V,E)$ be a $n$-vertex and $m$-edge graph with maximum degree $\Delta$.
Let $\*p = (p_e)_{e \in E}$ and $\*\lambda = (\lambda_v)_{v \in V}$, where $0 < p_e < 1$ and $0 < \lambda_v \leq 1$.
Let $\pi_{\-{wrc}}$ be the distribution specified by the random cluster model with parameters $(\*p, \*\lambda)$. 
The mixing of Glauber dynamics $P_{\-{GlauberRC}}$ on $\pi_{\-{wrc}}$ satisfies 
\begin{align*}
	T_{\-{mix}}(P_{\-{GlauberRC}},\epsilon) \leq C_1(p_{\min},p_{\max})  \cdot \min\set{n^4, \tp{\frac{1}{1-\lambda_{\max}}}^4} \cdot m^2 \cdot \tp{\log \frac{1}{\epsilon} + m },
\end{align*}
where $C_1(p_{\min},p_{\max}) = O\left(\frac{1}{\min\{p_{\min }, 1 - p_{\max }\}} \log \frac{1}{\min\{p_{\min }, 1 - p_{\max }\}}\right)$.

Furthermore, if there exists $\delta > 0$ such that $\lambda_v \leq 1-\delta$ for all $v \in V$, then the mixing time satisfies 
\begin{align*}
T_{\-{mix}}(P_{\-{GlauberRC}},\epsilon) \leq C_2(\Delta, \delta,p_{\min }, p_{\max } ) \cdot n \tp{\log n + \log \frac{1}{\epsilon} },	
\end{align*}
where $C_2(\Delta, \delta,p_{\min }, p_{\max } ) = \left(\frac{\Delta}{\delta^2 \min\{p_{\min},1-p_{\max}\}}\right)^{O\left(\frac{\Delta^2}{\delta^4 \min\{p_{\min},1-p_{\max}\}}\right)}$.
\end{theorem}
\begin{proof}
  Let $\pi_{\-{wrc}, \min} = \min_{S \subseteq E}\pi_{\-{wrc}}(S)$ denote the minimum probability in $\pi_{\-{wrc}}$. It is straightforward to verify that $\pi_{\-{wrc},\min} \geq \{p_{\min},1-p_{\max}\}^m/2^{m+n}$. By the data processing inequality,
\begin{align*}
	\DF{\nu P_{\-{GlauberRC}} }{ \pi_{\-{wrc}} } = \DF{\nu P_{\-{GlauberRC}} }{ \pi_{\-{wrc}} P_{\-{GlauberRC}}}  \leq \DF{\nu P_{\-{GlauberRC}}^{\downarrow} }{ \pi_{\-{wrc}} P_{\-{GlauberRC}}^{\downarrow}}.
\end{align*}
By \Cref{lemma-decay-RC} and \Cref{lemma-decay-RC-perturbed}, we know that after each transition step of Glauber dynamics, the $\chi^2$-divergence and KL-divergence between the current distribution of the stationary distribution decays by factors specified earlier. 
The $\chi^2$-divergence between the initial distribution and the stationary distribution is at most $\frac{1}{\pi_{\-{wrc}, \min}}$, and the KL-divergence is at most $\log \frac{1}{\pi_{\-{wrc}, \min}}$. 
By \Cref{lemma-decay-RC}, \Cref{lemma-decay-RC-perturbed}, and ~\eqref{eqn:chisq}, 
\begin{align*}
T_{\-{mix}}(P_{\-{GlauberRC}},\epsilon) &\leq \frac{10^4}{\min\left\{p_{\min}/2,1-p_{\max}\right\}} \cdot \min\set{n^4, \tp{\frac{1}{1-\lambda_{\max}}}^4} \cdot m^2 \tp{\log \frac{1}{\epsilon^2\pi_{\-{wrc}, \min}}  }\\
&=C_1(p_{\min},p_{\max})  \cdot \min\set{n^4, \tp{\frac{1}{1-\lambda_{\max}}}^4} \cdot m^2 \cdot \tp{\log \frac{1}{\epsilon} + m }.
\end{align*}
By \Cref{lemma-decay-RC}, \eqref{eqn:Pinsker} and $m \leq \Delta n$, if for all $\lambda_v \leq 1 - \delta$, then we have $1-\lambda_{\max} \geq \delta$ and
\begin{align*}
T_{\-{mix}}(P_{\-{GlauberRC}},\epsilon) &\leq  	\Delta \left(\frac{8\Delta}{\delta^2\min\left\{1-p_{\max},p_{\min}/2\right\}}\right)^{2+\frac{256\Delta^2}{\delta^4\min\{1-p_{\max},p_{\min}/2\}}} \cdot n \tp{\log \log\frac{1}{\pi_{\-{wrc}, \min}} + \log \frac{1}{2\epsilon^2}  }\\
&=C_2(\Delta, \delta,p_{\min }, p_{\max } ) \cdot n \tp{\log n + \log \frac{1}{\epsilon} }. \qedhere
\end{align*}
\end{proof}

\subsection{Comparing the decay rates of down walks}\label{section-compare-down-walk}
Here we consider a general projection from a larger state space to a smaller one.
Let $Q$ and $R$ be two finite sets, 
and let $\Omega \subseteq Q^V$ be the state space.
Consider a mapping $g:Q \to R$. 
(Note that here we can restrict $R$ to the range of $g$ without changing the rest of the argument. In other words, after the mapping the effective domain is never larger than $Q$,
although we do not need to require $\abs{Q}\ge\abs{R}$ a priori.)
Given any $\sigma \in \Omega$, we map $\sigma$ to $\tau = (\tau_v)_{v \in V}$, where $\tau_v = g(\sigma_v)$.
We abuse the notation and denote $\tau = g(\sigma)$.
Let $\Omega' = \{g(\sigma) \mid \sigma \in \Omega\} \subseteq R^V $. Define the projection matrix $P:\Omega \times \Omega' \to \{0,1\}$:
\begin{align*}
  \forall \sigma \in \Omega, \tau \in \Omega',\quad P(\sigma,\tau) = \mathbb{I}[\tau = g(\sigma)].
\end{align*}
We remark that $P$ is a stochastic matrix.

Let $\pi$ be a distribution with support $\Omega$. 
Define the distribution $\mu = \pi P$ with support $\Omega'$.
Let $P_{\-{Glauber},\pi}^\downarrow:\Omega\times\Omega_{\-{down}}\to \mathbb{R}_{\geq 0}$ denote the down walk of Glauber dynamics on $\pi$, where $\Omega_{\-{down}} = \{\sigma_{V \setminus \{v\}} \mid v \in V \land \sigma \in \Omega\}$.
Given any configuration $\sigma \in \Omega$, $P_{\-{Glauber},\pi}^\downarrow$ picks a variable $v \in V$ uniformly at random, and then transforms $\sigma$ to  $\sigma_{V \setminus \{v\}}$ by dropping the value of $v$.
Similarly, let $P_{\-{Glauber},\mu}^\downarrow$ denote the down walk of Glauber dynamics on the distribution $\mu=\pi P$. 

\begin{lemma}\label{lemma-compare}
Let $0 < \delta < 1$.
Let $f:\mathbb{R}_{\geq 0} \to \mathbb{R}$ be a convex function with $f(1) = 0$.
If $P_{\-{Glauber},\pi}^\downarrow$ satisfies that for any distribution $\nu$ with support $\Omega$, 
\begin{align*}
	\DF{\nu P_{\-{Glauber},\pi}^\downarrow}{\pi P_{\-{Glauber},\pi}^\downarrow} \leq (1-\delta)\DF{\nu}{\pi},
\end{align*}
then $P_{\-{Glauber},\mu}^\downarrow$ satisfies that for any distribution $\varphi$ with support $\Omega'$, 
\begin{align*}
	\DF{\varphi P_{\-{Glauber},\mu}^\downarrow}{\mu P_{\-{Glauber},\mu}^\downarrow} \leq (1-\delta)\DF{\varphi}{\mu}.
\end{align*}
\end{lemma}
\begin{proof}
Given any $\rho \in \Omega_{\-{down}}$, we can map $\rho$ to $\eta = g(\rho)$, where $\eta_u = g(\rho_u)$ for any variable $u$. Let $\Omega_{\-{down}}' = \{g(\rho) \mid \rho \in \Omega_{\-{down}}\}$. 
Define the projection matrix $P':\Omega_{\-{down}} \times \Omega_{\-{down}}' \to \{0,1\}$:
\begin{align*}
  \forall \rho \in \Omega_{\-{down}}, \eta \in \Omega_{\-{down}}',\quad P'(\rho,\eta) = \mathbb{I}[\eta = g(\rho)].
\end{align*}
We remark that $P'$  is a stochastic matrix.
Since both $P$ and $P'$ project the value of each variable independently, the following equation is straightforward to verify 
\begin{align}\label{eq-order}
P_{\-{Glauber},\pi}^\downarrow \cdot P' = P \cdot P_{\-{Glauber},\mu}^\downarrow.
\end{align}	
	
For any configuration $\tau \in \Omega'$, define the distribution $\pi^\tau$ over $\Omega$ by 
\begin{align*}
  \forall \sigma \in \Omega, \quad \pi^{\tau}(\sigma) = \frac{\mathbb{I}[g(\sigma)=\tau]\pi(\sigma)}{\mu(\tau)}.
\end{align*}
For any $\sigma \in \Omega$, let $\tau = g(\sigma)$, it holds that $\pi(\sigma) = \mu(\tau)\pi^\tau(\sigma)$.
Fix a distribution $\varphi$ with support $\Omega'$. Define the distribution $\nu$ by 
\begin{align}\label{eq-def-nu}
	\forall \sigma \in \Omega, \quad \nu(\sigma) = \varphi(\tau)\pi^{\tau}(\sigma), \quad\text{where } \tau = g(\sigma).
\end{align}
We have
\begin{align}\label{eq-equal-df}
	\DF{\nu}{\pi} = \Ex_{\sigma \sim \pi} \left[ f\tp{\frac{\nu(\sigma)}{\pi(\sigma)}} \right] = \Ex_{\tau \sim \mu}\Ex_{\sigma \sim \pi^\tau} \left[ f\tp{\frac{ \varphi(\tau)\pi^{\tau}(\sigma)}{\mu(\tau)\pi^\tau(\sigma)}} \right] = \Ex_{\tau \sim \mu} \left[ f\tp{\frac{ \varphi(\tau)}{\mu(\tau)}} \right] = \DF{\varphi}{\mu}.
\end{align}
By the definition in~\eqref{eq-def-nu}, we have for all $\tau \in \Omega'$,
\begin{align*}
(\nu P)(\tau) = \sum_{\sigma: g(\sigma) = \tau} \nu(\sigma) = \varphi(\tau)	\sum_{\sigma: g(\sigma) = \tau}\pi^\tau(\sigma) = \varphi(\tau),
\end{align*}
which implies $\varphi = \nu P$.
Recall that $\mu = \pi P$. 
We have
\begin{align*}
	\DF{\varphi P_{\-{Glauber},\mu}^\downarrow}{\mu P_{\-{Glauber},\mu}^\downarrow} &=  \DF{\nu P P_{\-{Glauber},\mu}^\downarrow}{\pi P P_{\-{Glauber},\mu}^\downarrow}\\
\text{(by~\eqref{eq-order})}\quad &= \DF{\nu  P_{\-{Glauber},\pi}^\downarrow P'}{\pi  P_{\-{Glauber},\pi}^\downarrow P'}\\
\text{(by data processing inequality)}\quad &\leq \DF{\nu  P_{\-{Glauber},\pi}^\downarrow  }{\pi  P_{\-{Glauber},\pi}^\downarrow }\\
\text{(by assumption)}\quad  &\leq (1-\delta)\DF{\nu  }{\pi   }\\
\text{(by~\eqref{eq-equal-df})}\quad &= (1-\delta)\DF{\varphi}{\mu}. \qedhere
\end{align*}
%This proves the lemma.
\end{proof}

We are now ready to prove \Cref{lemma-decay-RC}.
\begin{proof}[Proof of \Cref{lemma-decay-RC}]
Let $\Omega = \{0,1,2\}^E$ denote the support of $\pi_{\-{gm}}$. 
Define the map $g$ by $g(0) =0$, $g(1) = 1$ and $g(2) = 1$. 
By \Cref{lemma-marginal}, it holds that $\pi_{\-{wrc}} = \pi_{\-{gm}} P$. 
\Cref{lemma-decay-RC} follows from \Cref{lemma-var-decay}, \Cref{lemma-ent-decay} and \Cref{lemma-compare}.
\end{proof}

\subsection{Faster mixing via perturbed chains}\label{section-perturbed}
%The core of the refined analysis is to cap $\lambda_v$ from the above, or equivalently for the subgraph-world model, to cap $\eta_v$ from the below. 
Given a subgraph-world model $(G;\*p,\*\eta)$, we define the ``perturbed'' model $(G;\*p,\widehat{\*\eta})$ by
\begin{equation} \label{equ:perturb}
\widehat{\eta}_v=\begin{cases}
  \frac{1}{n}, &\text{ if }0\leq \eta_v\leq \frac{1}{n}\\
  \eta_v, &\text{otherwise.}
\end{cases}
\end{equation}
Call the induced distribution $\widehat{\pi_{\text{sg}}}$. Take a random subgraph $\+S$ according to $\widehat{\pi_{\text{sg}}}$, and add each remaining edge $e\in E\setminus\+S$ with probability $p_e/(1-p_e)$ to obtain $\+R$. By \Cref{lem:weighted-coupling}, the resulting distribution is $\pi_{\-{wrc}}(G;2\*p,\widehat{\*\lambda})=:\widehat{\pi_{\-{wrc}}}$, where $\widehat{\lambda}_v=\frac{1-\widehat{\eta}_{v}}{1+\widehat{\eta}_{v}}$.  
Let $\widehat{P_{\-{wrc}}}$ denote the Glauber dynamics on $\widehat{\pi_{\-{wrc}}}$.
Let $\widehat{P_{\-{wrc}}}^\downarrow$ denote the down-walk of $\widehat{P_{\-{wrc}}}$.
Applying the first item of \Cref{lemma-decay-RC} to the perturbed random-cluster model $(G;2\*p,\widehat{\*\lambda})$ yields that for any distribution $\nu$,
\begin{align*}
\chisq{\nu \widehat{P_{\-{wrc}}}^\downarrow }{\widehat{\pi_{\-{wrc}}} \widehat{P_{\-{wrc}}}^\downarrow } \leq \tp{1 - \frac{\min\left\{p_{\min},1-2p_{\max}\right\}}{m^2n^4}}\chisq{\nu}{\widehat{\pi_{\-{wrc}}}}	
\end{align*}
By \Cref{theorem-gap-decay}, we know that 
\begin{align*}
	\Gap(\widehat{P_{\-{wrc}}}) \geq \frac{\min\left\{p_{\min},1-2p_{\max}\right\}}{m^2n^4}.
\end{align*}
Based on this, the main effort of this subsection is to bound the spectral gap of the original model $(G;2\*p,\*\lambda)$ via the bounds for $(G;2\*p,\widehat{\*\lambda})$. 

We start with comparing the two distributions.
\begin{lemma} \label{lem:comp_dist_ratio_bound}
  For any $R\subseteq E$, 
  \[
    \frac{1}{9}\leq\frac{\widehat{\pi_{\-{wrc}}}(R)}{\pi_{\-{wrc}}(R)}<\mathrm{e}.
  \]
\end{lemma}

\begin{proof}
  Let $n=|V|$.
  If $n=1$, the only possible $R$ is $\emptyset$ and the lemma holds.
  We assume $n\ge 2$ in the rest.
To prove the first inequality, 
\begin{align*}
  \frac{\widehat{\pi_{\-{wrc}}}(R)}{\pi_{\-{wrc}}(R)}&=\frac{Z_{\-{wrc}}}{\widehat{Z_{\-{wrc}}}}\cdot\frac{\widehat{\wt_{\-{wrc}}}(R)}{\wt_{\-{wrc}}(R)}=\frac{Z_{\-{wrc}}}{\widehat{Z_{\-{wrc}}}}\cdot\prod_{C\in\kappa(V,S)}\frac{1+\prod_{u\in C}\widehat{\lambda}_u}{1+\prod_{u\in C}\lambda_u}. 
\end{align*}
Note that $\frac{Z_{\-{wrc}}}{\widehat{Z_{\-{wrc}}}}\geq 1$ 
because $\widehat{\lambda}_u\le \lambda_u$, which implies that the weight of each configuration of the random cluster model decreases after replacing $\*\lambda$ with $\widehat{\*\lambda}$. 
The second term can be handled by
\begin{align*}
  \prod_{C\in\kappa(V,S)}\frac{1+\prod_{u\in C}\widehat{\lambda}_u}{1+\prod_{u\in C}\lambda_u}\geq\prod_{C\in\kappa(V,S)}\frac{\prod_{u\in C}\widehat{\lambda}_u}{\prod_{u\in C}\lambda_u}\geq\left(\frac{n-1}{n+1}\right)^n\geq\frac{1}{9}
\end{align*}
as $n\geq 2$. 

For the second inequality,
the definition of $\pi_{\-{wrc}}$, together with the relation between $Z_{\-{wrc}}$ and $Z_{\text{sg}}$ in \Cref{equ:three_equivalence}, gives
\[
  \frac{\widehat{\pi_{\-{wrc}}}(R)}{\pi_{\-{wrc}}(R)}=\frac{Z_{\text{sg}}(G;\*p,\*\eta)}{Z_{\text{sg}}(G;\*p,\widehat{\*\eta})}\cdot\frac{\prod_{v \in V}\frac{1}{1+\widehat{\lambda}_v}}{\prod_{v \in V}\frac{1}{1+\lambda_v}}\cdot\frac{\prod_{C\in\kappa(V,R)}\left(1+\prod_{u\in C}\widehat{\lambda}_u\right)}{\prod_{C\in\kappa(V,R)}\left(1+\prod_{u\in C}\lambda_u\right)}. 
\]
There are three terms. For the first one, note that $\widehat{\eta}_v>\eta_v$ for all $v$, indicating that the weight of each configuration of the subgraph-world model is increased after replacing $\*\eta$ with $\widehat{\*\eta}$. As such, it is less or equal than $1$. The third term is also less or equal than $1$ due to $\widehat{\lambda}_v<\lambda_v$. The second term can be bounded by
\[
  \frac{\prod_{v \in V}\frac{1}{1+\widehat{\lambda}_v}}{\prod_{v \in V}\frac{1}{1+\lambda_v}}=\frac{\prod_{v \in V}(1+\widehat{\eta}_v)}{\prod_{v \in V}(1+\eta_v)}\leq\left(1+\frac{1}{n}\right)^n<\mathrm{e}
\]
which concludes this lemma. 
\end{proof}

We also have a bound on the ratio of the transition probability between the original and perturbed model in the Glauber dynamics. 

\begin{lemma} \label{lem:comp_tran_ratio_bound}
  Let $P_{\-{wrc}}$ and $\widehat{P_{\-{wrc}}}$ be the transition matrices of Glauber dynamics on the random cluster models $(G;2\*p,\*\lambda)$ and $(G;2\*p,\widehat{\*\lambda})$ respectively. Then it holds that
  \[
     \frac{1}{9\-{e}}\leq\frac{\widehat{P_{\-{wrc}}}(Z,Z')}{P_{\-{wrc}}(Z,Z')}\leq 9\mathrm{e} \qquad\text{for all }|Z\oplus Z'| = 1.
  \]
\end{lemma}

\begin{proof}
%The case $Z=Z'$ immediately follows from the chain being lazy with probability at least $1/2$. 
Assume $Z'=Z+e$ where $e\notin Z$. The case $Z'=Z-e$ where $e\in Z$ follows by a similar argument. We then have
\begin{align*}
 \frac{1}{9\-{e}} \leq \frac{\widehat{P_{\-{wrc}}}(Z,Z')}{P_{\-{wrc}}(Z,Z')}= \frac{ \widehat{\pi_{\-{wrc}}}(Z') ({\pi_{\-{wrc}}}(Z) + {\pi_{\-{wrc}}}(Z'))  }{ (\widehat{\pi_{\-{wrc}}}(Z) + \widehat{\pi_{\-{wrc}}}(Z'))\pi_{\-{wrc}}(Z')  } \leq 9\mathrm{e}. &\qedhere
\end{align*}
%By (\ref{equ:wrc_ratio_bound}), it holds that $1\leq A,B\leq 2$. We claim that $\frac{1}{2}\leq\frac{\min\left\{1,Ax\right\}}{\min\left\{1,Bx\right\}}\leq 2$, which follows from elementary algebra, finishing the proof. 
\end{proof}

Now we are ready to prove \Cref{lemma-decay-RC-perturbed}. 

\begin{proof}[Proof of \Cref{lemma-decay-RC-perturbed}]
Fix a test function $f$. Denote by $\+E(f,f)$, $\widehat{\+E}(f,f)$ the Dirichlet form of $P_{\-{wrc}}$ and $\widehat{P_{\-{wrc}}}$ respectively. Denote by $\-{Var}[f]$ and $\widehat{\-{Var}}[f]$ the variance of $f$ with respect to $\pi_{\-{wrc}}$ and $\widehat{\pi_{\-{wrc}}}$ respectively. Then by \Cref{lem:comp_dist_ratio_bound} and \Cref{lem:comp_tran_ratio_bound},  
\begin{align*}
  \frac{\+E(f,f)}{\-{Var}[f]}&=\frac{\displaystyle\sum\limits_{\substack{X,Y\subseteq E\\|X\oplus Y|= 1}}\pi_{\-{wrc}}(X)P_{\-{wrc}}(X,Y)\left(f(X)-f(Y)\right)^2}{\displaystyle\sum\limits_{\substack{X,Y\subseteq E\\|X\oplus Y|= 1}}\pi_{\-{wrc}}(X)\pi_{\-{wrc}}(Y)\left(f(X)-f(Y)\right)^2}\\
  &\geq\frac{\displaystyle\frac{1}{9\-e^2}\sum\limits_{\substack{X,Y\subseteq E\\|X\oplus Y|= 1}}\widehat{\pi_{\-{wrc}}}(X)\widehat{P_{\-{wrc}}}(X,Y)\left(f(X)-f(Y)\right)^2}{\displaystyle81\sum\limits_{\substack{X,Y\subseteq E\\|X\oplus Y|= 1}}\widehat{\pi_{\-{wrc}}}(X)\widehat{\pi_{\-{wrc}}}(Y)\left(f(X)-f(Y)\right)^2}>\frac{1}{10^4}\frac{\widehat{\+E}(f,f)}{\widehat{\-{Var}}[f]}. 
\end{align*}
Therefore, $\Gap(P_{\-{wrc}})\geq\frac{1}{10^4}\Gap(\widehat{P_{\-{wrc}}})\geq\frac{\min\left\{p_{\min},1-2p_{\max}\right\}}{10^4n^4m^2}$. 
\Cref{lemma-decay-RC-perturbed} follows from \Cref{theorem-gap-decay}.
\end{proof}

\section{Rapid mixing of Swendsen-Wang dynamics}
\label{sec:RC-SW}
Having analysed the edge-flipping dynamics,
now we turn to relating it with the Swendsen-Wang dynamics.
From this point on, we no longer need the grand model.
We first reiterate the settings for clarity.
Let $G=(V,E)$ be a graph.
We consider the Ising model on $G$ with parameters $\*\lambda = (\lambda_v)_{v \in V}$ and $\*\beta = (\beta_e)_{e \in E}$, 
where $0 < \lambda_v \leq 1$ for all $v \in V$ and $\beta_e > 1$ for all $e \in E$,
as well as the weighted random cluster model on $G$ with parameters $\*p = (p_e)_{e \in E}$ and $\*\lambda = (\lambda_v)_{v \in V}$, where $p_e = 1 - \frac{1}{\beta_e}$ for all $e \in E$. 
Let $\pi_{\-{Ising}}$ over $\Omega_{\+I} = \{0,1\}^V$ denote the Gibbs distribution of the Ising model,
and $\pi_{\-{wrc}}$ over $\Omega_{\+R} = \{0,1\}^E$ denote the distribution of the weighted random cluster model.
We remark that we view $\pi_{\-{wrc}}$ as a distribution over $\{0,1\}^E$ instead of $2^E$.

Let $P_{\-{SW}}^{\-{wrc}} = P_{\+R \to \+I}P_{\+I \to \+R}$ denote the transition matrix of the Swendsen-Wang dynamics for weighted random cluster models as defined in~\Cref{sec:SW},
and $P_{\-{GlauberRC}}$ denote the transition matrix of the Glauber dynamics for   weighted random cluster models. %as defined in~\Cref{section-Metropolis}.
In this section, we compare the Swendsen-Wang dynamics with the Glauber dynamics.
Ullrich \cite{Ull14} showed the following result about the variance decay (spectral gap) of the Swendsen-Wang dynamics.
\begin{lemma}[\text{\cite[Remark 2 and Theorem 5]{Ull14}}]\label{lemma-var-compare}
Suppose $0 < \lambda_v \leq 1$ for all $v \in V$. It holds that
\begin{align*}
	\Gap(P_{\-{SW}}^{\-{wrc}}) \geq \frac{\Gap\tp{P_{\-{GlauberRC}}}}{4}.
\end{align*}	
\end{lemma}
The above result is proved in \cite{Ull14} in the case where $p_e = p \in (0,1)$ for all $e \in E$ and $\lambda_v = 1$ for all $v \in V$.\footnote{In \cite{Ull14}, Ullrich proved this for general random cluster models with an arbitrary $q\ge 1$, but when $q\neq2$ that model cannot be easily translated to the notation we use.}
The model we consider allows that each $e$ has different $p_e \in (0,1)$ and each $v$ has different $\lambda_v \in (0,1]$. 
\Cref{lemma-var-compare} can be proved following the same lines of \cite{Ull14}. 

\Cref{lemma-var-compare} only compares the decay rate of the variance.
The main technical result in this section is the following comparison lemma on the decay rate of the relative entropy.
\begin{lemma}\label{lemma-sw-ent-decay}
Suppose $0 < \lambda_v \leq 1$ for all $v \in V$.
Let $0 < \delta < 1 $. For any distribution $\nu$ over $\Omega_{\+R}$, if 
\begin{align*}
	\KL{\nu P_{\-{GlauberRC}}^\downarrow}{\pi_{\-{wrc}} P_{\-{GlauberRC}}^\downarrow } \leq (1-\delta) \KL{\nu}{\pi_{\-{wrc}}},
\end{align*}
then it holds that
\begin{align*}
	\KL{\nu P_{\-{SW}}^{\-{wrc}} }{\pi_{\-{wrc}} P_{\-{SW}}^{\-{wrc}} } \leq \tp{1-\frac{\delta}{4}} \KL{\nu}{\pi_{\-{wrc}}}.
\end{align*}
\end{lemma}

We are now ready to prove the main results in \Cref{theorem-sw-main} and \Cref{theorem-nlogn}.
\begin{proof}[Proofs of \Cref{theorem-sw-main} and \Cref{theorem-nlogn}]
  Let $\pi_{\-{wrc}, \min} = \min_{S \subseteq E}\pi_{\-{wrc}}(S)$ denote the minimum probability in $\pi_{\-{wrc}}$. It is straightforward to verify that $\pi_{\-{wrc},\min} \geq \{p_{\min},1-p_{\max}\}^m/2^{m+n}$. By  the data processing inequality, \Cref{proposition-adjoint} and \Cref{theorem-gap-decay}, we have
\begin{align*}
	\chisq{\nu P_{\-{SW}}^{\-{wrc}} }{ \pi_{\-{wrc}} } &= \chisq{\nu P_{\-{SW}}^{\-{wrc}} }{ \pi_{\-{wrc}} P_{\-{SW}}^{\-{wrc}} }  \leq \chisq{\nu P_{\+R \to \+I} }{ \pi_{\-{wrc}} P_{\+R \to \+I}}\\
	&\leq \tp{1 - \Gap(P_{\-{SW}}^{\-{wrc}})}\chisq{\nu }{ \pi_{\-{wrc}} }.
\end{align*}
The $\chi^2$-divergence between the initial distribution and the stationary distribution is at most $\frac{1}{\pi_{\-{wrc}, \min}}$. 
%By \Cref{lemma-decay-RC} and \Cref{lemma-decay-RC-perturbed}, we know that after each transition step of Glauber dynamics, the $\chi^2$-divergence and KL-divergence between the current distribution of the stationary distribution decays by factors specified earlier. 
%and the KL-divergence is at most $\log \frac{1}{\pi_{\-{wrc}, \min}}$. 
By \Cref{thm:RC}, \Cref{lemma-var-compare} and~\eqref{eqn:chisq}, let $C_1$ be the constant in \Cref{thm:RC}, we have
\begin{align*}
T_{\-{mix}}(P_{\-{SW}}^{\-{wrc}},\epsilon) &\leq 40 \cdot C_1(p_{\min},p_{\max})  \cdot \min\set{n^4, \tp{\frac{1}{1-\lambda_{\max}}}^4} \cdot m^2 \cdot \tp{\log \frac{1}{\epsilon} + m }.
\end{align*}
By \eqref{eq-mixing-sw}, the mixing time of Swendsen-Wang dynamics on Ising model satisfies 
\begin{align*}
T_{\-{mix}}(P_{\-{SW}}^{\-{Ising}},\epsilon) &\leq C'_1(\beta_{\min},\beta_{\max})  \cdot \min\set{n^4, \tp{\frac{1}{1-\lambda_{\max}}}^4} \cdot m^2 \cdot \tp{\log \frac{1}{\epsilon} + m },
\end{align*}
where $p_{\min } = 1 - \frac{1}{\beta_{\min}}$, $p_{\max  } = 1 - \frac{1}{\beta_{\max}}$, and thus
\begin{align}\label{eq-constant-1}
C'_1(\beta_{\min},\beta_{\max}) &= O\left(\frac{1}{\min\{p_{\min }, 1 - p_{\max }\}} \log \frac{1}{\min\{p_{\min }, 1 - p_{\max }\}}\right)\notag\\
&= O\tp{\tp{\frac{\beta_{\min}}{1-\beta_{\min}} + \beta_{\max } } \log \tp{\frac{\beta_{\min}}{1-\beta_{\min}} + \beta_{\max } }}. 
\end{align}
This proves \Cref{theorem-sw-main}.

For the decay of the relative entropy, the initial KL-divergence is at most $\log \frac{1}{\pi_{\-{wrc}, \min}}$. Combining  \Cref{lemma-sw-ent-decay}, \Cref{lemma-decay-RC}, and~\eqref{eqn:Pinsker}, we have
\begin{align*}
T_{\-{mix}}(P_{\-{SW}}^{\-{wrc}},\epsilon) &\leq 4C_2(\Delta, \delta,p_{\min }, p_{\max } ) \cdot n \tp{\log n + \log \frac{1}{\epsilon} },
\end{align*}
where $C_2$ is the constant $C$ in \Cref{lemma-sw-ent-decay}. By \eqref{eq-mixing-sw}, the mixing time of Swendsen-Wang dynamics on Ising model satisfies 
\begin{align*}
T_{\-{mix}}(P_{\-{SW}}^{\-{Ising}},\epsilon) &\leq C_2'(\Delta, \delta,\beta_{\min }, \beta_{\max } ) \cdot n \tp{\log n + \log \frac{1}{\epsilon} }
\end{align*}
where
\begin{align}\label{eq-constant-2}
C_2'(\Delta, \delta,\beta_{\min }, \beta_{\max } ) &=  \left(\frac{\Delta}{\delta^2 \min\{p_{\min},1-p_{\max}\}}\right)^{O\left(\frac{\Delta^2}{\delta^4 \min\{p_{\min},1-p_{\max}\}}\right)}\notag\\
&= \tp{\frac{\Delta}{\delta^2}\tp{\frac{\beta_{\min}}{1-\beta_{\min}} + \beta_{\max } } }^{O\tp{\frac{\Delta^2}{\delta^4}\tp{\frac{\beta_{\min}}{1-\beta_{\min}} + \beta_{\max } } }}. %\qedhere
\end{align}
This proves \Cref{theorem-nlogn}.
\end{proof}

The rest of this section is dedicated to the proof of \Cref{lemma-sw-ent-decay}.

\subsection{FKES distribution and single-bond dynamics}
To compare the Swendsen-Wang dynamics to the Glauber dynamics, we first introduce the FKES (Fortuin-Kasteleyn-Edwards-Sokal) distribution~\cite{FK72,ES88} $\pi_{\-{FKES}}$ over $\Omega_{\+I} \times \Omega_{\+R}$,
which couples the Ising distribution $\pi_{\-{Ising}}$  and the random cluster distribution $\pi_{\-{wrc}}$:
\begin{align}\label{eq-def-FKES}
\forall \sigma \in \Omega_{\+I}, \tau \in \Omega_{\+R},\quad \pi_{\-{FKES}}(\sigma \tau) \defeq \pi_{\-{Ising}}(\sigma)P_{\+I \to \+R}(\sigma, \tau) \overset{(\star)}{=} \pi_{\-{wrc}}(\tau)P_{\+R \to \+I}(\tau, \sigma),
\end{align}
where $\Omega_{\+I} = \{0,1\}^V$, $\Omega_{\+R}=\{0,1\}^E$, $P_{\+I \to \+R}$ and $P_{\+R \to \+I}$ are defined in~\eqref{eq-def-P-I-R} and ~\eqref{eq-def-P-R-I} respectively. The equation $(\star)$ holds due to \Cref{proposition-adjoint}. 
We use $\Omega_{\-{FKES}} \subseteq \Omega_{\+I} \times \Omega_{\+R}$ to denote the support of the distribution $\pi_{\-{FKES}}$.
The above equation shows that 
\begin{itemize}
	\item the marginal distribution projected from $\pi_{\-{FKES}}$ to $\Omega_{\+I}$ is $\pi_{\-{Ising}}$;
	\item the marginal distribution projected from $\pi_{\-{FKES}}$ to $\Omega_{\+R}$ is $\pi_{\-{wrc}}$;
	\item conditional on $\sigma \in \Omega_{\+I}$, the marginal distribution projected from $\pi_{\-{FKES}}$ to $\Omega_{\+R}$ is $\P_{\+I \to \+R}(\sigma,\cdot)$;
	\item conditional on $\tau \in \Omega_{\+R}$, the marginal distribution projected from $\pi_{\-{FKES}}$ to $\Omega_{\+I}$ is $\P_{\+R \to \+I}(\tau,\cdot)$.
\end{itemize}

Define the following stochastic matrix from the weighted random cluster model to the FKES model
\begin{align*}
\forall \tau_1 \in  \Omega_{\+R} , \sigma\tau_2 \in 	\Omega_{\-{FKES}}, \quad P_{\+R \to \-{FKES}}(\tau_1, \sigma\tau_2) = P_{\+R \to \+I}(\tau_1,\sigma) \cdot \mathbb{I}[\tau_1 = \tau_2],
\end{align*}
The operator $P_{\+R \to \-{FKES}}$ maps from $L_2(\pi_{\-{FKES}})$ to $L_2(\pi_{\-{wrc}})$, where $L_2(\pi)$ is the vector space with the inner product $\inner{\cdot}{\cdot}_{\pi}$. The adjoint operator $P_{\-{FKES} \to \+R}$ is defined by 
\begin{align*}
\forall \sigma\tau_1 \in \Omega_{\-{FKES}}, \tau_2 \in \Omega_{\+R},\quad  P_{\-{FKES} \to \+R}(\sigma\tau_1,\tau_2) = \mathbb{I}[\tau_1 = \tau_2].	
\end{align*}
For any $f \in L_2(\pi_{\-{FKES}})$ and $g \in L_2(\pi_{\-{wrc}})$, it holds that $\inner{P_{\+R \to \-{FKES}}f}{g}_{\pi_{\-{wrc}}} = \inner{f}{P_{\-{FKES} \to \+R}g}_{\pi_{\-{FKES}}}$.

Next, we define the edge down-walk  on the joint distribution. 
Fix an edge $e \in E$.
Given $\sigma\tau \in \Omega_{\-{FKES}}$, let $P^{\downarrow}_e$ denote the edge down-walk that drops the value on edge $e$. 
Formally, $P^{\downarrow}_e$ is defined on any $\sigma\tau \in \Omega_{\-{FKES}}$ and any $\sigma' \tau' \in \Omega_{\-{FKES}}^e$, 
\begin{align*}
  P^{\downarrow}_e(\sigma\tau, \sigma'\tau') =  \mathbb{I}[\sigma=\sigma' \land \tau' = \tau_{E-e}],
\end{align*}
where we use $E - e$ to denote $E \setminus \{e\}$.
Let $\pi_{\-{FKES}}^e = \pi_{\-{FKES}}P^{\downarrow}_e$.
Let $\Omega_{\-{FKES}}^{e}$ denote the support of  $\pi_{\-{FKES}}P^{\downarrow}_e$.
Suppose $e = \{u,v\}$.
We then define the edge up-walk $P^{\uparrow}_e$, for all $\sigma'\tau' \in \Omega_{\-{FKES}}^{e}$ and $\sigma\tau \in \Omega_{\-{FKES}}$, 
\begin{align*}
  P^{\uparrow}_e(\sigma\tau,\sigma'\tau') = \mathbb{I}[\sigma=\sigma'\land \tau_{E-e}=\tau'] \times \begin{cases}
	p_e &\text{if } \tau_e = 1 \text{ and } \sigma(u) = \sigma(v);\\
	1-p_e &\text{if }\tau_e = 0 \text{ and } \sigma(u) = \sigma(v);\\
	0 &\text{if } \tau_e = 1 \text{ and } \sigma(u) \neq \sigma(v);\\
    1 &\text{if }\tau_e = 0 \text{ and } \sigma(u) \neq \sigma(v).\\
\end{cases}
\end{align*}
For any $f \in L_2(\pi_{\-{FKES}})$ and $g \in L_2(\pi_{\-{FKES}}^e)$, it holds that $\inner{P^{\uparrow}_ef}{g}_{\pi_{\-{FKES}}^e} = \inner{f}{P_{e}^{\downarrow}g}_{\pi_{\-{FKES}}}$.

Since in each transition step of $P^{\uparrow}_e$, $\sigma'\tau'_{E-e} = \sigma\tau_{E-e}$ and  the distribution of $\tau'_e$ depends only on $\sigma_u$ and $\sigma_v$, the following observation is straightforward to verify.
\begin{observation}\label{observation-pe}
For any $e,f \in E$, it holds that
\begin{itemize}
	\item $(P^{\downarrow}_e P^{\uparrow}_e) (P^{\downarrow}_f P^{\uparrow}_f) =  (P^{\downarrow}_f P^{\uparrow}_f)(P^{\downarrow}_e P^{\uparrow}_e)$.
	\item $P^{\downarrow}_e P^{\uparrow}_e = (P^{\downarrow}_e P^{\uparrow}_e)^2$.
\end{itemize}
\end{observation}

The single bound dynamics $P_{\-{SB}}: \Omega_{\+R} \times \Omega_{\+R} \to \mathbb{R}_{\geq 0}$ is defined as follows
\begin{align*}
P_{\-{SB}} = P_{\+R \to \-{FKES}}\tp{\frac{1}{m}\sum_{e \in E}P^{\downarrow }_e P^{\uparrow }_e  }P_{\-{FKES} \to \+R}.	
\end{align*}
Intuitively, given any $\tau \in \Omega_{\+R}$, $P_{\-{SB}}$ first transforms $\tau $ into a joint configuration $\sigma\tau \in \Omega_{\-{FKES}}$; samples an edge $e \in E$ uniformly at random; updates $\tau_e$ conditional on $\sigma$; drops $\sigma$ and keeps the random cluster configuration $\tau$. 
Similarly, we can decompose the single bound dynamics as $P_{\-{SB}} =P_{\-{SB}}^{\downarrow} P_{\-{SB}}^{\uparrow}$:
\begin{align}\label{eq-sb-down-walk}
  P_{\-{SB}}^{\downarrow} = P_{\+R \to \-{FKES}}\tp{\frac{1}{m}\sum_{e \in E}P^{\downarrow }_e} \text{ and }	P_{\-{SB}}^{\uparrow} = P_{E}^{\uparrow} P_{\-{FKES} \to \+R},
\end{align}
where for convenience, we treat $(\frac{1}{m}\sum_{e \in E}P^{\downarrow }_e)$ as a matrix defined on $\Omega_{\-{FKES}} \times (\cup_{e \in E}\Omega_{\-{FKES}}^e)$ 
and $P_E^{\uparrow}: (\cup_{e \in E}\Omega_{\-{FKES}}^e) \times \Omega_{\-{FKES}} \to \mathbb{R}_{\geq 0}$ is defined by $P_E^{\uparrow}(x,y) = P_e^{\uparrow}(x,y)$ 
where $x\in \Omega_{\-{FKES}}^e$ for some $e \in E$ and $y \in \Omega_{\-{FKES}}$.
Note that once $x$ is given, $e$ is uniquely determined, and $P_E^{\uparrow}$ agrees with $P_e^{\uparrow}$.
It is straightforward to check $(\frac{1}{m}\sum_{e \in E}P^{\downarrow}_e)$ and $P_{E}^{\uparrow}$ is a pair of adjoint operators.

%the Swendsen-Wang dynamics $P^{\-{wrc}}_{SW}$ can be written as
%\begin{align*}
%P^{\-{wrc}}_{SW} = P_{\+R \to \-{FKES}}\tp{ \prod_{e \in E}P^{\downarrow }_e P^{\uparrow }_e }P_{\-{FKES} \to \+R}.	
%\end{align*}

\begin{lemma}\label{lemma-sb-ent-decay}
Suppose $0 < \lambda_v \leq 1$ for all $v \in V$.
Let $0 < \delta < 1 $. For any distribution $\nu$ over $\Omega_{\+R}$, if 
\begin{align*}
	\KL{\nu P_{\-{GlauberRC}}^\downarrow}{\pi_{\-{wrc}} P_{\-{GlauberRC}}^\downarrow } \leq (1-\delta) \KL{\nu}{\pi_{\-{wrc}}},
\end{align*}
then it holds that
\begin{align}\label{eq-lemma-goal}
	\KL{\nu P_{\-{SB}}^{\downarrow}  }{\pi_{\-{wrc}} P_{\-{SB}}^{\downarrow}  } \leq \tp{1-\frac{\delta}{4}} \KL{\nu}{\pi_{\-{wrc}}}.
\end{align}
\end{lemma}
The proof of \Cref{lemma-sb-ent-decay} is deferred to~\Cref{section-gd-sb}.
We are now ready to prove \Cref{lemma-sw-ent-decay}.

\begin{proof}[Proof of \Cref{lemma-sw-ent-decay}]
By \Cref{observation-pe}, the Swendsen-Wang dynamics $P^{\-{wrc}}_{SW}$ can be written as
\begin{align*}
P^{\-{wrc}}_{\-{SW}} &= P_{\+R \to \-{FKES}}\tp{ \prod_{e \in E}P^{\downarrow }_e P^{\uparrow }_e }P_{\-{FKES} \to \+R} = P_{\+R \to \-{FKES}}\tp{\frac{1}{m}\sum_{e \in E}P^{\downarrow }_eP^{\uparrow }_e}\tp{ \prod_{e \in E}P^{\downarrow }_e P^{\uparrow }_e }P_{\-{FKES} \to \+R}\\
&= P_{\+R \to \-{FKES}}\tp{\frac{1}{m}\sum_{e \in E}P^{\downarrow }_e}P^{\uparrow}_E \tp{ \prod_{e \in E}P^{\downarrow }_eP^{\uparrow }_e} P_{\-{FKES} \to \+R} = P_{\-{SB}}^{\downarrow} P^{\uparrow}_E \tp{ \prod_{e \in E}P^{\downarrow }_eP^{\uparrow }_e} P_{\-{FKES} \to \+R}.
\end{align*}	
Hence, by the data processing inequality, we have
\begin{align*}
\KL{\nu P_{\-{SW}}^{\-{wrc}} }{\pi_{\-{wrc}} P_{\-{SW}}^{\-{wrc}} } \leq 	\KL{\nu P_{\-{SB}}^{\downarrow} }{\pi_{\-{wrc}} P_{\-{SB}}^{\downarrow} } \leq \tp{1-\frac{\delta}{4}} \KL{\nu}{\pi_{\-{wrc}}},
\end{align*}
where the last inequality holds due to \Cref{lemma-sb-ent-decay}.
\end{proof}

\begin{remark}[a simple proof of the main result in~\cite{Ull14}]\label{remark-simple-ull}
If we replace KL-divergence in the above proof with $\chi^2$-divergence.
The same proof shows that for any distribution $\nu$,
\begin{align*}
\chisq{\nu P_{\-{SW}}^{\-{wrc}} }{\pi_{\-{wrc}} P_{\-{SW}}^{\-{wrc}} } \leq 	\chisq{\nu P_{\-{SB}}^{\downarrow} }{\pi_{\-{wrc}} P_{\-{SB}}^{\downarrow} }.	
\end{align*} 
By \Cref{theorem-gap-decay}, we have the following result
\begin{align*}
	\Gap((P_{\-{SW}}^{\-{wrc}})^2) \geq \Gap(P_{\-{SB}}) \quad \implies \quad \Gap(P_{\-{SW}}^{\-{wrc}}) \geq \frac{\Gap(P_{\-{SB}})}{2},
\end{align*}
which recovers the main result in~\cite{Ull14} (losing a factor of $2$).
A more careful application of the data processing inequality in the above proof gives a stronger result $\chisq{\nu P_{\+R \to \+I} }{\pi_{\-{wrc}} P_{\+R \to \+I} } \leq 	\chisq{\nu P_{\-{SB}}^{\downarrow} }{\pi_{\-{wrc}} P_{\-{SB}}^{\downarrow} }$, which gives a better bound $\Gap(P_{\-{SW}}^{\-{wrc}}) \geq \Gap(P_{\-{SB}})$, matching~\cite{Ull14}.
\end{remark}

\subsection{Comparing Glauber dynamics to single-bond dynamics}\label{section-gd-sb}
We first introduce some notations.
Let $\mu$ be a distribution with support $\Omega \subseteq Q^V$.
 
For any $S \subseteq V$, we use $\mu_S$ to denote the marginal distribution on $S$ induced by $\mu$. 
Let $\Omega(\mu_S)$ denote the support of $\mu_S$.
Given any $x_S \in \Omega(\mu_S)$, we use $\mu^{x_S}$ to denote the distribution over $\Omega$ obtained from $\mu$ conditional on $x_S$. 
Formally, for any $y \in \Omega$, $\mu^{x_S}(y) = \mathbb{I}[y_S= x_S]\mu(y)/\mu_S(x_S)$,
where $y_S$ is the restriction of $y$ on $S$.
For any $\Lambda \subseteq V$, we use $\mu^{x_S}_{\Lambda}$ to denote the marginal distribution on $\Lambda$ induced by $\mu^{x_S}$.
We need the following chain rule of the KL-divergence. 
\begin{lemma}\label{lemma-KL-chain}
For any distribution $\nu$ be a distribution over $\Omega$, any $S \subseteq V$, it holds that
\begin{align*}
	\KL{\nu}{\mu} = \KL{\nu_S}{\mu_S} + \Ex_{x_S \sim \nu_S}\KL{\nu^{x_S}}{\mu^{x_S}} = \KL{\nu_S}{\mu_S} + \mu[\Ent{ V-S }{f}],
\end{align*}
where $V-S = V \setminus S$ and $f: \Omega \to \mathbb{R}_{\geq 0}$ is defined by $f(x) = \nu(x)/\mu(x)$ and
\begin{align*}
	\mu[\Ent{V-S}{f}] = \sum_{x_S \in \Omega(\mu_S)}\mu_S(x_S)\Ent{\mu^{x_S}}{f}.
\end{align*}
\end{lemma}
\begin{proof}
The first equation $\KL{\nu}{\mu} = \KL{\nu_S}{\mu_S} + \Ex_{x_S \sim \nu_S}\KL{\nu^{x_S}}{\mu^{x_S}}$ follows directly from the standard chain rule of KL-divergence. 
To prove the second equation, for any $x_S \in \Omega(\nu_S)$, define
\begin{align*}
  \forall y\in \Omega,\quad
  g^{x_S}(y) \defeq 
  \begin{cases}
    \frac{\nu^{x_S}(y)}{\mu^{x_S}(y)} = \frac{\nu(y) \mu_S(x_S) }{\mu(y) \nu_S(x_S)}  = \frac{\mu_S(x_S) }{\nu_S(x_S)}f(y) & \text{if $y_S=x_S$;}\\
    0 & \text{otherwise.}
  \end{cases}
\end{align*}
Since $\Omega(\nu_S) \subseteq \Omega(\mu_S)$, we have
\begin{align*}
  \Ex_{x_S \sim \nu_S}\KL{\nu^{x_S}}{\mu^{x_S}} 
  & = \sum_{x_S \in \Omega(\nu_S)} \nu(x_S)\Ent{\mu^{x_S}}{g^{x_S}} = \sum_{x_S \in \Omega(\nu_S)}\nu(x_S)\Ent{\mu^{x_S}}{\frac{\mu_S(x_S) }{\nu_S(x_S)}f}\\
  & = \sum_{x_S \in \Omega(\nu_S)}	\mu(x_S)\Ent{\mu^{x_S}}{f}. \tag*{(\text{as } \Ent{\mu^{x_S}}{cf} = c \Ent{\mu^{x_S}}{f})}
%& = \sum_{x_S \in \Omega(\nu_S)}	\mu(x_S)\Ent{\mu^{x_S}}{f} + \sum_{x_S \in \Omega(\mu_S) \setminus \Omega(\nu_S)}	\mu(x_S)\Ent{\mu^{x_S}}{f}
%\mu[\Ent{V-S}{f}].
\end{align*}
Note that for all $\sigma \in \Omega$ such that $\sigma_S \in \Omega(\mu_S) \setminus \Omega(\nu_S)$, 
it holds that $f(\sigma) = \frac{\nu(\sigma)}{\mu(\sigma)} = 0$, implying that  $\Ent{\mu^{\sigma_S}}{f} = 0$.
We have
\begin{align*}
\Ex_{x_S \sim \nu_S}\KL{\nu^{x_S}}{\mu^{x_S}} &= 	\sum_{x_S \in \Omega(\nu_S)}	\mu(x_S)\Ent{\mu^{x_S}}{f} + \sum_{x_S \in \Omega(\mu_S) \setminus \Omega(\nu_S)}	\mu(x_S)\Ent{\mu^{x_S}}{f}\\
&= 	\sum_{x_S \in \Omega(\mu_S)}	\mu(x_S)\Ent{\mu^{x_S}}{f} = \mu[\Ent{V-S}{f}]. \qedhere
\end{align*}
\end{proof}

Now we are ready to prove \Cref{lemma-sb-ent-decay}.
\begin{proof}[Proof of \Cref{lemma-sb-ent-decay}]
For any $e \in E$, let $E-e = E \setminus \{e\}$,
using \Cref{lemma-KL-chain}, it holds that
\begin{align*}
\KL{\nu}{\pi_{\-{wrc}}} = \KL{\nu_{E-e}}{\pi_{\-{wrc},E-e}} + \pi_{\-{wrc}}[\Ent{e}{f}], \quad \text{where } f(\tau) = \frac{\nu(\tau)}{\pi_{\-{wrc}}(\tau)}.
\end{align*}
Averaging over all $e\in E$, we get
\begin{align*}
\frac{1}{m}\sum_{e \in E}\pi_{\-{wrc}}[\Ent{e}{f}] &= \frac{1}{m}\sum_{e \in E} 
\KL{\nu}{\pi_{\-{wrc}}} - \frac{1}{m}\sum_{e \in E }	\KL{\nu_{E-e}}{\pi_{\-{wrc},E-e}}\\ 
&= \KL{\nu}{\pi_{\-{wrc}}} - \KL{\nu P_{\-{GlauberRC}}^\downarrow}{\pi_{\-{wrc}} P_{\-{GlauberRC}}^\downarrow }.
\end{align*}
By the assumption of \Cref{lemma-sb-ent-decay}, we have
\begin{align}\label{eq-condition-decay}
\frac{1}{m}\sum_{e \in E}\pi_{\-{wrc}}[\Ent{e}{f}] \geq \delta \KL{\nu}{\pi_{\-{wrc}}}.	
\end{align}
Next, by~\eqref{eq-sb-down-walk}, we have
\begin{align*}
	\KL{\nu P_{\-{SB}}^{\downarrow}  }{\pi_{\-{wrc}} P_{\-{SB}}^{\downarrow}  } &= \KL{\nu P_{\+R \to \-{FKES}}\tp{\frac{1}{m}\sum_{e \in E}P^{\downarrow }_e} }{\pi_{\-{wrc}} P_{\+R \to \-{FKES}}\tp{\frac{1}{m}\sum_{e \in E}P^{\downarrow }_e}  }\\
	&=\KL{\nu_{\-{joint}}\tp{\frac{1}{m}\sum_{e \in E}P^{\downarrow }_e} }{\pi_{\-{FKES}}\tp{\frac{1}{m}\sum_{e \in E}P^{\downarrow }_e}  },
\end{align*}
where $\nu_{\-{joint}} = \nu P_{\+R \to \-{FKES}}$ so that for any $\sigma\tau \in \Omega_{\-{FKES}}$, $\nu_{\-{joint}}(\sigma\tau) = \nu(\tau)\pi_{\-{FKES},V}^\tau(\sigma)$.
Hence, we have
\begin{align*}
\KL{\nu_{\-{joint}}}{\pi_{\-{FKES}}}=\sum_{\sigma\tau \in \Omega_{\-{FKES}}}\nu_{\-{joint}}(\sigma\tau)\log\frac{\nu(\tau)\pi_{\-{FKES},V}^\tau(\sigma)}{\pi_{\-{wrc}}(\tau)\pi_{\-{FKES},V}^\tau(\sigma)} = 	\KL{\nu  }{\pi_{\-{wrc}}}.
\end{align*}
With these two equations, our goal, \eqref{eq-lemma-goal}, is equivalent to 
\begin{align}\label{eq-g1}
\KL{\nu_{\-{joint}}}{\pi_{\-{FKES}}} -	\KL{\nu_{\-{joint}}\tp{\frac{1}{m}\sum_{e \in E}P^{\downarrow }_e} }{\pi_{\-{FKES}}\tp{\frac{1}{m}\sum_{e \in E}P^{\downarrow }_e}  } \geq \frac{\delta}{4} \KL{\nu_{\-{joint}}}{\pi_{\-{FKES}}}.
\end{align}
Using \Cref{lemma-KL-chain}, for any $e \in E$,  let $V+E-e$ be $V\cup E \setminus \{e\}$, it holds that
\begin{align*}
	\KL{\nu_{\-{joint}}}{\pi_{\-{FKES}}} = \KL{\nu_{\-{joint},V+E-e}}{\pi_{\-{FKES},V+E-e}} + \pi_{\-{FKES}}\left[\Ent{e}{\overline{f}}\right],
\end{align*}
where
\begin{align*}
\overline{f}(\sigma \tau) = \frac{\nu_{\-{joint}}(\sigma\tau)}{\pi_{\-{FKES}}(\sigma\tau)} = \frac{\nu(\tau)\pi_{\-{FKES},V}^\tau(\sigma)}{\pi_{\-{wrc}}(\tau)\pi_{\-{FKES},V}^\tau(\sigma)} = \frac{\nu(\tau)}{\pi_{\-{wrc}}(\tau)} = f(\tau).
\end{align*}
Hence,~\eqref{eq-g1} is equivalent to
\begin{align*}
\frac{1}{m}\sum_{e \in E} \pi_{\-{FKES}}\left[\Ent{e}{\overline{f}}\right] \geq  \frac{\delta}{4} \KL{\nu_{\-{joint}}}{\pi_{\-{FKES}}} = \frac{\delta}{4} \KL{\nu  }{\pi_{\-{wrc}}}.	
\end{align*}

Given~\eqref{eq-condition-decay}, to prove the above inequality, it suffices to show that for any $e \in E$,
\begin{align}\label{eq-g2}
	4\cdot \pi_{\-{FKES}}\left[\Ent{e}{\overline{f}}\right] \geq \pi_{\-{wrc}}[\Ent{e}{f}]. 
\end{align}
We now prove~\eqref{eq-g2}. 
We use $\sigma$ to denote the vertex configuration in $\{0,1\}^V$ and $\tau$ to denote the edge configuration $\tau \in \{0,1\}^E$.
Suppose $e = \{u,v\}$. We use $\tau_{-e}$ to denote a configuration in $\{0,1\}^{E-e}$. 
To ease the notation, we use $\pi_{\-{FKES}}(\sigma\tau_{-e})$ to denote $\pi_{\-{FKES},E-e}(\sigma\tau_{-e})$.
For any $\tau_e \in \{0,1\}$, we use $\tau_{-e}\tau_e$ to denote a full configuration $\tau$ in $\{0,1\}^E$. We have
\begin{align*}
  & \pi_{\-{FKES}}\left[\Ent{e}{\overline{f}}\right] 
  = \sum_{\sigma \tau_{-e}} 	\pi_{\-{FKES}}(\sigma\tau_{-e}) \Ent{\pi_{\-{FKES}}^{\sigma\tau_{-e}}}{\overline{f}}\\
  =\,& \sum_{\sigma \tau_{-e}} 	\pi_{\-{FKES}}(\sigma\tau_{-e}) \sum_{\tau_e \in \{0,1\}} \pi_{\-{FKES},e}^{\sigma\tau_{-e}}(\tau_e)\overline{f}(\sigma\tau_{-e}\tau_e)
  \log \frac{\overline{f}(\sigma\tau_{-e}\tau_e)}{\sum_{\tau_e \in \{0,1\}} \pi_{\-{FKES},e}^{\sigma\tau_{-e}}(\tau_e)\overline{f}(\sigma\tau_{-e}\tau_e)}.
\end{align*}
If $\sigma_u \neq \sigma_v$, then $\pi_{\-{FKES},e}^{\sigma\tau_{-e}}(0)=1$, and in this case
\begin{align*}
  \sum_{\tau_e \in \{0,1\}} \pi_{\-{FKES},e}^{\sigma\tau_{-e}}(\tau_e)\overline{f}(\sigma\tau_{-e}\tau_e)
  \log \frac{\overline{f}(\sigma\tau_{-e}\tau_e)}{\sum_{\tau_e \in \{0,1\}} \pi_{\-{FKES},e}^{\sigma\tau_{-e}}(\tau_e)\overline{f}(\sigma\tau_{-e}\tau_e)} = 0.
\end{align*}
Thus we only need to consider the case where the two endpoints of $e$ get the same spin.
Note that this always happens if $\tau_{-e}\in C_e$,
where $C_e \subseteq \{0,1\}^{E-e}$ is the set of $\tau_{-e}$ such that $u$ and $v$ are connected by edges assigned $1$ in $\tau_{-e}$.
Again, to ease the notation, let $\pi_{\-{wrc}}(\tau_{-e})$ be $\pi_{\-{wrc},E-e}(\tau_{-e})$.
Hence, we have
\begin{align}
\pi_{\-{FKES}}\left[\Ent{e}{\overline{f}}\right] &= \sum_{\tau_{-e} \in C_e}\pi_{\-{wrc}}(\tau_{-e})h(p_e,\tau_{-e}) + \sum_{\tau_{-e} \notin C_e}\pi_{\-{wrc}}(\tau_{-e}) \Pr_{\sigma \sim \pi_{\-{FKES},V}^{\tau_{-e}}}[\sigma_u = \sigma_v]h(p_e,\tau_{-e})\notag\\
&\ge \sum_{\tau_{-e} \in C_e}\pi_{\-{wrc}}(\tau_{-e})h(p_e,\tau_{-e}) + \frac{1}{2}\sum_{\tau_{-e} \notin C_e}\pi_{\-{wrc}}(\tau_{-e})h(p_e,\tau_{-e})\label{eqn:joint-entropy-split},
\end{align}
where 
\begin{align*}
  h(p_e,\tau_{-e})& \defeq p_e f(\tau_{-e},1)\log f(\tau_{-e},1) +(1-p_e) f(\tau_{-e},0)\log f(\tau_{-e},0)\\
  &\quad - (p_e f(\tau_{-e},1)+(1-p_e) f(\tau_{-e},0))\log (p_e f(\tau_{-e},1)+(1-p_e) f(\tau_{-e},0)).
\end{align*}
To see \eqref{eqn:joint-entropy-split}, since all external fields are consistent,
$\Pr_{\sigma \sim \pi_{\-{FKES},V}^{\tau_{-e}}}[\sigma_u = \sigma_v] \ge 1/2$.
This is because we can further condition on $\tau_e$: if $\tau_e=1$, then $\sigma_u=\sigma_v$ with probability $1$,
and if $\tau_e=0$, then $\sigma_u$ and $\sigma_v$ are independent and biased towards the same direction,
in which case they are equal with probability at least $1/2$.
The final probability is a linear combination of the two cases.

Similarly, we can expand the right hand side of \eqref{eq-g2},
\begin{align}
  \pi_{\-{wrc}}[\Ent{e}{f}] &= \sum_{\tau_{-e}}\pi_{\-{wrc}}(\tau_{-e})\Ent{\pi^{\tau_{-e}}_{\-{wrc}}}{f}\notag\\
  & =	\sum_{\tau_{-e}}\pi_{\-{wrc}}(\tau_{-e}) \sum_{\tau_e \in \{0,1\}}\pi^{\tau_{-e}}_{\-{wrc},e}(\tau_e)f(\tau_{-e}\tau_e)\log \frac{ f(\tau_{-e}
  \tau_e) }{\sum_{\tau_e \in \{0,1\}}\pi^{\tau_{-e}}_{\-{wrc},e}(\tau_e)f(\tau_{-e}\tau_e) }\notag\\
  &= \sum_{\tau_{-e} \in C_e}\pi_{\-{wrc}}(\tau_{-e})h(p_e,\tau_{-e}) + \sum_{\tau_{-e} \in C_e}\pi_{\-{wrc}}(\tau_{-e})h\tp{\frac{p_e}{1 - \alpha(\tau_{-e})(p_e - 1)},\tau_{-e}}. \label{eqn:wrc-entropy-split}
\end{align}
In the last step above we use $\frac{p_e}{1 - \alpha(\tau_{-e})(p_e - 1)} = \pi^{\tau_{-e}}_{\-{wrc},e}(1)$ where $\alpha(\tau_{-e})$ is a factor depending on $\tau_{-e}$, derived as follows. 
Suppose $e =\{u,v\}$. 
Consider the random cluster configuration with $e$ set not to be taken, and adding $e$ causes the two connected components $C_1$ and $C_2$ to be merged as one, where $u$ is in $C_1$ and $v$ is in $C_2$. 
Let $X=X(\tau_{-e}) = \prod_{w \in C_1}\lambda_w$ and $Y=Y(\tau_{-e})=\prod_{w \in C_2}\lambda_w$. We have
\begin{align*}
  \pi^{\tau_{-e}}_{\-{wrc},e}(1) = \frac{p_e(1+XY)}{p_e(1+XY) + (1-p_e)(1+X)(1+Y)}	 = \frac{p_e}{1-\frac{X+Y}{1+XY}(p_e-1)},
\end{align*}
which means we can take $\alpha(\tau_{-e})=\frac{X+Y}{1+XY}$. 
Moreover, we have $0 \leq \alpha(\tau_{-e})\leq 1$ since $0 < X \leq 1$ and $0 < Y \leq 1$.

To finish the proof, define the following functions
\begin{align*}
  t(x,p,\alpha)\defeq\frac{g(x,p)}{g(x,p/(1-\alpha(p-1)))}\qquad\text{where}\qquad g(x,p):=px\log x-(px+1-p)\log(px+1-p)
\end{align*}
for $0\leq p\leq 1$ and $0\leq \alpha\leq 1$. 
Define $t(0,p,\alpha)\defeq\lim_{x\downarrow 0}t(x,p,\alpha)$ and $t(1,p,\alpha)\defeq\lim_{x\to 1}t(x,p,\alpha)$. 
It is not hard to verify that $t(x,p,\alpha)$ is continuous with respect to $x$ over $[0,\infty)$ for any fixed $p$ and $\alpha$, 
and $t\left(\frac{f(\tau_{-e},1)}{f(\tau_{-e},0)},p_e,\alpha(\tau_{-e})\right) = \frac{h(p_e,\tau_{-e})}{h\left(\frac{p_e}{1-\alpha(\tau_{-e})(p_e-1)},\tau_{-e}\right)}$.
This function admits the following monotonicity property, whose proof is postponed till \Cref{app:proof-tx-min}.
\begin{lemma}  \label{lem:tx-min}
  For any $0\le p\le 1$ and $0\leq \alpha\leq 1$, $t(x,p,\alpha)$ is monotone decreasing in $x$ over $x\geq 0$.
\end{lemma}
Given this, $t(x,p,\alpha)$ has a lower bound
\[
  t(x,p,\alpha)\geq\lim_{x\to\infty} t(x,p,\alpha)=\frac{(1-\alpha(p-1))\log p}{\log p-\log (1-\alpha(p-1))}=:C(p,\alpha). 
\]
We remark that the constant $C=C(p,\alpha)$ satisfies 
\begin{equation} \label{equ:constant-bound}
  0.5\leq C(p,\alpha)\leq 2. 
\end{equation}
The proof is given in \Cref{app:proof-tx-min}, too. 
Using this fact, we conclude (\ref{eq-g2}) by comparing (\ref{eqn:joint-entropy-split}) with (\ref{eqn:wrc-entropy-split}). 

This finishes the proof of \Cref{lemma-sb-ent-decay}.
\end{proof}

\section{Perfect sampling via coupling from the past} \label{sec:CFTP}

In this section, we give a perfect sampler for the ferromagnetic Ising model with consistent fields. 
We first give a perfect sampler for the weighted random cluster model, then turn it into a perfect sampler for the Ising model. 

\begin{theorem}\label{theorem-algorithm-wrc}
There exists a perfect sampling algorithm such that given any weighted random cluster model on graph $G = (V,E)$ with parameters  $\*p = (p_e)_{e \in E}$ and $\*\lambda = (\lambda_v)_{v \in V}$,
if $0 < p_e < 1$ for all $e \in E$ and $0 < \lambda_v \leq 1$ for all $v \in V$,
the algorithm returns a perfect sample from weighted random cluster models in expected time $C_1(p_{\min},p_{\max})  N^4  m^4 \log n$,
where $N= \min\left\{n, \frac{1}{1 - \lambda_{\max}} \right\}$,  $\lambda_{\max}=\max_{v \in V}\lambda_v$,
$C_1(p_{\min},p_{\max}) = O\left(\frac{1}{\min\{p_{\min }, 1 - p_{\max }\}} \log \frac{1}{\min\{p_{\min }, 1 - p_{\max }\}}\right)$,
$p_{\max} = \max_{e \in E}p_e$ and $p_{\min} = \min_{e \in E}p_e$.

Furthermore, if there exists $\delta > 0$ such that $\lambda_v \leq 1-\delta$ for all $v \in V$, then the algorithm runs in time
$C_2(\Delta, \delta,p_{\min }, p_{\max } ) n \log ^ 2 n$,
where $C_2(\Delta, \delta,p_{\min }, p_{\max } ) = \left(\frac{\Delta}{\delta^2 \min\{p_{\min},1-p_{\max}\}}\right)^{O\left(\frac{\Delta^2}{\delta^4 \min\{p_{\min},1-p_{\max}\}}\right)}$.
\end{theorem}

Note that if $p_e=0$, we can simply remove $e$, and if $p_e=1$, we can contract $e$. Similarly if $\lambda_v=0$, we may pin $v$ to $0$ and absorb it into its neighbours external fields.
Thus for any weighted random cluster model, we can modify it so that it satisfies the condition of \Cref{theorem-algorithm-wrc}.

\subsection{Perfect ferromagnetic Ising sampler}
We now prove \Cref{theorem-perfect-ising}.
We give the perfect ferromagnetic Ising sampler assuming the algorithm in \Cref{theorem-algorithm-wrc}. 
Let $G = (V,E)$ be a graph. Let $\*\beta = (\beta_e)_{e \in E}$ and $\*\lambda = (\lambda_v)_{v \in V}$ be parameters for the Ising model, where $\beta_e > 1$ for all $e \in E$ and $0 < \lambda_v < 1$ for all $v \in V$. 
Let $p_e = 1 - \frac{1}{\beta_e}$ for all $e \in E$.
We first use algorithm in \Cref{theorem-algorithm-wrc} to draw a perfect random sample $\+S \subseteq E$ from the weighted random cluster model with parameters $\*p$ and $\*\lambda$. 
Then we using the Markov chain $\+P_{\+R \to \+I}$ in~\eqref{eq-def-P-R-I} to transform $\+S$ into a random Ising configuration $\sigma \in \{0,1\}^V$.
By~\Cref{proposition-adjoint}, since $\+S \sim \pi_{\-{wrc}}$, $\sigma$ is a perfect sample from the Ising model.
The running time of the transformation step is $O(n+m)$.
Note that 
\begin{align*}
p_{\min} = 1 - \frac{1}{\beta_{\min}}, \quad 1 - p_{\max} = \frac{1}{\beta_{\max}}.
\end{align*}
By~\Cref{theorem-algorithm-wrc}, the total running time is $C_1  N^4  m^4 \log n$ and $C_2 n^2 \log^2 n$ for all $\lambda_v \leq 1- \delta$, where 
\begin{equation}
  \label{eqn:perfect-sampler-constant}
  \begin{split}
	C_1 &= C_1(\beta_{\min }, \beta_{\max}) = O\tp{\tp{\beta_{\max} + \frac{\beta_{\min}}{\beta_{\min}-1}} \log \tp{\beta_{\max} + \frac{\beta_{\min}}{\beta_{\min}-1}}},\\
	C_2 &= C_2(\Delta, \delta,\beta_{\min }, \beta_{\max } ) = \tp{\frac{\Delta}{\delta^2}\tp{\beta_{\max} + \frac{\beta_{\min}}{\beta_{\min}-1}} }^{O\tp{\frac{\Delta^2}{\delta^4}\tp{\beta_{\max} + \frac{\beta_{\min}}{\beta_{\min}-1}} }}.
  \end{split}
\end{equation}

\subsection{CFTP for weighted random cluster models}
We give a perfect sampler for weighted random cluster models based on coupling form the past (CFTP) applied to the Glauber dynamics.
Here is an equivalent definition of the Glauber dynamics. %defined in \eqref{eq-def-edge-flipping-wrc}.
There is a one-to-one correspondence between vectors in $\{0,1\}^E$ and subsets in $2^E$ (i.e.~for any $X \in \{0,1\}^E$, let $S_X = \{e \in E \mid X_e = 1\}$).
We assume that the Markov chain is defined over the state space $\{0,1\}^E$.
The Glauber dynamics starts from an arbitrary subset of edges $X_0 \in \{0,1\}^E$. 
For the $t$-th transition step, the chain does the following:
\begin{itemize}
	\item pick an edge $e_t \in E$ uniformly at random;
	\item sample a real number $r_t \in [0,1]$ uniformly at random; if $r_t < a_t$, let $X_{t} = X_{t-1}^{e \gets 1}$; if $r_t \geq a_t$, let $X_{t} = X_{t-1}^{e \gets 0}$, where $X_{t-1}^{e \gets c}$ satisfies $X_{t-1}^{e \gets c}(E \setminus \{e\}) = X_{t-1}(E \setminus \{e\})$ and $X_{t-1}^{e \gets c}(e) = c$, and
	\begin{align}\label{eq-def-at}
		a_t = a(X_{t-1},e) \defeq \frac{\pi_{\-{wrc}}(X_{t-1}^{e \gets 1})}{\pi_{\-{wrc}}(X_{t-1}^{e \gets 0}) + \pi_{\-{wrc}}(X_{t-1}^{e \gets 1}) }.
	\end{align}
\end{itemize}

The Glauber dynamics for weighted random cluster models admits a \emph{grand monotone coupling}. 
Let $\Omega = \{0,1\}^E$. Let $P:\Omega \times \Omega \to \mathbb{R}_{\geq 0}$ denote the transition matrix of the Glauber dynamics.
We use the function $\varphi(\cdot,\cdot)$ to represent each transition step of edge flipping dynamics.
For any $t$, given the current configuration $X_{t-1} \in \Omega$, the next configuration can be generated by $X_{t} = \varphi(X_{t-1},U_t)$, where $U_t$ is the randomness used in the $t$-th transition step. 
Specifically,
\begin{align*}
	U_t \sim \+D \text{ and }U_t = (e_t,r_t) \in \Omega_{R} =  E \times [0,1],
\end{align*}
where $\+D$ is a distribution such that $e_t$ is a uniform random edge in $E$, $r_t$ is a uniform random real number in $[0,1]$, and they are independent. 
The function $\varphi$ uses the transition rule defined above to map $X_{t-1}$ to a random state $X_t = \varphi(X_{t-1},U_t)$, where the randomness of $X_t$ is determined by the randomness of $U_t \sim \+D$.
The function $\varphi(\cdot,\cdot)$ is called a \emph{grand coupling} of flipping dynamics because
\begin{align*}
	\forall \sigma, \tau \in \Omega, \quad \Pr_{U \sim D}\left[ \varphi(\sigma,U) = \tau \right] = P(\sigma,\tau).
\end{align*}

Define a partial ordering $\preceq$ among all vectors in $\{0,1\}^E$: for any $X,Y \in \{0,1\}^E$,
\begin{align*}
	X \preceq Y \quad \text{if } X(e) \leq Y(e) \text{ for all } e \in E.
\end{align*} 
Let $X^{\min} = \* 0$ be the constant 0 vector and $X^{\max} = \* 1$ be the constant 1 vector, so that $X^{\min} \preceq X \preceq X^{\max}$ for all $X \in \{0,1\}^E$.
The next lemma shows that the grand coupling $\varphi$ is monotone with respect to the partial ordering $\preceq$.

\begin{lemma}\label{lemma-coupling}
Suppose $0\leq p_e < 1$ for all $e \in E$ and $0 < \lambda_v \leq 1$ for all $v \in V$.
The grand coupling $\varphi$ of the Glauber dynamics for weighted random cluster models is monotone, i.e.  for any $\sigma,\tau \in \Omega$ with $\sigma \preceq \tau$, any $U \in \Omega_R$, it holds that $\varphi(\sigma,U) \preceq \varphi(\tau,U)$. 
\end{lemma}

The proof of \Cref{lemma-coupling} is deferred to \Cref{sec-proof-monotone}.
With the monotone grand coupling $\varphi$, we apply CFTP to the Glauber dynamics for weighted random cluster models in \Cref{alg-CFTP}.

\begin{algorithm}[h]
\caption{\textsf{CFTP of the Glauber dynamics for weighted random cluster models}}\label{alg:main}\label{alg-CFTP}
\KwIn{a weighted random cluster model on graph $G=(V,E)$ with parameters $\*\lambda = (\lambda_v)_{v \in V}$ and $\*p = (p_e)_{e \in E}$, where $0 < p_e < 1$ for all $e \in E$ and $0 < \lambda_v \leq 1$ for all $v \in V$.} %the monotone grand coupling $\varphi$ for edge flipping dynamics.} 
\KwOut{a perfect sample $X \sim \pi_{\mathrm{wrc}}$, where $\pi_{\mathrm{wrc}}$ is the distribution over $\{0,1\}^E$ defined by the input weighted random cluster model.}
generate $U_t \sim \+D$ independently for all integers $t \in (-\infty, -1]$\label{line-U_t}\;
$T = 1$\;
\Repeat{
$X^{\min} = X^{\max}$
}{
$X^{\min} = \*0$ and $X^{\max} = \*1$\;
\For{$t = -T$ to $-1$}{
	$X^{\min} \gets \varphi(X^{\min}, U_t)$\;
	$X^{\max} \gets \varphi(X^{\max}, U_t)$\;\tcp{$\varphi$ is the monotone grand coupling in \Cref{lemma-coupling}}
}
$T \gets 2T$
}
\Return{$X^{\min}$\;\label{line-last}}
\end{algorithm}

\begin{remark}
In \Cref{alg-CFTP}, infinitely many $U_t$ are generated in \Cref{line-U_t}.
To implement the algorithm, we can first generate $U_{-1}$, and then generate $(U_t)_{-2T \leq t < -T}$ when updating $T \gets 2T$.
\end{remark}
Let  $T_{\+D}$ be the time cost for generating a random sample from $\+D$.
Let  $T_{\varphi}$ be the time cost for computing the value of the function $\varphi$.
Let $T_{\-{mix}}(\cdot)$ denote the mixing time of the edge flipping dynamics for weighted random cluster models.
By the standard result of the CFTP for monotone systems~\cite{propp1996exact} (also see~\cite[Chapter 25]{levin2017markov}), we have the following proposition about~\Cref{alg-CFTP}.

\begin{proposition}[\text{\cite{propp1996exact}}]
Suppose the input weighted random cluster model satisfies $0\leq p_e < 1$ for all $e \in E$ and $0 < \lambda_v \leq 1$ for all $v \in V$.
\Cref{alg-CFTP} returns a perfect sample for the stationary distribution of edge flipping dynamics for weighted random cluster models, i.e.~the distribution $\pi_{\-{wrc}}$. 
The expected running time of \Cref{alg-CFTP} is $O((T_{\+D} + T_{\varphi})T_{\-{mix}}(\frac{1}{4e}) \log n)$.
\end{proposition}

Now, we are ready to prove \Cref{theorem-algorithm-wrc}.
\begin{proof}[Proof of \Cref{theorem-algorithm-wrc}]
By definitions of $\+D$ and $\varphi$, it is straightforward to verify that $T_{\+D} = O(1)$ and $T_{\varphi} = O(n+m)$. The mixing time can be obtained from \Cref{thm:RC}.
%By  \Cref{lemma-wrc-main}, the mixing of edge flipping dynamics is $O(N^4m^3 \log b)$, where $b = \max_{e \in E}\max\{\frac{1}{p_e},\frac{1}{1-p_e}\}$, $N= \min\{n, \frac{1}{1 - \lambda_{\max}} \}$ and $\lambda_{\max} = \max_{v \in V} \lambda_v$. Hence, the expected running time of \Cref{alg-CFTP} can be bounded by $O(N^4m^4\log n\log b)$.
\end{proof}

\subsection{Proof of monotonicity}\label{sec-proof-monotone}
Here we prove \Cref{lemma-coupling}. Fix $\sigma ,\tau \in \{0,1\}^E$ such that $\sigma \preceq \tau$. Fix $U = (e,r) \in \Omega_{R}$. 
Let $e = \{u,v\}$.
Let $\sigma_{-e}$ and $\tau_{-e}$ denote $\sigma(E \setminus \{e\})$ and $\tau(E \setminus \{e\})$ respectively,
and $G_\sigma$ and $G_\tau$ be the graphs with vertices $V$ and edges in $\sigma_{-e}$ and $\tau_{-e}$ respectively.
Note that $G_\sigma$ is a subgraph of $G_\tau$.
We prove the lemma by considering three cases (1) $u,v$ are connected in both $G_\sigma$ and $G_\tau$ (2) $u,v$ are neither connected in neither $G_\sigma$ nor $G_\tau$ (3) $u,v$ are connected in $G_\tau$ but not in $G_\sigma$.

First suppose $u,v$ are connected in both $G_\sigma$ and $G_\tau$. In this case $a(\sigma,e) = a(\tau,e) = p_e$, where $a(\cdot,\cdot)$ is defined in~\eqref{eq-def-at}. The lemma holds trivially.

Next assume $u,v$ are neither connected in neither $G_\sigma$ nor $G_\tau$. 
Suppose $u,v$ belong to connected components $C_1,C_2$ (or $C_1',C_2'$) in $G_\sigma$ (or $G_\tau$) respectively. Define
\begin{align*}
	x_1^\sigma \defeq \prod_{w \in C_1}\lambda_w,\quad x_2^\sigma \defeq \prod_{w \in C_2}\lambda_w,\quad x_1^\tau \defeq \prod_{w \in C_1'}\lambda_w,\quad x_2^\tau \defeq \prod_{w \in C_2'}\lambda_w.
\end{align*}
We have
\begin{align*}
	a(\sigma,e) &= \frac{p_e(1 + x_1^\sigma x_2^\sigma) }{p_e(1 + x_1^\sigma x_2^\sigma) + (1-p_e) (1 + x_1^\sigma)(1+ x_2^\sigma)},\\
	a(\tau,e) &= \frac{p_e(1 + x_1^\tau x_2^\tau) }{p_e(1 + x_1^\tau x_2^\tau) + (1-p_e) (1 + x_1^\tau)(1+ x_2^\tau)}.
\end{align*}
Since $\lambda_w \leq 1$ for all $w \in V$, $x_1^\sigma \geq x_1^\tau$ and  $x_2^\sigma \geq x_2^\tau$, which implies
\begin{align*}
\frac{(1 + x_1^\sigma)(1+ x_2^\sigma)}{(1 + x_1^\sigma x_2^\sigma)} \geq \frac{(1 + x_1^\tau)(1+ x_2^\tau)}{(1 + x_1^\tau x_2^\tau)}.
\end{align*}
Hence $a(\sigma,e) \leq a(\tau,e)$, which implies the lemma.

Lastly suppose $u,v$ are connected in $G_\tau$ but not in $G_\sigma$. Suppose $u,v$ belong to connected components $C_1,C_2$ in $G_\sigma$. Define $x_1^\sigma$ and $x_2^\sigma$ in the same way.
\begin{align*}
a(\sigma,e) = \frac{p_e(1 + x_1^\sigma x_2^\sigma) }{p_e(1 + x_1^\sigma x_2^\sigma) + (1-p_e) (1 + x_1^\sigma)(1+ x_2^\sigma)}, \quad 	a(\tau,e) = p_e.
\end{align*}
Since $(1 + x_1^\sigma)(1+ x_2^\sigma)\geq 1 + x_1^\sigma x_2^\sigma$, $a(\sigma,e) \leq 	a(\tau,e)$, which implies the lemma.

\ifdoubleblind

\else
\section*{Acknowledgement}
We thank Mary Cryan for bringing this question to our attention and for some preliminary discussion.
This project has received funding from the European Research Council (ERC) under the European Union's Horizon 2020 research and innovation programme (grant agreement No.~947778).
\fi

\bibliographystyle{alpha} 
\bibliography{ferro-ising-field}

\newpage
\appendix

\section{Proof of the equivalence result}\label{app-eq}
\subsection{Equivalence between Ising and weighted random cluster models} Fix a graph $G=(V,E)$. We first show the first equation in  (\ref{equ:three_equivalence}). Observe that we can decompose the Ising model interaction matrix as 
\[
f^{\text{Ising}}_e=\begin{pmatrix}\beta_e&1\\1&\beta_e\end{pmatrix}
=
\begin{pmatrix}1&1\\1&1\end{pmatrix}
+
\begin{pmatrix}\beta_e-1&0\\0&\beta_e-1\end{pmatrix}
=:f_e^{(0)}+f_e^{(1)}.
\]
By definition, $f_e^{(1)}$ forces the two endpoints to take the same spin, while $f_e^{(0)}$ poses no requirements. In this way, we can perform an extra enumeration over all the assignments over the edges $\tau:E\to\{0,1\}$, the decompose the effect of $f^{\text{Ising}}_e$ into $f_e^{(0)}$ and $f_e^{(1)}$. The partition function of Ising model then becomes
\begin{align*}
\sum_{\sigma \in \{0,1\}^V}\wt_{\text{Ising}}(\sigma)&=\sum_{\sigma \in \{0,1\}^V}\prod_{e=(u,v)\in E}f^{\text{Ising}}_e(\sigma(u),\sigma(v))\prod_{u\in V}\lambda_u^{\sigma (u)}\\
&=\sum_{\sigma \in \{0,1\}^V}\prod_{e=(u,v)\in E}\left(\sum_{\tau(e)\in \{0,1\}}f_e^{(\tau(e))}(\sigma(u),\sigma(v))\right)\prod_{u\in V}\lambda_u^{\sigma (u)}\\
%&=\sum_{\sigma \in \{0,1\}^V}\sum_{\tau \in \{0,1\}^E}\prod_{e=(u,v)\in E}f_e^{(\tau(e))}(\sigma(u),\sigma(v))\prod_{u\in V}\lambda_u^{\sigma (u)}\\
&=\sum_{\tau \in \{0,1\}^E}\sum_{\sigma \in \{0,1\}^V}\prod_{e=(u,v)\in E}f_e^{(\tau(e))}(\sigma(u),\sigma(v))\prod_{u\in V}\lambda_u^{\sigma (u)}\tag{$*$}.
\end{align*}
Fix $\tau$. Consider the subgraph $G'=(V,S)$ where $S$ is the set of edges assigned to $1$ under $\tau$. Each connected component $C \subseteq V$ of $G'$ must take the same spin in $\sigma$, otherwise the contribution to the sum is 0.
Let $E_C \subseteq S$ denote all the edges in component $C$. The total weight of the component $C$ is 
$\prod_{e \in E_C}(\beta_e - 1)\left(1+\prod_{u \in C}\lambda_u\right)$.
Combining all components yields
\[
  \sum_{\sigma \in \{0,1\}^V}\prod_{e=(u,v)\in E}f^{(\tau(e))}_e(\sigma(u),\sigma(v))\prod_{u\in V}\lambda_u^{\sigma (u)}=\prod_{e \in S}(\beta_e - 1)\prod_{C\in\kappa(V,S)}\tp{1+\prod_{u \in C}\lambda_u}.
\]
And hence
\begin{align*}
(*)&=\sum_{S\subseteq E}\prod_{e \in S}(\beta_e - 1)\prod_{C\in\kappa(V,S)}\tp{1+\prod_{u \in C}\lambda_u}\\
&=\tp{\prod_{e \in E}\beta_e} \cdot \sum_{S\subseteq E} \prod_{e \in S} \left(1-\frac{1}{\beta_e}\right) \prod_{f \in E \setminus S} \frac{1}{\beta_f} \prod_{C\in\kappa(V,S)}\tp{1+\prod_{u\in C}\lambda_u}=Z_{\text{wrc}}(G;2\*p,\*\lambda)
\end{align*}
by taking $2p_e=1-1/\beta_e$.

\subsection{Equivalence between Ising and subgraph-world}
To apply \Cref{thm:holant}, we express the Ising model $(G=(V,E); \*\beta,\*\lambda)$ as a Holant problem. 
Given an Ising model on graph $G=(V,E)$. We define a bipartite graph $H$ with left part $V_1 = V$ corresponding to vertices in $G$ and right part $V_2=E$ corresponding to edges in $G$.
Two vertices $v \in V_1$ and $e \in V_2$ are adjacent in graph $H$ if $v$ is incident to $e$ in graph $G$.
By definition, each edge $e = (u,v)$ in $G$ is decomposed into two half-edges $(v,e)$ and $(u,e)$ in graph $H$.

For any vertex $v \in V_1$, we force the assignment to its incident half-edges to be equal, and further more, if they are all ones, then we multiply the weight by $\lambda_v$. 
This yields the signature $[1,0,\cdots,0,\lambda_v]=[1,0]^{\otimes d_v}+\lambda_v [0,1]^{\otimes d_v}$ on each vertex $v$, 
where $d_v$ is the degree of $v$ in $G$. 
For any edge $e$ in $G$,
%because the value of each of its two incident half-edges must take the same as the corresponding neighbouring vertex, 
its signature is $[\beta_e,1,\beta_e]$ to model the ferromagnetic Ising interaction. 
Define 
\begin{align*}
	\+F_{\mathrm{Ising}} = \set{\left[1,0\right]^{\otimes d_v}+\lambda_v \left[0,1\right]^{\otimes d_v} \mid v \in V} \text{ and }\+G_{\mathrm{Ising}} = \set{[\beta_e,1,\beta_e] \mid e \in E}.
\end{align*}
It is straightforward to verify
\begin{align}\label{eq-Ising=ho}
\holant(H;\+F_{\mathrm{Ising}}\mid\+G_{\mathrm{Ising}}) = Z_{\text{Ising}}(G;\*\beta,\*\lambda).
\end{align}

%Therefore, the resulting equivalent Holant problem is $\left(H;\left[1,0\right]^{\otimes d}+\lambda \left[0,1\right]^{\otimes d}\mid [\beta,1,\beta]\right)$.
For subgraph-world models, we define a Holant problem on the same bipartite graph $H$.
The signature on each vertex $v$ is defined by $[1,\eta_v,1,\eta_v,\cdots]$, and on each edge $e \in E$, it is defined by $[1-p_e,0,p_e]$. 
Define 
\begin{align*}
  \+F_{\mathrm{sg}} = \set{[1,\eta_v,1,\eta_v,\cdots] \mid v \in V} \text{ and }\+G_{\mathrm{sg}} = \set{[1-p_e,0,p_e] \mid e \in E}.
\end{align*}
It is straightforward to verify
\begin{align}\label{eq-sg=ho}
\holant(H;\+F_{\mathrm{sg}}\mid\+G_{\mathrm{sg}}) = Z_{\text{sg}}(G;\*p,\*\eta).
\end{align}
  %the weight of a subset of edges $S\subset E$ is $\wt_{\text{sg}}(S)=q^{\mathsf{odd}(S)}\left(\frac{\beta-1}{2\beta}\right)^{|S|}\left(\frac{\beta+1}{2\beta}\right)^{|E\setminus S|}$ where $\mathsf{odd}(S)$ is the number of odd-degree vertices in $S$. 
 %The signature on vertices is then $[1,q,1,q,\cdots]$, and for edges, it is $[\frac{\beta+1}{2\beta},0,\frac{\beta-1}{2\beta}]$. The following lemma then establishes the equivalence between the Ising model and the subgraph-world model accordingly. 
%\begin{lemma} \label{lem:ising-subgraph-holant}
%It holds that
%\begin{equation*}
%\holant\left(\Omega;\left[1,0\right]^{\otimes d}+\lambda \left[0,1\right]^{\otimes d}\mid [\beta,1,\beta]\right)
%=
%(1+\lambda)^{|V|}\beta^{|E|}\holant\left(\Omega';[1,q,1,q,\cdots]\mid \left[\frac{\beta+1}{2\beta},0,\frac{\beta-1}{2\beta}\right]\right)
%\end{equation*}
%where $q=\frac{1-\lambda}{1+\lambda}$, and the grid $\Omega$ and $\Omega'$ shares the same graph $G=(V,E)$. And hence the second equality in (\ref{equ:three_equivalence}) holds.
%\end{lemma}

%\begin{proof}
Take $T=\Tmatrix$. Let $p_e = \frac{1}{2}\tp{1 - \frac{1}{\beta_e}}$. It holds that
\[
  \left({\bm T}^{-1}\right)^{\otimes 2}\left(\beta_e,1,1,\beta_e\right)^{\top}=\left(\frac{\beta_e+1}{2},0,0,\frac{\beta_e-1}{2}\right)^{\top}=\beta_e\left[\frac{\beta_e+1}{2\beta_e},0,\frac{\beta_e-1}{2\beta_e}\right] = \beta_e [1-p_e,0,p_e].
\]
Let $\eta_v=\frac{1-\lambda_v}{1+\lambda_v}$. We have
\[
\left(\left(1,0\right)^{\otimes d_v}
+\lambda_v\left(0,1\right)^{\otimes d_v}
\right)
\*T^{\otimes d_v}
=
\left(1,1\right)^{\otimes d_v}
+\lambda_v\left(1,-1\right)^{\otimes d_v}
=
(1+\lambda_v)\left[1,\eta_v,1,\eta_v,\cdots\right].
\]
%by taking $q=\frac{1-\lambda}{1+\lambda}$. 
%\end{proof}
Combining~\Cref{thm:holant},~\eqref{eq-Ising=ho} and~\eqref{eq-sg=ho} with the above, it holds that
\begin{align*}
Z_{\text{Ising}}(G;\*\beta,\*\lambda)	 = \tp{\prod_{v \in V}(1+\lambda_v)}\tp{\prod_{e \in E}\beta_e}  Z_{\text{sg}}(G;\*p,\*\eta).
\end{align*}

\section{Proof of the adjointness}\label{app-adjoint}
\begin{proof}[Proof of \Cref{proposition-adjoint}]
Let $D_{\-{Ising}} = \-{diag}(\pi_{\-{Ising}})$ and $D_{\-{wrc}} = \-{diag}(\pi_{\-{wrc}})$ denote the diagonal matrices induced from vectors $\pi_{\-{Ising}}$ and $\pi_{\-{wrc}}$ respectively. We have
\begin{align*}
\inner{f}{P_{\+I \to \+R}g}_{\pi_{\-{Ising}}} = f^T D_{\-{Ising}}P_{\+I \to \+R}g \quad\text{and}\quad  \inner{P_{\+R \to \+I}f}{g}_{\pi_{\-{wrc}}} = f^T	 P_{\+R \to \+I}^T D_{\-{wrc}}g.
\end{align*}
For any $\sigma \in \{0,1\}^V$ and $S \subseteq E$, we show that
\begin{align*}
\tp{D_{\-{Ising}}P_{\+I \to \+R}}(\sigma,S) = \tp{P_{\+R \to \+I}^T D_{\-{wrc}}}(\sigma, S)
\end{align*}
Recall $M(\sigma) = \{\{u,v\} \in E \mid \sigma_u = \sigma_v\}$. It holds that
\begin{align}\label{eq-app-a-1}
\tp{D_{\-{Ising}}P_{\+I \to \+R}}(\sigma,S) &= \mathbb{I}[S \subseteq M(\sigma)] \cdot \pi_{\-{Ising}}(\sigma) \cdot \prod_{e \in S}\tp{1-\frac{1}{\beta_e}} \prod_{f \in M(\sigma \setminus S)} \frac{1}{\beta_f}\notag\\
 &= \mathbb{I}[S \subseteq M(\sigma)] \cdot \frac{1}{Z_{\-{Ising}}} \cdot \prod_{v \in V} \lambda_v^{\sigma(v)} \prod_{h \in M(\sigma)}\beta_h\prod_{e \in S}\tp{1-\frac{1}{\beta_e}} \prod_{f \in M(\sigma \setminus S)} \frac{1}{\beta_f}\notag\\
 &= \mathbb{I}[S \subseteq M(\sigma)] \cdot \frac{1}{Z_{\-{Ising}}} \cdot \prod_{v \in V} \lambda_v^{\sigma(v)} \prod_{e \in S}(\beta_e - 1).
\end{align}
Recall $\kappa(V,S)$ is the set of all connected components of graph $(V,S)$. It holds that
\begin{align}\label{eq-app-a-2}
&\tp{P_{\+R \to \+I}^T D_{\-{wrc}}}(\sigma, S) = \mathbb{I}[S \subseteq M(\sigma)] \cdot \pi_{\-{wrc}}(S) \cdot \prod_{C \in \kappa(V,S)} 	\frac{\prod_{v \in C}\lambda_v^{\sigma(v)}}{1+\prod_{v \in C}\lambda_v}\notag\\
=\,& \mathbb{I}[S \subseteq M(\sigma)] \cdot \frac{1}{Z_{\-{wrc}}} \cdot \prod_{e \in S}\tp{1-\frac{1}{\beta_e}} \prod_{f \in E \setminus S} \frac{1}{\beta_f }\prod_{C \in \kappa(V,S)}\tp{1 + \prod_{u \in C}\lambda_u} \cdot \prod_{C \in \kappa(V,S)} 	\frac{\prod_{v \in C}\lambda_v^{\sigma(v)}}{1+\prod_{v \in C}\lambda_v}\notag\\
=\,& \mathbb{I}[S \subseteq M(\sigma)] \cdot \frac{1}{Z_{\-{wrc}}} \cdot \prod_{e \in S}\tp{1-\frac{1}{\beta_e}} \prod_{f \in E \setminus S} \frac{1}{\beta_f }\prod_{v \in V}\lambda_v^{\sigma(v)}\notag\\
=\,& \mathbb{I}[S \subseteq M(\sigma)] \cdot  \frac{1}{Z_{\-{wrc}}}\cdot \prod_{h \in E}\frac{1}{\beta_h}\prod_{v \in V}\lambda_v^{\sigma(v)} \prod_{e \in S}\tp{\beta_e - 1} 
\end{align}
By \Cref{theorem-eqv}, we know that
\begin{align*}
	\tp{\prod_{e \in E}\beta_e} Z_{\-{wrc}} = Z_{\-{Ising}}.
\end{align*}
Combining above equation with~\eqref{eq-app-a-1} and~\eqref{eq-app-a-2} proves $\tp{D_{\-{Ising}}P_{\+I \to \+R}}(\sigma,S) = \tp{P_{\+R \to \+I}^T D_{\-{wrc}}}(\sigma, S) $. 
\end{proof}

\section{Proof of analytic lemmata} \label{app:proof-tx-min}

This section of appendix proves Lemma~\ref{lem:tx-min} and (\ref{equ:constant-bound}). 

\begin{proof}[Proof of \Cref{lem:tx-min}]

The goal is to show $\partial t(x,p,\alpha)/\partial x<0$ for all $x\in(0,1)\cup(1,+\infty)$.
The lemma then follows by combining this with continuity. 

A straightforward calculation shows that
\begin{align*}
  \frac{\partial t(x,p,\alpha)}{\partial x}=\frac{-(1-\alpha(1-p))(1-p)p}{\left(xp\log x-((1+\alpha)(1-p)+px)\log \left(1+\frac{p(x-1)}{1+\alpha(1-p)}\right)\right)^2} s(x,p,\alpha)
\end{align*}
where
\begin{align*}
  s(x,p,\alpha)\defeq (1+\alpha)(\log x)\log\left(1+\frac{p(x-1)}{1+\alpha(1-p)}\right)-\left(\log x+\alpha\log\left(1+\frac{p(x-1)}{1+\alpha(1-p)}\right)\right)\log (1+p(x-1)).
\end{align*}
This means $\sgn(\partial t(x,p,\alpha)/\partial x)=-\sgn(s(x,p,\alpha))$, and hence we only need to show $s(x,p,\alpha)>0$ whenever $x\in(0,1)\cup(1,+\infty)$. 

From now on in this section, we use the notation $A\gtrless_{x} B$ to represent that $A>B$ when $x>1$, and $A<B$ when $0<x<1$. In other words, when $x>1$, $\gtrless_{x}$ should be read as $>$, and vice versa. 

We first claim the following inequalities:
\begin{equation} \label{equ:monotone-proof-1}
	(1+\alpha)\log x-\alpha\log(1+p(x-1))\gtrless_{x} 0;
\end{equation}
\begin{equation} \label{equ:monotone-proof-2}
	\log\left(1+\frac{p(x-1)}{1+\alpha(1-p)}\right)\gtrless_{x} 0;
\end{equation}
\begin{equation} \label{equ:monotone-proof-3}
  \log (1+(x-1)p)\gtrless_{x} 0. 
\end{equation}
We focus on $s(x,p,\alpha)$ and postpone the proof of these simple inequalities till the end. 
By collecting terms of $\log\left(1+\frac{p(x-1)}{1+\alpha(1-p)}\right)$, one can find that $s(x,p,\alpha)>0$ if and only if
\begin{align*}
  \left((1+\alpha)\log x-\alpha\log(1+p(x-1))\right) \log\left(1+\frac{p(x-1)}{1+\alpha(1-p)}\right) >(\log x)\log (1+p(x-1)). 
\end{align*}
By using (\ref{equ:monotone-proof-1}), it is equivalent to show that
\begin{equation*}
	\log\left(1+\frac{p(x-1)}{1+\alpha(1-p)}\right) \gtrless_{x} \frac{(\log x)\log (1+p(x-1))}{(1+\alpha)\log x-\alpha\log(1+p(x-1))}, 
\end{equation*}
or equivalently, using (\ref{equ:monotone-proof-1})(\ref{equ:monotone-proof-2})(\ref{equ:monotone-proof-3}), to show that
\begin{equation} \label{equ:monotone-proof-4}
  \frac{1}{\log(1+(x-1)p)}\gtrless_{x}\frac{\alpha}{1+\alpha}\cdot\frac{1}{\log x}+\frac{1}{1+\alpha}\cdot\frac{1}{\log\left(1+\frac{p(x-1)}{1+\alpha(1-p)}\right)}.  
\end{equation}

Note that the following function
\[
  u_{x,p}(y)\defeq\frac{1}{\log\left(1+\frac{p(x-1)}{y}\right)}
\]
reveals the essence of (\ref{equ:monotone-proof-4}) in the way that (\ref{equ:monotone-proof-4}) is equivalent to
\begin{equation} \label{equ:monotone-proof-5}
  u_{x,p}(1)\gtrless_{x}\frac{\alpha}{1+\alpha}\cdot u_{x,p}(p)+\frac{1}{1+\alpha}\cdot u_{x,p}(1-\alpha(p-1)), 
\end{equation}
and note that
\[
  1=\frac{\alpha}{1+\alpha}\cdot p+\frac{1}{1+\alpha}\cdot(1-\alpha(p-1)). 
\]
This means (\ref{equ:monotone-proof-5}) follows if for fixed $x>1$ (resp., $0<x<1$) and $p$, $u_{x,p}(y)$ is a concave (resp., convex) function over $y\in (p,2)\supseteq (p,1-\alpha(p-1))$, which would conclude the proof. We verify this as follows. 

A straightforward calculation shows that
\[
  \frac{\mathrm{d}^2}{\mathrm{d}y^2}u_{x,p}(y)=\frac{(x-1)p}{y(y+(x-1)p)^2\log^3\left(1+\frac{p(x-1)}{y}\right)}\left(2\cdot\frac{p(x-1)}{y}-\left(2+\frac{p(x-1)}{y}\right)\log\left(1+\frac{p(x-1)}{y}\right)\right). 
\]
It is not hard to verify that
\begin{equation} \label{equ:monotone-proof-6}
  \log\left(1+\frac{p(x-1)}{y}\right)\gtrless_{x} 0,
\end{equation}
which we prove later. With a bit more endeavour, we can also show that
\begin{equation} \label{equ:monotone-proof-7}
  -\left(2\cdot\frac{p(x-1)}{y}-\left(2+\frac{p(x-1)}{y}\right)\log\left(1+\frac{p(x-1)}{y}\right)\right)\gtrless_{x} 0,
\end{equation}
whose proof is postponed as well. Concavity/Convexity is then established by combining the expression for the second-order derivative, (\ref{equ:monotone-proof-6}) and (\ref{equ:monotone-proof-7}). 
\end{proof}

\begin{proof}[Proof of (\ref{equ:monotone-proof-1}), (\ref{equ:monotone-proof-2}), (\ref{equ:monotone-proof-3}), (\ref{equ:monotone-proof-6}), and (\ref{equ:monotone-proof-7})]
For (\ref{equ:monotone-proof-1}), because $\log x\gtrless_{x} 0$, we only need to show
\[
  \frac{x}{1+(x-1)p}\gtrless_{x} 1.
\]
Note that $1+(x-1)p$ is positive. The above is hence equivalent to
\[
  (x-1)(1-p)\gtrless_{x} 0,
\]
which is obvious. 

All of (\ref{equ:monotone-proof-2}), (\ref{equ:monotone-proof-3}) and (\ref{equ:monotone-proof-6}), after simple calculation, are equivalent to $(x-1)p\gtrless_{x} 0$, which is obvious, too. 

Finally, we show (\ref{equ:monotone-proof-7}). Let $z:=p(x-1)/y$. LHS is then $r(z):=(2+z)\log(1+z)-2z$. It is not hard to show that $r(z)$ is monotone in $z$ over $z\in(-1,+\infty)$, by observing that $r'(z)=\frac{1}{1+z}-1-\log\frac{1}{1+z}$, which is non-negative as $\log x\le x-1$ for $x>0$. Moreover, $r(0)=0$. Therefore, when $x>1$, we have $z > 0$, and (\ref{equ:monotone-proof-7}) holds. When $0<x<1$, we have $-1 < (x-1)\leq z <0$, and (\ref{equ:monotone-proof-7}) holds too. 
\end{proof}

\begin{proof}[Proof of (\ref{equ:constant-bound})]
For convenient reference, the expression of interest is
\[
  C(p,\alpha):=\frac{(1-\alpha(p-1))\log p}{\log p-\log (1-\alpha(p-1))}. 
\]
Taking derivative with respect to $\alpha$, we get
\[
  \frac{\partial}{\partial \alpha}C(p,\alpha)=\frac{(1-p)\log(p)\left(1+\log\left(\frac{p}{1+\alpha(1-p)}\right)\right)}{\left(\log\left(\frac{p}{1+\alpha(1-p)}\right)\right)^2}. 
\]
A simple calculation shows that
\begin{itemize}
	\item if $p\leq 1/e$, then $C(p,\alpha)$ is increasing with $\alpha$, and hence lies between $C(p,0)=1$ and $C(p,1)=\frac{(2-p)\log p}{\log p-\log (2-p)}$; 
	\item if $1/e<p<2/(1+e)$, then $C(p,\alpha)$ is decreasing within $\alpha\in(0,(ep-1)/(1-p))$ and increasing within $\alpha\in ((ep-1)/(1-p),1)$, and hence lies between $C(p,(ep-1)/(1-p))=-ep\log p\geq 2e\log((1+e)/2)/(1+e)>0.90$ and $\max\{C(p,0),C(p,1)\}$; and
	\item if $p\geq 2/(1+e)$, then $C(p,\alpha)$ is decreasing, and hence lies between $C(p,1)=\frac{(2-p)\log p}{\log p-\log (2-p)}$ and $C(p,0)=1$. 
\end{itemize}
From the case-by-case analysis, it suffices to show that $0.5\leq \frac{(2-p)\log p}{\log p-\log (2-p)}\leq 2$, which is a simple analytic exercise. 
\end{proof}

\end{document}